\documentclass[12pt]{article}
\usepackage{amsmath}
\usepackage{graphicx}
\usepackage{enumerate}
\usepackage{natbib}
\usepackage{url} % not crucial - just used below for the URL 

\usepackage{amsmath}
\usepackage{graphicx,psfrag,epsf}
\usepackage{enumerate}
\usepackage{natbib}
\usepackage{url} % not crucial - just used below for the URL 
\usepackage{amsthm}
\usepackage{booktabs}
\usepackage{framed}  
\usepackage{caption}
\usepackage{pgfplots}
\usepgfplotslibrary{dateplot}
\usetikzlibrary{snakes}
\usepackage{rotating}
\usepackage{float}
\usepackage{amsfonts}
\usepackage{multirow, booktabs}
\usepackage{cleveref}
\usepackage{subcaption}
\usepackage[ruled,vlined]{algorithm2e}
\usepackage[T1]{fontenc}
\usepackage[utf8]{inputenc}
\usepackage{authblk}
\usepackage[multiple]{footmisc}
\usepackage{blindtext,titlefoot}
\usepackage{sectsty}
\usepackage{xcolor}
\usepackage{tikz}
\usetikzlibrary{positioning, shapes.geometric}

\usepackage{tabularx}
\usepackage{ragged2e} % added for better adjust of cells' content
\newcolumntype{L}{>{\RaggedRight}X} % for cells with left aligned content
\usepackage{lipsum} % just for dummy text

\usepackage[ruled,vlined]{algorithm2e}
\usepackage{amsmath, amssymb}
\usepackage{amsfonts, multirow, epsfig, subfig}
\usepackage{graphicx, pdflscape, verbatim, enumerate, colortbl, setspace}
\usepackage{setspace, color,bm}
\usepackage[normalem]{ulem}
\usepackage{cite}
\usepackage{multirow}
\usepackage{booktabs,array}
\usepackage{url}
\usepackage{bbm}
\newcommand{\wh}{\widehat}

\newtheorem{condition}{Condition}
\newtheorem{corollary}{Corollary}
\newtheorem{proposition}{Proposition}
\newtheorem{lemma}{Lemma}
\newtheorem{definition}{Definition}
\newtheorem{theorem}{Theorem}

\newtheorem{remark}{Remark}

\newtheorem{assumption}{Assumption}

\makeatletter
\newcommand*{\rom}[1]{\expandafter\@slowromancap\romannumeral #1@}
\makeatother

\begin{document}

\def\spacingset#1{\renewcommand{\baselinestretch}%
{#1}\small\normalsize} \spacingset{1}

%%%%%%%%%%%%%%%%%%%%%%%%%%%%%%%%%%%%%%%%%%%%%%%%%%%%%%%%%%%%%%%%%%%%%%%%%%%%%%
%TCIMACRO{\TeXButton{Section}{\sectionfont{\bfseries\large\sffamily}}}%
%BeginExpansion
\sectionfont{\bfseries\large\sffamily}%
%EndExpansion
%
\newcommand*\emptycirc[1][1ex]{\tikz\draw (0,0) circle (#1);} 
\newcommand*\halfcirc[1][1ex]{%
  \begin{tikzpicture}
  \draw[fill] (0,0)-- (90:#1) arc (90:270:#1) -- cycle ;
  \draw (0,0) circle (#1);
  \end{tikzpicture}}
\newcommand*\fullcirc[1][1ex]{\tikz\fill (0,0) circle (#1);} 

%TCIMACRO{\TeXButton{Subsection}{\subsectionfont{\bfseries\sffamily\normalsize
%}}}%
%BeginExpansion
\subsectionfont{\bfseries\sffamily\normalsize}%
%EndExpansion
%

%\bibliographystyle{natbib}

\def\spacingset#1{\renewcommand{\baselinestretch}%
{#1}\small\normalsize} \spacingset{1}

\begin{center}
    \Large \bf Sensitivity Analysis for Matched Observational Studies with Continuous Exposures and Binary Outcomes
\end{center}

%Analyzing Randomized Experiments Subject to Outcome Misclassification via Integer Programming

\begin{center}
  \large  $\text{Jeffrey Zhang}^{*, 1}$, $\text{Dylan Small}^{ 1}$ and $\text{Siyu Heng}^{2}$
\end{center}

\begin{center}
   \large \textit{$^{1}$Department of Statistics and Data Science, University of Pennsylvania}
\end{center}
\begin{center}
   \large \textit{$^{2}$Department of Biostatistics, New York University}
\end{center}

\let\thefootnote\relax\footnotetext{$^{*}$Address for Correspondence: Jeffrey Zhang, Department of Statistics and Data Science, University of Pennsylvania (email: jzhang17@wharton.upenn.edu). }

\bigskip

\begin{abstract}

Matching is one of the most widely used study designs for adjusting for measured confounders in observational studies. However, unmeasured confounding may exist and cannot be removed by matching. Therefore, a sensitivity analysis is typically needed to assess a causal conclusion's sensitivity to unmeasured confounding. Sensitivity analysis frameworks for binary exposures have been well-established for various matching designs and are commonly used in various studies. However, unlike the binary exposure case, there still lacks valid and general sensitivity analysis methods for continuous exposures, except in some special cases such as pair matching. To fill this gap in the binary outcome case, we develop a sensitivity analysis framework for general matching designs with continuous exposures and binary outcomes. First, we use probabilistic lattice theory to show our sensitivity analysis approach is finite-population-exact under Fisher's sharp null. Second, we prove a novel design sensitivity formula as a powerful tool for asymptotically evaluating the performance of our sensitivity analysis approach. Third, to allow effect heterogeneity with binary outcomes, we introduce a framework for conducting asymptotically exact inference and sensitivity analysis on generalized attributable effects with binary outcomes via mixed-integer programming. Fourth, for the continuous outcomes case, we show that conducting an asymptotically exact sensitivity analysis in matched observational studies when both the exposures and outcomes are continuous is generally NP-hard, except in some special cases such as pair matching. As a real data application, we apply our new methods to study the effect of early-life lead exposure on juvenile delinquency. We also develop a publicly available \textsf{R} package for implementation of the methods in this work.
\end{abstract}

\noindent%

%\spacingset{1.9} % DON'T change the spacing!
\spacingset{1.73} % DON'T change the spacing!

\section{Introduction}
\label{intro}
\subsection{Motivating example}\label{subsec: motivating example}

Early-life lead (Pb) exposure is associated with many later-life problems, including ``decreased IQ, behavior problems, lower lifetime earnings, and increased criminal activity" (\citealp{lead_crime}). Over the past decades, the United States has implemented several public health interventions to prevent early-life lead poisoning. However, elevated lead exposure continues to harm US children, and the Centers for Disease Control and Prevention (CDC) decreased its recommended reference level for elevated blood lead (Pb) to 3.5 µg/dL in October 2021. Recently, \citet{lead_crime} compiled an anonymized dataset of 13,580 individuals in North Carolina by linking later-life demographic and juvenile delinquency records with early-life measurements of lead exposure and demographic variables. Some individuals in the data had multiple blood lead level measurements, and we followed the strategy proposed in \citet{Gibson2020} to use the first recorded measurement for these individuals (see the ``Database Development" section in \citet{Gibson2020} for the consideration behind this strategy). It should be acknowledged that a one-time early-life blood lead measurement is an imperfect proxy for overall early-life lead exposure, particularly because there is no guarantee that lead exposure is stable and sustained throughout the early life of each individual in the dataset. However, variables such as the age of the individual at measurement, the year of the blood test, and the type of blood draw were collected and adjusted for in the analysis. The authors studied the effect of the early-life blood lead level on two juvenile delinquency outcomes (binary indicators): (i) any delinquency and (ii) any serious delinquency. Here, the exposure is continuous, and the outcome is binary. They used a model-based approach to estimate the effect. While they found statistically significant effects, the data is observational, so the causal conclusions might be biased by potential unmeasured confounding. Therefore, a sensitivity analysis is needed to assess sensitivity of their qualitative conclusions to violations of the no unmeasured confounding assumption. We will re-analyze the data from \citet{lead_crime} through a matched cohort study and corresponding sensitivity analyses. For different hypothesized levels of unmeasured confounding, our sensitivity analyses will consider testing the causal null of no effect and conducting inferences on attributable effects (which allows for effect heterogeneity).

\subsection{Previous literature and our contributions}

Among various methods for adjusting for measured confounders in observational studies, matching is one of the most widely used ones. In matched observational studies, units with different exposure values (such as exposed and unexposed indicators in the binary exposure case, or different exposure doses in the continuous exposure case) are matched for measured confounders to effectively remove or reduce potential bias contributed by measured confounding (\citealp{rosenbaum_obs, Lu2011OptimalNM, imbens2015causal}). However, like many other widely used confounder adjustment methods such as weighting and outcome regression, matching cannot adjust for unmeasured confounders. Therefore, a sensitivity analysis is typically needed in matched observational studies to examine sensitivity of a causal inference output (e.g., $p$-value for testing null hypothesis of no effect or the confidence interval for the effect) to unmeasured confounding (\citealp{rosenbaum_obs, imbens2015causal}). 

Sensitivity analysis has been well-established and routinely reported in matched observational studies with binary exposures (\citealp{rosenbaum1987sensitivity, rosenbaum_obs, imbens2015causal}). However, for matched observational studies with continuous exposures, there is still a lack of valid and general sensitivity analysis frameworks, except in some special cases such as pair matching (\citealp{rosenbaum1989sensitivity, gastwirth1998dual}). Since many matched observational studies consider continuous exposures and the relevant literature has been rapidly growing in the last two decades (e.g., \citealp{Lu_NBP, baiocchi2010building, rosenbaum2020design, fogarty2019biased, bo_continuous, yu2023risk}), there is a substantial demand for a general and powerful sensitivity analysis framework for matching with continuous exposures. In this work, we fill this gap in the binary outcome case; many matched observational studies consider continuous exposures and binary outcomes (e.g., matched case-control studies with continuous exposures). Specifically, for general matched observational studies with continuous exposures, we develop exact (or asymptotically exact) sensitivity analysis methods (in the binary outcome case) for testing the null effect and a generalized attributable effect and also show the computational hardness of conducting exact sensitivity analysis with continuous outcomes; see Table \ref{contribution} for a summary. We also developed a publicly available \textsf{R} package \texttt{doseSens} for implementation of our methods.

\begin{table}[htbp]
\centering
\caption{Classification of sensitivity analysis methods for general matching designs based on the types of exposures and outcomes. }
\begin{tabular}{r|rr}
& Binary Exposure & Continuous Exposure \\
\hline
Binary Outcome & Rosenbaum (1987, \textit{Biometrika}) & Solved in this work  \\
Continuous Outcome & Gastwirth et al. (2000, \textit{JRSSB}) & Hardness results given in this work  \\
%\hline
%$\Gamma = 4.00$ & 0.00 & 0.12 & 0.05 & 0.00 & 0.08 & 0.00 & 0.05 & 0.01 & 0.00 & 0.02\\
%\hline
%$\Gamma = 4.25$ & 0.00 & 0.10 & 0.03 & 0.00 & 0.07 & 0.00 & 0.05 & 0.01 & 0.00 & 0.02\\
\end{tabular}
\label{contribution}

\end{table}

%See \citet{Lu2011OptimalNM} and for a review of such pair-matching designs with continuous exposures. Some matching algorithms have also been proposed for continuous exposures that can produce matched sets beyond matched pairs (i.e., beyond pair matching) and have been shown to outperform pair matching in many aspects (\citealp{bo_continuous}). Moreover, in case-control settings, the number of units in match sets often exceeds size 2. 

\section{Review and Preliminary Results}\label{background}

\subsection{Notations for matched observational studies with continuous exposures and binary outcomes}\label{subsec: notations for perfect matching}

We follow the commonly used notations and framework for matched observational studies with continuous exposures (\citealp{rosenbaum1989sensitivity, gastwirth1998dual, bo_continuous}). Suppose that there are $I$ matched sets, with $n_i$ units in matched set $i$ ($i=1,\dots, I$), giving a total of $N=\sum_{i=1}^I n_i$ units. Let ${x}_{ij}$ denote the measured confounders vector, $Z_{ij}$ the exposure dose, and $R_{ij}$ the observed outcome of subject $j$ in matched set $i$. Let ${Z}=(Z_{11},\dots, Z_{In_{I}})$ be the exposure doses vector and ${R}=(R_{11},\dots, R_{In_{I}})$ the observed outcomes vector. Note that each matched set $i$ has $n_i$ observed exposure doses $Z_{i1}, \dots, Z_{in_{i}}$, potentially all distinct. Let $\widetilde{Z}$ denote the set of all possible exposure doses and $r_{ij}(z)\in \{0,1\}$ the potential outcome of subject $ij$ under the exposure dose $z\in \widetilde{Z}$. It follows that the observed outcome $R_{ij}=r_{ij}(z)$ given $Z_{ij}=z$. Under the randomization inference framework, the potential outcomes $r_{ij}(z)$ are fixed numbers, and the only source of randomness in statistical inference comes from exposure dose assignments (\citealp{rosenbaum_obs, li_ding_finite, bo_continuous}). Let $\mathcal{F}_{0}=\{{x}_{ij},r_{ij}(z),i=1,...,I,j=1,\dots, n_{i}, z \in \widetilde{Z}\}$. Within each matched set $i$, conditional on the given (observed) exposure doses ${z}_{i}=(z_{i1}, \dots, z_{in_{i}})$, there are $n_{i}!$ possible realizations of the random exposure dose vector ${Z}_{i}=(Z_{i1}, \dots, Z_{in_{i}})$ (corresponding to $n_i!$ permutations of the given (observed) dose vector ${z}_{i}$). Let $\mathcal{Z}_{i}=\{{z}_{i\pi_{i}}=(z_{i\pi_{i}(1)}, \dots, z_{i\pi_{i}(n_{i})}) \mid \pi_{i}=(\pi_{i}(1),\dots, \pi_{i}(n_{i})) \in S_{n_i}\}$ denote the set of $n_{i}!$ permutations of ${z}_{i}$, in which $S_{n_i}$ denotes the collection of all $n_{i}!$ permutations of the index set $\{1, \dots, n_{i}\}$. For example, if a matched set has three exposure doses 0.1, 0.3, and 0.8, then there are $3!=6$ possible assignments of these three exposure doses to the three units within this matched set: $\mathcal{Z}_{i}=\{(0.1, 0.3, 0.8), (0.1, 0.8, 0.3), (0.3, 0.1, 0.8), (0.3, 0.8, 0.1), (0.8, 0.1, 0.3), (0.8, 0.3, 0.1)\}$. We then let $\mathcal{Z}=\mathcal{Z}_{1}\times \mathcal{Z}_{2} \times \dots \times \mathcal{Z}_{I}$ denote all possible dose assignments for the whole matched sets, so the cardinality of $|\mathcal{Z}|=n_{1}!\dots n_{I}!$. Conditional on matching on the measured confounders (i.e., ${x}_{ij}={x}_{ij^{\prime}}$ or ${x}_{ij}\approx {x}_{ij^{\prime}}$ for all $1 \leq j,j^{\prime} \leq n_i$) and the no unmeasured confounding assumption, the exposure dose assignments $p_{i\pi_{i}}$ are uniform (or approximately uniform) within each matched set $i$:
\begin{equation}\label{eqn: random assignment assumption}
    \text{$p_{i\pi_{i}}=\text{pr}({Z}_i={z}_{i\pi_{i}}|\mathcal{F}_{0},\mathcal{Z}_{i})=\frac{1}{n_{i}!}$ for all $\pi_{i} \in S_{n_i}$ (i.e., for all ${z}_{i\pi_{i}} \in \mathcal{Z}_{i}$).}
\end{equation}
 Then, randomization inference can be conducted under the uniform dose assignment assumption (\ref{eqn: random assignment assumption}). For example, for testing Fisher's sharp null hypothesis $H_{0}: r_{ij}(z)=r_{ij}(z^{\prime})$ for any $i, j,$ and $z, z^{\prime} \in {Z}_{i}$, given the test statistics $T({Z}, {R})$ and its observed value $t$, the corresponding one-sided finite-population-exact $p$-value is $\text{pr}(T \geq t|\mathcal{F}_{0},\mathcal{Z})=|\{{z} \in \mathcal{Z}: T({Z}={z}, {R}) \geq t\}|/(n_{1}!\dots n_{I}!)$, which can be approximated by Monte Carlo simulation (\citealp{rosenbaum1989sensitivity, gastwirth1998dual, bo_continuous}).

\subsection{The Rosenbaum sensitivity analysis model with continuous exposures}
\label{sens_model}
In an observational study, because unmeasured confounding may exist, the random exposure dose assignments assumption (\ref{eqn: random assignment assumption}) may not hold in practice, even if all the measured confounders were adjusted for. Therefore, a sensitivity analysis is typically needed to assess sensitivity of causal conclusions (e.g., $p$-values under some causal null hypothesis or confidence intervals for some causal estimands) to potential unmeasured confounding. In matched observational studies with continuous exposures, the Rosenbaum exposure dose assignment model (\citealp{rosenbaum1989sensitivity, gastwirth1998dual}) is one of the most widely used sensitivity analysis models, and many other sensitivity analysis models are built on or closely related to it (e.g., \citet{tan_msm, bonvini2022sensitivity}). It models the conditional density of exposure dose $Z$ given the measured confounder ${x}$ and a hypothetical unmeasured confounder $u$ as follows:
\begin{equation}\label{eqn: Rosenbaum exposure dose model}
    f(Z=z|{x},u)=\zeta({x},u) \eta(z,{x})\exp(\gamma z u)\propto \eta(z,{x})\exp(\gamma z u),
\end{equation} 
where $\zeta({x},u)=1/\int \eta(z,{x})\exp(\gamma z u)dz$ is the normalizing constant, the $\eta(z,{x})$ is an arbitrary nuisance function, the $\gamma>0$ (or equivalently, $\Gamma=\exp(\gamma)$) is the sensitivity parameter, and $u \in [0,1]$ was normalized to make $\gamma$ (or $\Gamma$) more interpretable. Specifically, following the arguments in \citet{rosenbaum1989sensitivity} and \citet{gastwirth1998dual}, consider two subjects $n_{1}$ and $n_{2}$ with the same or similar measured confounders (i.e., ${x}_{n_{1}}={x}_{n_{2}}$ or ${x}_{n_{1}}\approx {x}_{n_{2}}$) but possibly different unmeasured confounders $u_{n_{1}}$ and $u_{n_{2}}$. For two arbitrary fixed exposure doses $z$ and $z^{\prime}$ with $z<z^{\prime}$, let $p_{n_{1}}$ (or $p_{n_{2}}$) denote the probability that subject $n_{1}$ (or $n_{2}$) received higher dose $z^{\prime}$ (or symmetrically, lower dose $z$). That is, we have $p_{n_{1}}+p_{n_{2}}=1$ with
\begin{align*}
   p_{n_{1}}=\text{pr}(Z_{n_{1}}=z^{\prime}, Z_{n_{2}}=z \mid {x}_{n_{1}}, {x}_{n_{2}}, u_{n_{1}}, u_{n_{2}}, \min\{Z_{n_{1}}, Z_{n_{2}}\}=z, \max\{Z_{n_{1}}, Z_{n_{2}}\}=z^{\prime}).
\end{align*}
Under the Rosenbaum dose assignment model (\ref{eqn: Rosenbaum exposure dose model}), it is straightforward to show that:
\begin{align*}
   -\gamma(z^{\prime}-z) \leq \text{logit}\ p_{n_{1}} \leq \gamma(z^{\prime}-z).
\end{align*}
That is, the sensitivity parameter $\gamma\geq 0$ (or equivalently, $\Gamma=\exp(\gamma)\geq 1$) quantifies how unmeasured confounding would bias the logit (i.e., log odds) of receiving the higher (or lower) dose after matching on measured confounders, standardized by the difference in exposure dose $z^{\prime}-z$ (\citealp{rosenbaum1989sensitivity, gastwirth1998dual}).  

We here make three further remarks about the Rosenbaum exposure dose model (\ref{eqn: Rosenbaum exposure dose model}). First, when the exposure is binary (i.e., there are only two doses), model (\ref{eqn: Rosenbaum exposure dose model}) reduces to the classic Rosenbaum propensity score model for sensitivity analysis (\citealp{rosenbaum1987sensitivity}), one of the most widely used sensitivity analysis model for matched observational studies with binary exposures. Second, as discussed in \citet{rosenbaum1989sensitivity}, model (\ref{eqn: Rosenbaum exposure dose model}) contains many widely used exposure dose assignment models as special cases. For example, when the exposure is continuous, model (\ref{eqn: Rosenbaum exposure dose model}) contains the partially linear model $Z=g({x})+\gamma u + \epsilon$, where $g({x})$ is an arbitrary function and $\epsilon$ follows a normal distribution. Third, model (\ref{eqn: Rosenbaum exposure dose model}) imposes a monotonic relationship between the unmeasured confounder(s) and receiving a higher dose, which may not always hold. Fortunately, the results in this paper continue to hold if the $\exp(\gamma z u)$ in model (\ref{eqn: Rosenbaum exposure dose model}) is replaced with $\exp\{\gamma \phi(z) u\}$ for some arbitrary known function $\phi$, and all subsequent instances of $z$ can be replaced with $\phi(z)$. For clarity and simplicity of exposition, our presentation will proceed with the original dose $z$ rather than a transformed dose $\phi(z)$. Still, readers should be aware that our results can hold in much broader settings beyond the Rosenbaum model (\ref{eqn: Rosenbaum exposure dose model}) as our results do not require assuming any monotonic relationships between unmeasured confounding and exposure dose.

Let $u_{ij}\in [0,1]$ denote a hypothetical unmeasured confounder of unit $j$ in matched set $i$, and ${u}=(u_{11},\dots, u_{In_{I}})\in [0,1]^{N}$. We define $\mathcal{F}=\{{x}_{ij}, u_{ij}, r_{ij}(z),i=1,\dots,I, j=1,\dots, n_{i}, z \in \widetilde{Z}\}$. Under the Rosenbaum exposure dose model (\ref{eqn: Rosenbaum exposure dose model}), using the Bayes rule, it is straightforward to show that within each matched set $i$, for any exposure doses assignment ${z}_{i\pi_{i}} \in \mathcal{Z}_{i}=\{{z}_{i\pi_{i}} \mid \pi_{i} \in S_{n_i}\}$, we have
\begin{equation}\label{eqn: dose assignment after matching}
    p_{i\pi_{i}}=\text{pr}({Z}_i={z}_{i\pi_{i}}|\mathcal{F},\mathcal{Z}_{i})=\frac{\exp\{\gamma({z}_{i\pi_{i}}{u}_i^{T})\}}{\sum_{\widetilde{\pi}_{i}\in S_{n_{i}}}\exp\{\gamma({z}_{i\widetilde{\pi}_{i}}{u}_i^{T})\}}, \ \text{where ${u}_{i}=(u_{i1}, \dots, u_{in_{i}})\in [0,1]^{n_{i}}$}.
\end{equation}
When there are only two doses (i.e., binary exposures), equation (\ref{eqn: dose assignment after matching}) reduces to the classic Rosenbaum biased exposure assignments model in matched observational studies with binary exposures (\citealp{rosenbaum1987sensitivity, rosenbaum_obs}). Assuming that the matched sets are independent, we have 
\begin{equation*}
    \text{pr}({Z}={z}|\mathcal{F},\mathcal{Z})=\prod_{i=1}^{I} \frac{\exp\{\gamma({z}_{i}{u}_i^{T})\}}{\sum_{\pi_{i}\in S_{n_{i}}}\exp\{\gamma({z}_{i\pi_{i}}{u}_i^{T})\}}, \ \text{where ${z}=({z}_{1}, \dots, {z}_{I})\in \mathcal{Z}$ with $z_{i}=(z_{i1},\dots, z_{in_{i}})$}.
\end{equation*}
In sensitivity analysis, when testing a causal null hypothesis, given each prespecified sensitivity parameter, researchers typically report the worst-case $p$-value, defined as the largest $p$-value over all possible allocations of unmeasured confounders ${u}$ (\citealp{rosenbaum1987sensitivity, rosenbaum_obs}). For example, in matched observational studies with continuous exposures, for each prespecified sensitivity parameter $\Gamma=\exp(\gamma)$, given some test statistics $T({Z}, {R})$ and its observed value $t$, the (one-sided) worst-case $p$-value under $H_{0}$ is 
\begin{align}\label{eqn: worst-case p-value}
    \max_{{u}\in [0,1]^{N} }\text{pr}(T\geq t \mid \mathcal{F}, \mathcal{Z})&= \max_{{u}\in [0,1]^{N} }\sum_{{z}\in \mathcal{Z}} \Big[ \mathbbm{1}\{T({Z}={z}, {R})\geq t\}\times \text{pr}({Z}={z}|\mathcal{F},\mathcal{Z}) \Big].
\end{align}
The central problem in sensitivity analysis literature is to solve the worst-case $p$-value (\ref{eqn: worst-case p-value}) in various settings (\citealp{rosenbaum_obs}). When there are only two doses (i.e., binary exposures), the worst-case $p$-value (\ref{eqn: worst-case p-value}) has been solved for general matching designs and a wide range of commonly used test statistics (\citealp{rosenbaum1987sensitivity, rosenbaum_obs}). However, for matching with continuous exposures, except for some special cases such as pair matching (\citealp{rosenbaum1989sensitivity, gastwirth1998dual}), solving the worst-case $p$-value (\ref{eqn: worst-case p-value}) is still an open problem. In this paper, we solve this open problem with binary outcomes (which is often encountered in practice and is always the case for case-control studies) for both constant effects (Fisher's sharp null) and heterogenous effects (threshold attributable effects). For continuous outcomes, we show that solving the worst-case $p$-value (\ref{eqn: worst-case p-value}) is generally NP-hard.

\section{Sensitivity Analysis for Matched Observational Studies with Continuous Exposures and Binary Outcomes: the Sharp Null Case}
\label{sens_anal}

\subsection{Finite-population-exact sensitivity analysis via probabilistic lattice theory}\label{subsec: Finite-population-exact sensitivity analysis via probabilistic lattice theory}

In this section, we will derive a finite-population-exact sensitivity analysis procedure (i.e., solving the worst-case $p$-value (\ref{eqn: worst-case p-value})) for matched observational studies with continuous exposures and binary outcomes for testing Fisher's sharp null $H_{0}: r_{ij}(z)=r_{ij}(z^{\prime})$ for any $i, j$ and $z, z^{\prime}\in {Z}_{i}$. We will  introduce a distributive lattice tailored to the problem setting, define a general class of test statistics to which our results will apply, and establish that a specific allocation of unmeasured confounders yields a random variable that stochastically dominates all random variables arising from any other allocations (thus we solve the worst-case $p$-value problem (\ref{eqn: worst-case p-value})). First, in each matched set $i$, we reorder the indexes of units such that the binary outcomes are ordered from smallest to largest, i.e., $r_{i(1)} \leq \ldots \leq r_{i(n_i)}$. Suppose that in matched set $i$, there are $m_i$ outcomes with value 1, then we have $r_{i(n_{i}-m_{i}+1)}=\dots=r_{i(n_{i})}=1$ and $r_{i(1)} = \dots = r_{i(n_i-m_i)}=0$. For simplicity of notations, we assume that the exposure doses are distinct within each matched set, which is typically the case for continuous exposures. For a detailed discussion on handling ties of doses, see Section 2.10.3 of \citet{rosenbaum_obs}. To present our arguments based on probabilistic lattice theory, we require the following definition of distributive lattice (\citealp{gratzer2002general}).
\begin{definition}\label{def: distributive lattice}
A distributive lattice $\Omega$ is a lattice endowed with partial order $\lesssim$, join operation $\lor$, and meet operation $\land$ such that for any $x, y, z \in \Omega$,
$$ x \lor (y \land z) = (x \lor y) \land (x \lor z) \text{ and }
x \land (y \lor z) = (x \land y) \lor (x \land z).$$
\end{definition}

We then construct a distributive lattice in our problem setting as follows: for each matched set $i$, let ${s}^0_{i}$ denote the exposure doses assigned to the units with outcome 0, arranged in increasing order, and ${s}^{1}_{i}$ the doses assigned to the units with outcome 1, arranged in increasing order. We then let ${s}_i=({s}^0_i,{s}^1_i)$. In a matched set $i$, there are $n_i!$ possible exposure doses assignments, but only ${n_i \choose m_i}$ possible ${s}_i$. For example, suppose we have a matched set with doses assignment $(0.1,0.44,0.54,0.73,0.8)$ and outcome vector $(0,0,0,1,1)$. This specific exposure dose assignment yields ${s}_{i}^{0}= (0.1,0.44,0.54)$ and ${s}_i^1= (0.73,0.8)$. Note that any exposure dose assignment that permutes $(0.1,0.44,0.54)$ between the first three units and $(0.73,0.8)$ between the last two units corresponds to the same ${s}_i$. Also, the last $m_i$ elements of ${s}_i$ determine the first $n_i-m_i$ elements, and vice versa. We define the partial order `$\lesssim$' over the collection of all possible exposure dose assignments ${s}$ (denoted as $\Omega$): ${s} \lesssim {s}^{*} \text{ if and only if } {s}^1_i \leq {s}^{1*}_i \text{ elementwise for all } i=1,\ldots, I,$ where ${s}=({s}_{1},\dots,{s}_{I})$ and ${s}^{*}=({s}_{1}^{*},\dots,{s}_{I}^{*})$ are the vectorized (concatenated) versions of ${s}_i$ and ${s}^{*}_i$ over all matched sets. It is straightforward to verify that `$\lesssim$' is a well-defined partial order assuming no ties of doses within each matched set. We then let $({s} \lor {s}^{*})_i^1 \equiv \max\{{s}_i^1,{s}^{1*}_i\}$ and $({s} \land {s}^{*})_i^1 \equiv \min\{{s}_i^1,{s}^{1*}_i\}$, where the $\max$ and $\min$ operations are taken elementwise. Then each $({s} \lor {s}^{*})_{i}$ and $({s} \land {s}^{*})_{i}$ are also well-defined as they are uniquely determined by $({s} \lor {s}^{*})_i^1$ and $({s} \land {s}^{*})_i^1$ respectively. We let ${s} \lor {s}^{*}=(({s} \lor {s}^{*})_{1},\dots, ({s} \lor {s}^{*})_I)$ and ${s} \land {s}^{*}=(({s} \land {s}^{*})_1,\dots, ({s} \land {s}^{*})_I)$ be the vectorized version of $({s} \lor {s}^{*})_i$ and that of $({s} \land {s}^{*})_i$ respectively. Also, we introduce the notation $\text{sum}({s})$ to be the sum of all the entries of the vector ${s}$. It is straightforward to verify that $(\Omega, \wedge, \vee)$ is a distributive lattice. Consider an arbitrary test statistic $T({Z},{R})$ that can be written as a function of ${S}=({S}_1, \dots, {S}_I)$, where ${S}_i \equiv (\{Z_{ij}: R_{ij} = 0\}, \{Z_{ij}: R_{ij} = 1\})$ (i.e., ${S}$ is a random variable taking values in $\Omega$). This is a mild requirement since it only requires $T({Z},{R})$ to be a function of the doses assigned to units with outcome 1 (and 0). Next, we formally define the conditionally isotonic and the stratum-wise isotonic test statistic $T({Z},{R})$ for ${Z}\in \mathbbm{R}^{N}$ and ${R}\in \{0,1\}^{N}$. 
\begin{definition}\label{def:isotonic}
A test statistic $T({Z},{R})$ is called conditionally isotonic if it can be written as $T({Z},{R}) = f({S})=f({S}_1,\ldots,{S}_I)$ for some function $f$ and after fixing (i.e., conditioning on any fixed value of) ${R}={r}\in \{0,1\}^{N}$, it holds that $f({s}) \leq f({s}^*)$ when ${s} \lesssim {s}^*$, where ${s},{s}^* \in \Omega$.
\end{definition}

An important subclass of conditionally isotonic test statistics is the class of stratum-wise isotonic test statistics, of which the formal definition is as follows:

\begin{definition}
\label{def:stratum-wise}
A test statistic $T$ is a stratum-wise isotonic test statistic if it can be written in the form $T = \sum_{i=1}^I q_{i}({Z}_{i}, {R}_{i})$ for some function $q_{i}$ of ${Z}_{i}=(Z_{i1},\dots, Z_{in_{i}})$ and ${R}_{i}=(R_{i1},\dots, R_{in_{i}})$, where each $q_{i}$ is isotonic in $S_{i}=(\{Z_{ij}: R_{ij} = 0\}, \{Z_{ij}: R_{ij} = 1\})$ for each fixed ${R}_{i}$.
\end{definition}
A wide range of commonly used test statistics with continuous exposures and binary outcomes satisfy the conditionally isotonic property in Definition~\ref{def:isotonic} and the stratum-wise property in Definition ~\ref{def:stratum-wise}. For example, the permutational t-test $Z R^{T}=\sum_{i=1}^{I}\sum_{j=1}^{n_{i}}Z_{ij}R_{ij}$ is stratum-wise isotonic, which reduces to the Mantel-Haenszel test when the exposures are binary. If researchers plan to use a test with dichotomized exposures, then ${Z}_{>c}{R}^{T}=\sum_{i=1}^{I}\sum_{j=1}^{n_{i}}\mathbbm{1}\{Z_{ij}>c\}R_{ij}$ is a widely used choice and is also stratum-wise isotonic, where $c$ is a scientifically meaningful threshold for dichotomizing a continuous dose. Generally, for any monotonically non-decreasing functions $m_{i1}$ and $m_2$, the test statistic $\sum_{i=1}^{I}\sum_{j=1}^{n_{i}}m_{i1}(Z_{ij})m_2(R_{ij})$ is stratum-wise isotonic. Here the function $m_{i1}$ can depend on the matched set index $i$, such as the rank functions of exposure doses within or across matched sets (\citealp{rosenbaum1989sensitivity, gastwirth1998dual, bo_continuous}). We then introduce the Holley inequality (\citealp{holley1974remarks}) in probabilistic lattice theory, which is a key tool for proving our theoretical results in this section.
\begin{lemma}
\label{useful1} 
(Holley Inequality) Let $f$ be an isotonic function on a finite distributive lattice $L$ (i.e., for any ${s}, {s}^{*}\in L$, we have $f({s}) \leq f({s}^{*})$ whenever ${s} \lesssim {s}^{*}$). If ${A}$ and ${B}$ are two random variables satisfying 
\begin{equation*}
    \text{pr}({A}={s}\vee {s}^{*})\times \text{pr}({B}={s}\wedge {s}^{*}) \geq \text{pr}({A}={s})\times \text{pr}({B}= {s}^{*}) \quad \text{for all ${s}, {s}^{*} \in L$,}
\end{equation*}
then we have $E\{f({A})\} \geq E\{f({B})\}$.
\end{lemma}
Since we have established that $(\Omega, \wedge, \vee)$ is a distributive lattice, to find the worst-case $p$-value (\ref{eqn: worst-case p-value}), we aim to apply the Holley inequality stated in Lemma~\ref{useful1} to show that we can find an exact stochastically dominating random variable $T^{+}$ for the test statistic $T({Z},{R})=f({S})$ (i.e., $\text{pr}(T^{+}\geq t)=\max_{{u}\in [0,1]^{N} }\text{pr}(T\geq t \mid \mathcal{F}, \mathcal{Z})\geq \text{pr}(T\geq t \mid \mathcal{F}, \mathcal{Z})$ for all $t$) under the sharp null and the Rosenbaum sensitivity analysis model (\ref{eqn: Rosenbaum exposure dose model}). Define ${u}^+$ as the vector that takes the value of 1 if the corresponding unit's outcome is 1 and 0 if the corresponding unit's outcome is 0 (i.e., ${u}^{+}=(R_{11}, \dots, R_{In_{I}})$). Under model (\ref{eqn: dose assignment after matching}) and some fixed value $\gamma$ of the sensitivity parameter, let ${S}_{{u}}$ and ${S}_{{u}^{+}}$ denote the distribution of ${S}$ when unmeasured confounders equal ${u}$ and that when unmeasured confounders equal ${u}^{+}$, respectively. Theorem~\ref{holley_condition} claims that the random variables ${S}_{{u}}$ and ${S}_{{u}^{+}}$ satisfy the condition in the Holley inequality, and Corollary~\ref{stoch_dom} solves the worst-case $p$-value problem. All the detailed proofs in this paper can be found in the supplementary material. 
\begin{theorem}\label{holley_condition}
Let ${u}^{+}$ take the value $(R_{11}, \dots, R_{In_{I}})\in \{0,1\}^{N}$. For any ${u}\in [0,1]^{N}$, under the sharp null $H_{0}$, we have 
\begin{equation*}
    \text{pr}({S}_{{u}^{+}}={s}\vee {s}^{*})\times \text{pr}({S}_{{u}}={s}\wedge {s}^{*}) \geq \text{pr}({S}_{{u}^{+}}={s})\times \text{pr}({S}_{{u}}= {s}^{*}) \quad \text{for all ${s}, {s}^{*} \in \Omega$.}
\end{equation*}
\end{theorem}
\begin{corollary}\label{stoch_dom}
Suppose that ${Z}\in \mathbbm{R}^{N}$ and ${R}\in \{0,1\}^{N}$. For any test statistic $T({Z},{R})=f({S})$ that is conditionally isotonic, under the sharp null $H_{0}$, for any ${u}\in [0,1]^{N}$ and any $t$, we have $\text{pr}(f({S}_{{u}}) \geq t \mid \mathcal{F}, \mathcal{Z}) \leq \text{pr}(f({S}_{{u}^{+}}) \geq t).$ That is, for any $t$, we have 
\begin{equation}\label{eqn: worst-case p-value solved}
\max_{{u}\in [0,1]^{N} }\text{pr}(T\geq t \mid \mathcal{F}, \mathcal{Z})=\sum_{{z}\in \mathcal{Z}} \Big[ \mathbbm{1}\{T({Z}={z}, {R})\geq t\}\times \prod_{i=1}^{I} \frac{\exp\{\gamma({z}_{i}{u}_{i}^{+T})\}}{\sum_{\pi_{i}\in S_{n_{i}}}\exp\{\gamma({z}_{i\pi_{i}}{u}_{i}^{+T})\}}\Big],
\end{equation}
where ${u}_{i}^{+}=(R_{i1}, \dots, R_{in_{i}})$ for each $i$.
\end{corollary}
\begin{proof}
  This is an immediate consequence of the Holley inequality (stated in Lemma~\ref{useful1}) and Theorem~\ref{holley_condition} after noticing that $\text{pr}(T({Z}, {R}) \geq t\mid \mathcal{F}, \mathcal{Z}) =\text{pr}(f({S}) \geq t\mid \mathcal{F}, \mathcal{Z}) = E[\mathbbm{1}\{f({S}) \geq t\}\mid \mathcal{F}, \mathcal{Z}]$ and that for any fixed $t$, $\mathbbm{1}\{f({S}) \geq t\}$ is isotonic in ${S}$ if $f({S})$ is isotonic.
\end{proof}

Corollary~\ref{stoch_dom} gives a way of conducting a finite-population-exact sensitivity analysis for testing the sharp null $H_{0}$ with continuous exposures and binary outcomes by giving a closed form for the finite-population-exact worst-case $p$-value, which can be approximated by Monte Carlo simulation. When the sample size is large, an alternative (and more computationally efficient) way of approximating the worst-case $p$-value (\ref{eqn: worst-case p-value solved}) is by normal approximation, which will also be utilized in later sections and is stated as follows:
\begin{theorem}\label{thm: normal approximation} 
Let $T({Z},{R})= \sum_{i=1}^I q_{i}({Z}_{i}, {R}_{i})$ be a stratum-wise isotonic test statistic where ${Z}\in \mathbbm{R}^{N}$ and ${R}\in \{0,1\}^{N}$. Under Condition~\ref{normal_reg} in Appendix A.2 in the online supplementary material (a mild regularity condition; see the corresponding discussion in Appendix A.2 for details) and the sharp null $H_{0}$, we have
    \begin{equation*}
    \max_{{u}\in [0,1]^{N} }\text{pr}(T\geq t \mid \mathcal{F}, \mathcal{Z}) \doteq 1-\Phi\big[\{t-E_{{u}^+}(T)\} / \{\text{var}_{{u}^+}(T)\}^{1/2}\big] \ \text{as $I\rightarrow \infty$},
\end{equation*}
where $\doteq$ means asymptotically equal, $\Phi$ is the distribution function of standard normal distribution, the value of ${u}^+={R}$, and $E_{{u}^+}$ and $\text{var}_{{u}^+}$ denote the expectation and variance of $T$ under unmeasured confounding ${u}^+$ respectively. Specifically, let $s_{i}=\sum_{\Tilde{\pi}_{i}}\exp(\gamma ({Z}_{i\Tilde{\pi}_{i}}{R}_i^{T}))$, we have $$E_{{u}^{+}}(T)= \sum_{i=1}^I \sum_{\pi_{i}} s_{i}^{-1}\big\{\exp(\gamma ({Z}_{i\pi_{i}}{R}_i^{T}))\times q_{i}({Z}_{i\pi_{i}},{R}_{i})\big\},$$ $$\text{var}_{{u}^{+}}(T)=\sum_{i=1}^I \sum_{\pi_{i}} s_{i}^{-1}\big\{\exp(\gamma ({Z}_{i\pi_{i}}{R}_i^{T}))\times q_{i}^{2}({Z}_{i\pi_{i}},{R}_{i})\big\}-\sum_{i=1}^I \Big[\sum_{\pi_{i}} s_{i}^{-1}\big\{\exp(\gamma ({Z}_{i\pi_{i}}{R}_i^{T}))\times q_{i}({Z}_{i\pi_{i}},{R}_{i})\big\}\Big]^{2}.$$

\end{theorem}

\subsection{Asymptotic evaluation of the power of sensitivity analysis via design sensitivity}\label{design_power}

In Section~\ref{subsec: Finite-population-exact sensitivity analysis via probabilistic lattice theory}, for a wide range of commonly used tests, we established the corresponding valid sensitivity analysis approach. The next important question is: how to evaluate the performance of various tests in sensitivity analysis? In an analysis that assumes no unmeasured confounding, we look at the power of a test, which is the probability of successfully rejecting the null hypothesis under some alternative. In parallel, the \textit{power of a sensitivity analysis} is the probability that a sensitivity analysis will successfully reject the null under some alternative for any possible allocations of the unmeasured confounders given some prespecified sensitivity parameter (e.g., the $\Gamma=\exp(\gamma)$ in Rosenbaum's model (\ref{eqn: Rosenbaum exposure dose model})). Specifically, for each fixed $\Gamma$, the power of an $\alpha$ level sensitivity analysis using some test statistic $T$ is calculated as the probability that the worst-case $p$-value (\ref{eqn: worst-case p-value}) is below $\alpha$. Following Rosenbaum's power of sensitivity analysis framework (\citealp{rosenbaum2004design, rosenbaum2020design}), we consider power under the alternative of a ``favorable situation" in which there is a treatment effect and no hidden bias. That is, power of a sensitivity analysis measures a test's ability to detect an actual treatment effect when taking potential unmeasured confounding into account (e.g., conducting a sensitivity analysis). Under some regularity assumptions, there typically exists a threshold called the design sensitivity, such that as the sample size increases, the power of a sensitivity analysis goes to one if the analysis is performed with $\Gamma < \widetilde{\Gamma}$, and the power of a sensitivity analysis goes to zero if performed with $\Gamma > \widetilde{\Gamma}$. That is, the design sensitivity $\widetilde{\Gamma}$ is the sharp transition of consistency of a test in a sensitivity analysis (\citealp{rosenbaum2004design, rosenbaum2020design}). It is a powerful tool for asymptotically comparing two test statistics under each data-generating process -- the test with a larger design sensitivity $\widetilde{\Gamma}$ is asymptotically more robust to unmeasured confounding.

To facilitate evaluating the asymptotic performances of various tests in sensitivity analysis, we derive a design sensitivity formula that works for a large subclass of stratum-wise isotonic test statistics. The detailed design sensitivity formula is stated in the following Theorem~\ref{thm: design sensitivity}.

\begin{theorem}\label{thm: design sensitivity}
Consider a general class of stratum-wise isotonic test statistics with the form $T=\sum_{i=1}^{I}q_{i}({Z}_i,{R}_i)$ with $q_{i}({Z}_i,{R}_i)=\sum_{j=1}^{n_i} m(Z_{ij}) \times R_{ij}$, where $m$ is some bounded and non-decreasing function. Suppose that $(Z_{i1}, \ldots,  Z_{in_i}, R_{i1}, \ldots,  R_{in_i}, n_i)$ are independent and identically distributed realizations from some multivariate distribution $F$, where $2 \leq n_i \leq M$ is an integer random variable, $Z_{ij}\in \mathbbm{R}$, and $R_{ij} \in \{0,1\}$. Suppose that, conditional on $n_i$, the $(Z_{i1},R_{i1}), \ldots, (Z_{in_i},R_{i_{n_i}})$ are pairwise independent and identically distributed, and $0 < \text{cor}(m(Z_{ij}),R_{ij}) < 1$. Then, the following equation
\begin{equation}\label{eqn: design sensitivity}
E\left[\sum_{\pi_{i}\in S_{n_{i}} }\Big \{ \frac{\exp(\gamma ({Z}_{i\pi_{i}}{R}_i^{T}))}{\sum_{\widetilde{\pi}_{i}\in S_{n_{i}} }\exp(\gamma ({Z}_{i\widetilde{\pi}_{i}}{R}_i^{T}))}\times q_{i}({Z}_{i\pi_{i}}, {R}_{i}) \Big\} \right]=E\left[q_{i}({Z}_i,{R}_i) \right]
\end{equation}
has a unique solution $\widetilde{\gamma} > 0$. Then $\widetilde{\Gamma} = \exp(\widetilde{\gamma})$ is the design sensitivity. That is, if we let $\Psi_{\Gamma, I}$ denote the power of sensitivity analysis conducted under the sensitivity parameter $\Gamma=\exp(\gamma)$ and the number of matched sets $I$, then we have: $\Psi_{\Gamma, I}\rightarrow 1$ if $\Gamma < \widetilde{\Gamma}$ and $\Psi_{\Gamma, I}\rightarrow 0$ if $\Gamma > \widetilde{\Gamma}$.
\end{theorem}
Theorem~\ref{thm: design sensitivity} shows that the design sensitivity $\widetilde{\Gamma}$ depends only on the underlying data generating distribution $F$ and is independent of the sample size and the significance level $\alpha$. Therefore, it provides an elegant theoretical tool for asymptotically comparing different test statistics in sensitivity analysis.

We now use the design sensitivity formula (\ref{eqn: design sensitivity}) to asymptotically compare various tests in matched observational studies. For illustration, we compare the following four test statistics: the permutational t-test $\sum_{i=1}^{I}\sum_{j=1}^{n_{i}} Z_{ij} R_{ij}$ (denoted as $T_{t}$) and the Mantel-Haenszel test $\sum_{i=1}^{I}\sum_{j=1}^{n_{i}} \mathbbm{1}\{Z_{ij} > c\} R_{ij}$ for $c = 0.1, 0.25, 0.5$ (denoted as $T_{0.1}, T_{0.25}, T_{0.5}$). For conveniently generating various fixed datasets, we draw independent and identically distributed samples from some data-generating model to investigate the design sensitivity and finite-population power, although our sensitivity analysis method does not rely on such assumptions to be finite-population-valid. Specifically, we consider the following outcome model:
\begin{equation*}
    r_{ij}(z)\sim \text{Bern}(\text{expit}(A_i+f(z)\beta)),
\end{equation*}
where $A_i$ is a random effect for matched set $i$. In addition, we randomly draw the number of units in matched set $i$ (i.e., the $n_i$) from some discrete distribution. Let $\delta_0$ denote the point mass at $0$. We consider various data-generating processes under the following five factors: 1) $f(z)=z^a$ for $a \in \{1/4,1/2,2,4\}$. 2) $A_i \sim N(0,1)$. 3) $\beta = 1.5$.
4) $Z_{ij} \sim \text{Unif}[0,1]$ or $Z_{ij} \sim \text{Beta}(2,2)$. 5) $n_i \sim 2 + (0.9 \times \delta_0 + 0.1 \times \text{Poisson(0.5)})$. This gives a total of $4 \times 1 \times 1 \times 2 \times 1 = 8$ data-generating processes. Under each data-generating process, we can calculate the corresponding design sensitivity by solving equation (\ref{eqn: design sensitivity}) via the following two steps: we first use the Monte Carlo method to calculate the expectations involved in (\ref{eqn: design sensitivity}) for each specific $\Gamma=\exp(\gamma)$ and then use the bisection method search or line search to find the $\widetilde{\gamma}$ such that equation (\ref{eqn: design sensitivity}) holds (then the design sensitivity $\widetilde{\Gamma}=\exp(\widetilde{\gamma})$). To check if the insights about the asymptotic performances of various tests obtained from calculating the design sensitivity would also hold with only a moderate sample size, we also conduct finite-sample power simulations using the same set of test statistics and data-generating processes as those considered above, with 2000 matched sets. For ease of computation, we utilize normal approximation (Theorem~\ref{thm: normal approximation}) to calculate the worst-case $p$-values under various $\Gamma$. We also report the simulation results of an adaptive test combining $T_{t}$ and $T_{0.1}$ using the Bonferroni correction (denoted as $T_{\text{adap}}$) (\citealp{rosenbaum2012testing, heng2021increasing}). The design sensitivity of such an adaptive test is the larger of the design sensitivity of the $T_{t}$ and that of $T_{0.1}$ (\citealp{rosenbaum2012testing}). The simulation results of design sensitivity and finite-sample power under $0.05$ significance level are in Table~\ref{tab:ds+pow} (recall that for each fixed test statistic, based on Corollary~\ref{stoch_dom} and Theorem~\ref{thm: normal approximation}, our framework can produce corresponding sensitivity analysis that is less conservative than any other sensitivity analysis methods).

An important observation from Table~\ref{tab:ds+pow} is that different test statistics can perform very differently in a sensitivity analysis, and which test statistic outperforms the others depends on the underlying data-generating process. Specifically, when the chosen transformation of the doses $m(\cdot)$ leads to higher $\text{cor}(m(Z_{ij}), R_{ij})$, the design sensitivity and power are typically higher. Typically, we expect higher $\text{cor}(m(Z_{ij}), R_{ij})$ when $m(\cdot)$ better approximates $f(\cdot)$ in the outcome generating process. For example, when $f(z) = z^{0.25}$ or $f(z) = z^{0.5}$, the $T_{0.1}$ and $T_{0.25}$ generally perform better than $T_{0.5}$ and $T_{t}$ as $m(z) = \mathbbm{1}\{z > 0.1\}$ and $m(z) = \mathbbm{1}\{z > 0.25\}$ are closer in form to the functions $f(z) = z^{0.25}$ and $f(z) = z^{0.5}$. Instead, when $f(z) = z^{2}$ and $f(z) = z^{4}$, the $T_{t}$ and $T_{0.5}$ have larger values of design sensitivity and power since $m(z) = z$ and $m(z) = \mathbbm{1}\{z > 0.5\}$ are closer in form to the functions $f(z) = z^{2}$ and $f(z) = z^{4}$. Generally, higher design sensitivity implies higher power of sensitivity analysis, although there may be exceptions. For example, when $f(z)=z^2$, even though $T_{0.5}$ has a comparable design sensitivity to that of $T_{t}$, the $T_{t}$ has slightly higher power. This is not surprising since $T_{0.5}$ loses efficiency through dichotomization, leading to more concordant matched sets (i.e., matched sets in which all the observed outcomes are equal).  In finite samples, we notice that an adaptive test $T_{\text{adap}}$ that combines $T_{t}$ and $T_{0.1}$ using the Bonferroni correction always performs much better than the worse one among the two candidate tests $T_{t}$ and $T_{0.1}$, regardless of the unknown data-generating process. Therefore, the design sensitivity formula allows us to gain theoretical insights into when a test would outperform the others and when it would not, which can provide useful guidance for study designs (\citealp{rosenbaum2004design, rosenbaum2020design}). In Appendix F in the online supplementary material, we present additional simulation results covering a broader range of scenarios, including the setting when the exposure dose distribution has a positive mass at the zero dose (i.e., many units/individuals are completely unexposed). This setting is common in many real data applications, including the real data example described in Section~\ref{subsec: motivating example}. The general pattern of these additional simulation results agrees with those presented here in the main text; see Appendix F for detailed simulation results and related discussions.

\begin{centering}
\begin{table}[htbp]
\centering
\caption{Design sensitivity and finite-sample power for the five test statistics under $f(z) = z^a$ for $a = 1/4, 1/2, 2, 4$ and $\beta=1.5$ for the outcome model. The random effect for each matched set is drawn from $N(0,1)$, and the exposure dose $Z$ is drawn from $\text{Unif}[0,1]$ or $\text{Beta}(2,2)$.}
\small
\begin{tabular}{lrrrrr|rrrrr}
\hline
\multicolumn{11}{c}{$f(z)=z^{0.25}$} \\
 & \multicolumn{5}{c|}{$Z\sim \text{Unif}[0,1]$} & \multicolumn{5}{c}{$Z\sim$ Beta(2,2)} \\
\cline{1-1} \cline{2-6} \cline{7-11}
Test & $T_{t}$ & $T_{0.1}$ & $T_{0.25}$ & $T_{0.5}$ & $T_{\text{adap}}$ & $T_{t}$ & $T_{0.1}$ & $T_{0.25}$ & $T_{0.5}$ & $T_{\text{adap}}$ \\
\hline
$\widetilde{\Gamma}$ & 2.17 & 3.16 & 2.64 & 2.11 & 3.16 & 1.95 & 2.87 & 2.32 & 1.96 & 2.87\\
$\Gamma = 1.75$ & 0.37 & 0.54 & 0.45 & 0.25 & 0.48 & 0.19 & 0.19 & 0.24 & 0.13 & 0.21\\
$\Gamma = 2.00$ & 0.14 & 0.39 & 0.25 & 0.10 & 0.28 & 0.08 & 0.14 & 0.13 & 0.04 & 0.12\\
$\Gamma = 2.25$ & 0.03 & 0.25 & 0.11 & 0.03 & 0.16 & 0.02 & 0.12 & 0.07 & 0.02 & 0.08\\
$\Gamma = 2.50$ & 0.01 & 0.16 & 0.05 & 0.01 & 0.08 & 0.01 & 0.09 & 0.04 & 0.01 & 0.06\\
%\hline
%$\Gamma = 2.75$ & 0.00 & 0.08 & 0.02 & 0.00 & 0.04 & 0.00 & 0.07 & 0.02 & 0.00 & 0.04\\
%\hline
%$\Gamma = 3.00$ & 0.00 & 0.06 & 0.01 & 0.00 & 0.04 & 0.00 & 0.06 & 0.02 & 0.00 & 0.04\\
\end{tabular}

\begin{tabular}{lrrrrr|rrrrr}
\hline
\multicolumn{11}{c}{$f(z)=z^{0.5}$} \\
 & \multicolumn{5}{c|}{$Z\sim \text{Unif}[0,1]$} & \multicolumn{5}{c}{$Z\sim$ Beta(2,2)} \\
\cline{1-1} \cline{2-6} \cline{7-11}
Test & $T_{t}$ & $T_{0.1}$ & $T_{0.25}$ & $T_{0.5}$ & $T_{\text{adap}}$ & $T_{t}$ & $T_{0.1}$ & $T_{0.25}$ & $T_{0.5}$ & $T_{\text{adap}}$ \\
\hline
$\widetilde{\Gamma}$ & 3.21 & 4.51 & 3.65 & 3.08 & 4.51 & 2.82 & 4.43 & 3.38 & 2.62 & 4.43\\
$\Gamma = 3.00$ & 0.18 & 0.36 & 0.27 & 0.09 & 0.27 & 0.06 & 0.10 & 0.09 & 0.03 & 0.08\\
$\Gamma = 3.25$ & 0.06 & 0.28 & 0.17 & 0.04 & 0.18 & 0.02 & 0.08 & 0.05 & 0.00 & 0.05\\
$\Gamma = 3.50$ & 0.03 & 0.21 & 0.10 & 0.02 & 0.13 & 0.01 & 0.08 & 0.04 & 0.00 & 0.04\\
$\Gamma = 3.75$ & 0.01 & 0.16 & 0.07 & 0.01 & 0.10 & 0.00 & 0.07 & 0.02 & 0.00 & 0.03\\
%\hline
%$\Gamma = 4.00$ & 0.00 & 0.12 & 0.05 & 0.00 & 0.08 & 0.00 & 0.05 & 0.01 & 0.00 & 0.02\\
%\hline
%$\Gamma = 4.25$ & 0.00 & 0.10 & 0.03 & 0.00 & 0.07 & 0.00 & 0.05 & 0.01 & 0.00 & 0.02\\
\end{tabular}

\begin{tabular}{lrrrrr|rrrrr}
\hline
\multicolumn{11}{c}{$f(z)=z^{2}$} \\
 & \multicolumn{5}{c|}{$Z\sim \text{Unif}[0,1]$} & \multicolumn{5}{c}{$Z\sim$ Beta(2,2)} \\
\cline{1-1} \cline{2-6} \cline{7-11}
Test & $T_{t}$ & $T_{0.1}$ & $T_{0.25}$ & $T_{0.5}$ & $T_{\text{adap}}$ & $T_{t}$ & $T_{0.1}$ & $T_{0.25}$ & $T_{0.5}$ & $T_{\text{adap}}$ \\
\hline
$\widetilde{\Gamma}$ & 4.55 & 2.98 & 3.46 & 4.59 & 4.55 & 4.42 & 2.90 & 3.20 & 4.57 & 4.42\\
$\Gamma = 3.00$ & 0.67 & 0.03 & 0.12 & 0.56 & 0.55 & 0.51 & 0.03 & 0.06 & 0.40 & 0.37\\
$\Gamma = 3.25$ & 0.47 & 0.01 & 0.06 & 0.42 & 0.36 & 0.36 & 0.03 & 0.04 & 0.29 & 0.26\\
$\Gamma = 3.50$ & 0.32 & 0.01 & 0.03 & 0.28 & 0.22 & 0.26 & 0.02 & 0.02 & 0.22 & 0.19\\
$\Gamma = 3.75$ & 0.21 & 0.00 & 0.02 & 0.16 & 0.11 & 0.19 & 0.02 & 0.01 & 0.16 & 0.09\\
%\hline
%$\Gamma = 4$ & 0.10 & 0.00 & 0.01 & 0.10 & 0.05 & 0.10 & 0.01 & 0.00 & 0.11 & 0.05\\
%\hline
%$\Gamma = 4.25$ & 0.05 & 0.00 & 0.01 & 0.05 & 0.02 & 0.05 & 0.01 & 0.00 & 0.07 & 0.04\\
\end{tabular}

\begin{tabular}{lrrrrr|rrrrr}
\hline
\multicolumn{11}{c}{$f(z)=z^{4}$} \\
 & \multicolumn{5}{c|}{$Z\sim \text{Unif}[0,1]$} & \multicolumn{5}{c}{$Z\sim$ Beta(2,2)} \\
\cline{1-1} \cline{2-6} \cline{7-11}
Test & $T_{t}$ & $T_{0.1}$ & $T_{0.25}$ & $T_{0.5}$ & $T_{\text{adap}}$ & $T_{t}$ & $T_{0.1}$ & $T_{0.25}$ & $T_{0.5}$ & $T_{\text{adap}}$ \\
\hline
$\widetilde{\Gamma}$ & 3.24 & 1.87 & 2.15 & 3.00 & 3.24 & 2.85 & 1.42 & 1.88 & 2.69 & 2.85\\
$\Gamma = 1.75$ & 0.97 & 0.07 & 0.20 & 0.88 & 0.94 & 0.74 & 0.03 & 0.06 & 0.51 & 0.63\\
$\Gamma = 2.00$ & 0.86 & 0.02 & 0.08 & 0.67 & 0.79 & 0.52 & 0.02 & 0.03 & 0.30 & 0.37\\
$\Gamma = 2.25$ & 0.65 & 0.01 & 0.04 & 0.44 & 0.50 & 0.29 & 0.01 & 0.01 & 0.17 & 0.21\\
$\Gamma = 2.50$ & 0.37 & 0.00 & 0.01 & 0.22 & 0.27 & 0.17 & 0.01 & 0.01 & 0.10 & 0.10\\
\hline
%$\Gamma = 2.75$ & 0.21 & 0.00 & 0.00 & 0.12 & 0.13 & 0.08 & 0.01 & 0.00 & 0.04 & 0.05\\
%\hline
%$\Gamma = 3$ & 0.10 & 0.00 & 0.00 & 0.06 & 0.07 & 0.04 & 0.00 & 0.00 & 0.02 & 0.03\\
\end{tabular}

\label{tab:ds+pow}
\end{table}
\end{centering}
 
\section{Asymptotically Exact Sensitivity Analysis Beyond the Sharp Null}
\label{attribut}

\subsection{Generalizing attributable effects to the continuous exposure case}\label{subsec: Generalizing attributable effects to the continuous exposure case}

While the sharp null is one of the most widely considered causal null hypotheses and is often regarded as a first step in a cause-and-effect analysis \citep{imbens2015causal}, in many observational studies with binary exposures, researchers have also investigated other weak null hypotheses that allow effect heterogeneity (e.g., treatment may affect some study units but not others); see \citet{rosenbaum_attributable}, \citet{rand_bin}, \citet{li_ding_finite}, and \citet{fogarty_composite} for some examples, among many others. Among various weak nulls considered in matched observational studies with binary exposures and outcomes, \textit{attributable effects} are among the most widely considered ones (\citealp{rosenbaum_attributable, rosenbaum_attributable_obs, hansen2009attributing, hasegawa2017sensitivity}). Specifically, if the exposure $Z_{ij}\in \{0,1\}$ (i.e., treated or control), the attributable effects is defined as $\text{AE}= \sum_{i=1}^{I}\sum_{j=1}^{n_{i}} Z_{ij}(R_{ij}-r_{ij}(0))$, the number of treated units who experienced events (i.e., the binary outcome equalling one) caused by the treatment, that is, events that would not have occurred if the treated units received controls. A valid randomization inference procedure for attributable effects was given in \citet{rosenbaum_attributable}, and the corresponding sensitivity analysis method can be found in \citet{rosenbaum_attributable_obs}. In this section, we generalize the concept of attributable effect to the continuous exposure case and derive corresponding randomization inference and sensitivity analysis procedures. We consider the following assumption:
\begin{assumption}\label{mon}
(Monotonicity):  $r_{ij}(z) \leq r_{ij}(z')$ for all $z \leq z'$.
\end{assumption}
Assumption~\ref{mon} requires that the potential outcomes are non-increasing with respect to the exposure dose, which has been commonly used in many settings (\citealp{rosenbaum_obs, zhang2023social, heng2023instrumental}). Under Assumption~\ref{mon}, a natural generalization of the attributable effects to the continuous exposure case is the following \textit{Threshold Attributable Effect} (TAE): 
\begin{equation*}
    \text{TAE}(c) = \sum_{i=1}^{I}\sum_{j=1}^{n_{i}} \mathbbm{1}\{Z_{ij} > c\}(R_{ij}-r_{ij}(0))
\end{equation*}
where $c$ is some pre-specified, scientifically or policy-meaningful dose level, and $r_{ij}(0)$ is the potential outcome under zero dose. Under Assumption \ref{mon}, we have TAE$(c) \geq 0$, and the TAE$(c)$ has the interpretation of the number of units with exposure greater than c who had an event (i.e., $R_{ij}=1$) that would not occur if they had been completely unexposed. Since our inferential method will rely on the pivot $\sum_{i=1}^{I}\sum_{j=1}^{n_{i}} \mathbbm{1}\{Z_{ij} > c\}r_{ij}(0)$, to appropriately draw inferences about the TAE, there must be a non-negligible fraction of individuals for which we observe the potential outcome $r_{ij}(0)$. This is indeed the case in our real data example, as a significant proportion of individuals were completely unexposed. Alternatively, if there exists a scientifically-informed exposure level $\epsilon < c$ for which $r_{ij}(z) = r_{ij}(0)$ for all $0 < z < \epsilon$, the potential outcome $r_{ij}(0)$ would be observed for all individuals with exposure less than $\epsilon$. Thus, the TAE is a scientifically meaningful and statistically testable estimand in many settings, and our goal is to construct a confidence set for the TAE. To do so, note that $\sum_{i=1}^{I}\sum_{j=1}^{n_{i}} \mathbbm{1}\{Z_{ij}>c\}R_{ij}$ is fully observed, so we only have $\sum_{i=1}^{I}\sum_{j=1}^{n_{i}} \mathbbm{1}\{Z_{ij}>c\}r_{ij}(0)$ as the pivotal quantity. The vector ${r}_0=(r_{11}(0),\dots, r_{In_{I}}(0))$ is partially observed, but we may test its equality to any fixed ${r}_0'=(r_{11}(0)', \dots, r_{In_{I}}(0)') \in \{0,1\}^{N}$. Specifically, we can test the null $H_0^{\prime}: {r}_0 = {r}_0'$ using the test statistic $\sum_{i=1}^{I}\sum_{j=1}^{n_{i}} \mathbbm{1}\{Z_{ij}>c\} r_{ij}(0)'$ and by using the randomization distribution of ${Z}$. Also, some ${r}_0'$ will simply be incompatible with the observed data. First, for any subject $ij$ for which $r_{ij}(0)$ is observed, any $H_0^{\prime}$ where $r_{ij}(0)' \neq R_{ij}$ can be safely rejected. Second, under Assumption~\ref{mon} (monotonicity assumption), for a subject with positive dose exposure but $R_{ij} = 0$, any $H_0$ where $r_{ij}(0)' = 1 > R_{ij}$ can be safely rejected. Importantly, for testing $H_0^{\prime}: {r}_0 = {r}_0'$ with any fixed ${r}_0'$, we can also conduct a sensitivity analysis in an identical fashion as described in Section \ref{sens_anal}. A naive procedure to construct a $100(1-\alpha)\%$ confidence set for the TAE for any prespecified $\Gamma=\exp(\gamma)$ could follow a two-step procedure. In Step 1, for all ${r}_0' \in \{0,1\}^N$ that are compatible, we find all ${r}_0'$ such that TAE is equal to some $\Delta$; we call this set $\mathcal{R}_0^\Delta$. In Step 2, for all ${r}_0' \in \mathcal{R}_0^\Delta$, we test the sharp null $H_0^{\prime}: {r}_0 = {r}_0'$ at level $\alpha$ under the pre-specified $\Gamma$. If there exists any ${r}_0' \in \mathcal{R}_0^\Delta$ that cannot be rejected, we include $\Delta$ in the $100(1-\alpha)\%$ confidence set. Otherwise, we do not include $\Delta$. However, the size of $\mathcal{R}_0^\Delta$ can be very large, so direct enumeration may be infeasible. In Section~\ref{subsec: Sensitivity Analysis for the Threshold Attributable Effect in Matched Observational Studies via Mixed Integer Programming}, we propose a mixed-integer programming approach to more computationally efficiently conduct randomization inferences and sensitivity analyses for TAE.

\subsection{Sensitivity analysis for the threshold attributable effect via mixed-integer programming}\label{subsec: Sensitivity Analysis for the Threshold Attributable Effect in Matched Observational Studies via Mixed Integer Programming}

Using the normal approximation, we can conduct inference for TAE by solving a mixed-integer quadratically constrained linear program. Specifically, under normal approximation, the event that the null $H_0^{\prime}: {r}_0 = {r}_0'$ is not rejected at level $\alpha$ under model (\ref{eqn: dose assignment after matching}) with fixed $\Gamma=\exp(\gamma)$ can be expressed as 
\begin{equation}
\label{reject}
    \min_{{u}}  \big[
  \{{z}_{>c}{r}_0^{\prime T} - E_{{u}}({Z}_{>c}{r}_0^{\prime T})\}^2 - \chi^2_{1-\alpha} \times \text{var}_{{u}}({Z}_{>c}{r}_0^{\prime T}) \big] \leq 0,
\end{equation}
where $\chi^2_{1-\alpha}$ is the $1-\alpha$ quantile of the chi-squared distribution with one degree of freedom, ${Z}_{>c}=(\mathbbm{1}\{Z_{11}>c\}, \dots, \mathbbm{1}\{Z_{In_{I}}>c\})$ is the random threshold indicator vector, and ${z}_{>c}=(\mathbbm{1}\{z_{11}>c\}, \dots, \mathbbm{1}\{z_{In_{I}}>c\})$ is the observed threshold indicator vector, and we have  
\begin{equation*}
\begin{aligned}
    E_{{u}}({Z}_{>c}{r}_0^{\prime T}) &= \sum_{i=1}^I \mu_i, \text{where } \mu_i = \sum_{\pi_{i}}\Big[ \frac{\exp\{\gamma ({z}_{i\pi_{i}}{u}_i^{T})\}}{\sum_{\widetilde{\pi}_{i}}\exp\{\gamma ({z}_{i\widetilde{\pi}_{i}}{u}_i^{T})\}}\times \sum_{j=1}^{n_{i}}\mathbbm{1}\{z_{i\pi_{i}(j)}>c\}r_{ij}(0)^{\prime}\Big];\\ \text{var}_{{u}}({Z}_{>c}{r}_0^{\prime T}) &= \sum_{i=1}^I \nu_i, \text{where }\nu_i = \sum_{\pi_{i}}\Big[ \frac{\exp\{\gamma ({z}_{i\pi_{i}}{u}_i^{T})\}}{\sum_{\widetilde{\pi}_{i}}\exp\{\gamma ({z}_{i\widetilde{\pi}_{i}}{u}_i^{T})\}}\times \big\{\sum_{j=1}^{n_{i}}\mathbbm{1}\{z_{i\pi_{i}(j)}>c\}r_{ij}(0)^{\prime}\big\}^2\Big]-\mu_i^2.
\end{aligned}
\end{equation*}
Note that in Section~\ref{subsec: Finite-population-exact sensitivity analysis via probabilistic lattice theory}, we showed that for testing $H_{0}^{\prime}: {r}_{0}={r}_{0}^{\prime}$ with the test statistics $\sum_{i=1}^{I}\sum_{j=1}^{n_{i}} \mathbbm{1}\{Z_{ij}>c\}r_{ij}(0)^{\prime}$, the corresponding worst-case $p$-value will be taken when each $u_{ij}= r_{ij}(0)^{\prime}$ or when each $u_{ij}= 1-r_{ij}(0)^{\prime}$. Assuming independence between matched sets, we only need to consider each matched set separately. For matched set $i$, there are up to $2^{n_i}$ possibilities for vector ${r}_{i}^{\prime}=(r_{i1}(0)^{\prime},\dots, r_{in_{i}}(0)^{\prime})$ that are compatible with the observed data. There may be less than $2^{n_i}$ possibilities for $r_{i}^{\prime}$ since some can be safely ruled out based on the observed data as described in Section~\ref{subsec: Generalizing attributable effects to the continuous exposure case}. Let $l_i\leq 2^{n_i}$ be the number of possibilities of $r_{i}^{\prime}$ that are compatible with the observed data. Next, we introduce binary decision variables $d_{ik}$ for each of the $l_i$ possible values of ${r}_{i}^{\prime}$, $k = 1,\ldots, l_i$. More concretely, $d_{ik} = 1$ indicates that we take ${r}_{i}^{\prime}$ to be the $k$th of the $l_i$ possibilities; so $d_{ik}\in \{0, 1\}$ can only equal one for one of the $k \in \{1,\ldots,l_i\}$. For each $k$-th possibility of the $l_i$ possibilities of ${r}_{i}^{\prime}$, there is a corresponding value of the pivot $\sum_{j=1}^{n_{i}}\mathbbm{1}\{z_{ij}>c\}r_{ij}(0)^{\prime}$ (denoted as $t_{ik}$), a corresponding expectation of $\sum_{j=1}^{n_{i}}\mathbbm{1}\{Z_{ij}>c\}r_{ij}(0)^{\prime}$ when each $u_{ij}= r_{ij}(0)^{\prime}$ and that when each $u_{ij}= 1-r_{ij}(0)^{\prime}$ (denoted as $E^{\text{upp}}_{ik}$ and $E^{\text{low}}_{ik}$ respectively), and a corresponding variance of $\sum_{j=1}^{n_{i}}\mathbbm{1}\{Z_{ij}>c\}r_{ij}(0)^{\prime}$ when each $u_{ij}= r_{ij}(0)^{\prime}$ and that when each $u_{ij}= 1-r_{ij}(0)^{\prime}$  (denoted as $V^{\text{upp}}_{ik}$ and $V^{\text{low}}_{ik}$ respectively). Since we use chi-squared distribution as the asymptotic reference distribution under the null, the minimizing ${u}$ in (\ref{reject}) will depend on the relationship between ${z}_{>c}{r}_0^{\prime T}$ and $E_{{u}}({Z}_{>c}{r}_0^{\prime T})$. We impose directional penalties in a similar fashion as in \citet{fogarty_composite} to ensure that we may reject the null only when the observed ${z}_{>c}{r}_0^{\prime T}$ has the correct sign relationship with the upper or lower expectations. We encode one-sided test restrictions using chi-squared distribution by replacing $\chi^2_{1-\alpha}$ with $\chi^2_{1-2\alpha}$. Putting these arguments together, the hypothesis testing problem for testing $\text{TAE}=\Delta$ in a sensitivity analysis under sensitivity parameter $\Gamma$ and level $\alpha$ is equivalent to solving the following mixed-integer quadratically constrained linear program:
\begin{equation}
\label{opt_procedure}
\begin{aligned}
 &\min_{y, d_{ik}, b^{\text{low}}, b^{\text{upp}}} \ \ y \\
 &\textrm{subject to }   \ \ \   ({t}{d}^{T} - {E^{\text{low}}}{d}^{T})^2 - \chi^2_{1-2\alpha} \times {V^{\text{low}}}{d}^{T} \leq y + Mb^{\text{low}}, \\ 
 & \quad \quad \quad \quad \quad ({t}{d}^{T}- {E^{\text{upp}}}{d}^{T})^2 - \chi^2_{1-2\alpha} \times {V^{\text{upp}}}{d}^{T} \leq y + Mb^{\text{upp}}, \\ 
 & \quad \quad \quad \quad \quad M(1-b^{\text{low}}) \leq {t}{d}^{T} - {E^{\text{low}}}{d}^{T} \leq Mb^{\text{low}}, \\
 & \quad \quad \quad \quad \quad Mb^{\text{upp}} \leq {t}{d}^{T} - {E^{\text{upp}}}{d}^{T} \leq M(1-b^{\text{upp}}),\\ 
 &\quad \quad \quad \quad \quad \sum_{i=1}^{I}\sum_{j=1}^{n_{i}} \mathbbm{1}\{Z_{ij}>c\} R_{ij} -{t}{d}^{T}= \Delta ,\\
 & \quad \quad \quad \quad \quad \sum_{k=1}^{l_i} d_{ik} = 1, \ d_{ik} \in \{0,1\} \ \text{for all $i, k$}; \ b^{\text{low}},b^{\text{upp}} \in \{0,1\}; \ y \in \mathbb{R},  
\end{aligned}
\end{equation}
where $y$ is an auxiliary variable, ${t}=(t_{11},\dots, t_{Il_{I}})$, $d=(d_{11},\dots, d_{Il_{I}})$, ${E^{\text{low}}}=(E_{11}^{\text{low}}, \dots, E_{Il_{I}}^{\text{low}})$, ${E^{\text{upp}}}=(E_{11}^{\text{upp}}, \dots, E_{Il_{I}}^{\text{upp}})$, ${V^{\text{low}}}=(V_{11}^{\text{low}}, \dots, V_{Il_{I}}^{\text{low}})$, ${V^{\text{upp}}}=(V_{11}^{\text{upp}}, \dots, V_{Il_{I}}^{\text{upp}})$, and the directional penalty $M>0$ is a prespecified sufficiently large constant. Essentially, based on the normal approximation and the Rosenbaum sensitivity analysis model (\ref{eqn: dose assignment after matching}) with some prespecified sensitivity parameter $\Gamma=\exp(\gamma)$, the mixed-integer program (\ref{opt_procedure}) seeks to find an allocation of $r_0^{\prime}$ (corresponding to the decision variables $d_{ik}$) that is ``least likely" to reject the $H_0^{\prime}: {r}_0 = {r}_0'$ at level $\alpha$, among all possible allocations of $r_0'$ that satisfy the constraint $\text{TAE}=\Delta$. Therefore, an optimal value (of program (\ref{opt_procedure})) $y^* < 0$ indicates that in a sensitivity analysis under the fixed $\Gamma$ and level $\alpha$, we fail to reject $\text{TAE}=\Delta$ as there exists some allocation of $r_{0}^{\prime}$ that satisfies $\text{TAE}=\Delta$ but the $H_0^{\prime}: {r}_0 = {r}_0'$ cannot be rejected based on the observed data. In practice, for certain values of $\Delta$ that do not constrain much of the search space of the mixed-integer program, finding the exact optimal value $y^{*}$ can be computationally expensive for large datasets. Fortunately, solving the exact $y^{*}$ is unnecessary for deciding whether to reject $\text{TAE}=\Delta$; simply removing the objective function of $(\ref{opt_procedure})$ and adding the constraint $y < 0$ and checking if the mixed-integer program is feasible is sufficient. Moreover, if computation time remains an issue, we can simply relax the binary constraint on the $d_{ik}$ to allow them to be any real numbers between 0 and 1. This will offer a statistically conservative but computationally efficient testing procedure. Then we can easily construct a valid $100(1-\alpha)\%$ confidence set for the TAE by solving (\ref{opt_procedure}) or its relaxed version for each $\Delta \in \mathcal{D}$ (where $\mathcal{D}$ is restricted by the observed data and/or monotonicity assumptions). Call the optimal value of this problem $y_{\Delta}^*$. Then an $100(1-\alpha)\%$ confidence set is $\{\Delta \in \mathcal{D} \mid y_{\Delta}^* \leq 0 \}$.  In practice, instead of solving the corresponding mixed-integer program for many different values of $\Delta$, we can simply replace the $ ``= \Delta"$ in (\ref{opt_procedure}) with $ ``\leq \Delta"$ or $ ``\geq \Delta"$, which will automatically yield confidence sets that are intervals. Moreover, in a study where each matched set has a binary contribution to the TAE under the monotonicity assumption, a procedure relying on asymptotic separability analogous to that proposed in \citet{rosenbaum_attributable_obs} and \citet{asymp_sep} is sufficient, and solving the mixed-integer program (\ref{opt_procedure}) is not necessary for this particular setting; see Appendix D in the online supplementary material.

%In this setting, the mixed integer program imposing monotonicity is equivalent to the procedure relying on asymptotic separability, but in contrast, the mixed integer program (\ref{opt_procedure}) \emph{can} be initialized with or without monotonicity encoded as an assumption, at the cost of computational speed and potentially statistical conservatism. 
\begin{remark}
Although Assumption~\ref{mon} (i.e., the monotonicity assumption) has been adopted in many previous studies and is plausible in many settings, it is a strong assumption and may not hold in many other settings. When Assumption~\ref{mon} does not hold, TAE may not be a suitable estimand as its scientific interpretation is not as clear as when Assumption~\ref{mon} holds. Specifically, recall that under Assumption~\ref{mon}, TAE is the number of units that experienced an exposure dose above some scientifically meaningful threshold $c$ and experienced the event of interest but would not have had they been completely unexposed. In this case, TAE is a natural generalization of the classic attributable effect in the binary exposure case, which also assumes a corresponding monotonicity assumption (i.e., $r_{ij}(1)\geq r_{ij}(0)$ for all $i, j$) \citep{rosenbaum_attributable, rosenbaum_attributable_obs}. Instead, without Assumption~\ref{mon}, this interpretation of ATE no longer holds since there could be the case that the outcome $R_{ij} = 0$ under some positive exposure dose above the threshold $c$ but the potential outcome under the zero dose $r_{ij}(0) = 1$. It is a meaningful future research direction to develop a more appropriate causal estimand for matched observational studies with continuous exposures when Assumption~\ref{mon} (the monotonicity assumption) does not hold. 
\end{remark}

\section{Computational Hardness of Sensitivity Analysis With Continuous Outcomes}\label{sec: hardness of SA for continuous outcomes}

%Corollary~\ref{stoch_dom} shows that the worst-case $p$-value can be found by taking the unmeasured confounder for each subject as their observed binary outcome variable. Thus, one might hope that the same solution will apply to the continuous treatment continuous outcome case. 
In the previous sections, we proposed two methods for sensitivity analysis for matched observational studies with continuous exposures and binary outcomes: an explicit analytical approach (for the sharp null) and an implicit optimization approach (for threshold attributable effects). In this section, we show that when both the exposure and the outcome are continuous, the sensitivity analysis is generally computationally infeasible, except for some special cases (e.g., pair matching).

We first show that the classic analytical approach for sensitivity analysis with continuous outcomes cannot be generalized from the binary to the continuous exposure case. Specifically, in the binary exposure literature, \citet{rosen_krieger} showed that for any test statistics of the form $T = Zq^{T}=\sum_{i=1}^{I}\sum_{j=1}^{n_{i}}Z_{ij}q_{ij}$ where each $q_{ij} = q_{ij}({R})$ is some arbitrary function of the outcomes, the worst-case $p$-value (\ref{eqn: worst-case p-value}) under Fisher's sharp null $H_{0}$ is taken at some normalized unmeasured confounders ${u}^{*}=(u_{11}^{*},\dots, u_{In_{I}}^{*})\in [0,1]^{N}$ that satisfy the following two conditions: (C1) in each matched set $i$, the $(u_{i1}^{*},\dots, u_{in_{i}}^{*})$ must have the same ranks as those of $(q_{i1},\dots, q_{in_{i}})$; and (C2) the ${u}^{*}$ must lie at the edge of the cube $[0,1]^{N}$, that is, $u_{ij}^{*}\in \{0,1\}$. This seminal work forms the foundation of many analytical solutions to the sensitivity analysis in the binary exposure case (\citealp{rosenbaum_obs}; Chapter 4). One might hope that principles (C1) and (C2) could be extended to the case of continuous exposures and outcomes. Although the principle (C1) remains true with continuous exposures, we found that the principle (C2) may not hold. To see a concrete example, consider the statistic $T=\sum_{n=1}^{5}Z_{n}q_{n}$ with observed dose vector $(Z_{1},\dots, Z_{5})=(0.1,0.44,0.54,0.73,0.8)$ and $(q_{1},\dots, q_{5})=(1.5,1.5,3,4.5,4.5)$. Then one can numerically verify that the worst-case one-sided $p$-value $\text{pr}({T \geq 9.03})$ under Rosenbaum's sensitivity model (\ref{eqn: Rosenbaum exposure dose model}) with $\gamma = 2$ is taken at $(u_{1}^{*},\dots, u_{5}^{*})=(0, 0, 0.9483617, 1, 1)$, in which $u_{3}^{*}\neq 0$ or $1$. Therefore, except for some special cases (e.g., pair matching), an explicit analytical solution to solving the worst-case $p$-value with continuous outcomes is generally infeasible in the continuous exposure case.

Without an explicit analytical solution to the worst-case $p$-value (\ref{eqn: worst-case p-value}) when both the exposure and outcome are continuous, an alternative approach would be to use an optimization approach to numerically solve the worst-case $p$-value to conduct a sensitivity analysis. As before, let $p_{i\pi_{i}}=\text{pr}({Z}_i={z}_{i\pi_{i}}|\mathcal{F},\mathcal{Z}_{i})$ denote the exposure dose assignment probability for ${z}_{i\pi_{i}} \in \mathcal{Z}_{i}=\{{z}_{i\pi_{i}} \mid \pi_{i} \in S_{n_i}\}$ within matched set $i$ given the sensitivity parameter $\Gamma$, as defined in (\ref{eqn: dose assignment after matching}). Consider a test statistic with the form $T=\sum_{i=1}^{I}q_{i}({Z}_{i}, {R})$ for testing Fisher's sharp null $H_{0}$, in which ${Z}_{i}=(Z_{i1},\dots, Z_{in_{i}})$. Given the observed value $t$ of $T$, we define $\zeta_{{u}}
=(t-\mu_{{u}})^2-\chi^2_{1-\alpha}\times \sigma^2_{{u}}$, where $\mu_{{u}}=\sum_{i=1}^I \sum_{\pi_{i}} p_{i\pi_{i}}q_{i}({z}_{i\pi_{i}}, {R})$ denote the expectation of $T$ given ${u}$ and $\sigma^2_{{u}}=\sum_{i=1}^I \sum_{\pi_{i}} p_{i\pi_{i}} q_{i}^2({z}_{i\pi_{i}}, {R})- \sum_{i=1}^I \{\sum_{\pi_{i}} p_{i\pi_{i}} q_{i}({z}_{i\pi_{i}}, {R})\}^2$ the variance of $T$ given ${u}$. Note that all the $q_{i}({z}_{i\pi_{i}}, {R})$ are fixed values under the sharp null $H_{0}$. In a sensitivity analysis, we reject the sharp null $H_{0}$ if and only if the worst-case $p$-value (\ref{eqn: worst-case p-value}) is above the significance level $\alpha$. Invoking the finite-population central limit theorem (\citealp{li_ding_finite}), this is equivalent to checking if the optimal value of the following optimization problem is greater than zero or not:
\begin{equation}\label{opt: original form with continuous outcomes}
    \text{min}_{{u}} \ \zeta_{{u}} \quad  \text{subject to } \\ {u}\in [0,1]^{N},
\end{equation}
where ${u}=({u}_1,\ldots,{u}_I)=(u_{11},\dots, u_{In_{I}})$. An optimal value of the optimization problem (\ref{opt: original form with continuous outcomes}) less than or equal to zero means that we fail to reject $H_{0}$ in a sensitivity analysis given the sensitivity parameter $\Gamma=\exp(\gamma)$. To investigate the computational complexity of solving (\ref{opt: original form with continuous outcomes}), we consider the following reparametrization: for each matched set $i$, letting $z_{i(1)}\leq z_{i(2)}\leq \dots \leq z_{i(n_{i})}$ be the order statistics of the doses, we define $w_{ijk}=\exp(\gamma z_{i(j)}u_{ik})$ for $j,k=1,...,n_i$ and $s_i=\sum_{{z}_{i\pi_{i}}\in \mathcal{Z}_{i}}\exp(\gamma{z}_{i\pi_{i}} {u}_i^{T})$. Since $\zeta_{{u}}$ is a quadratic function of $p_{i\pi_{i}}$, we can also denote $\zeta_{{u}}$ as $\zeta_{{p}}$, where ${p}=\{p_{i\pi_{i}}: i=1\dots, I, \pi_{i} \in S_{n_{i}}\}$. Then, we can rewrite the optimization problem (\ref{opt: original form with continuous outcomes}) as the following optimization problem:
 \begin{equation}
 \label{eqn: signomial optimization problem}
\begin{aligned}
&\min_{p_{i\pi_{i}}, s_{i}, w_{ijk}} \ \ \ \zeta_{{p}} \\
& \textrm{subject to }   \quad p_{i\pi_{i}}\times s_i =\prod_{k=1}^{n_i}w_{i\pi_{i}(k)k}, \ i=1,...,I; \ \pi_{i} \in S_{n_i} \\
& \quad \quad \quad \quad \quad 
s_i=\sum_{\pi_{i} \in S_{n_i}} \prod_{k=1}^{n_i}w_{i\pi_{i}(k)k}, \ i=1,...,I; \ \pi_{i} \in S_{n_i} \\  
& \quad \quad \quad \quad \quad w_{ijk}=w_{i1k}^{z_{i(j)}/z_{i(1)}}, \ i=1,...,I; \ j,k=1,...,n_i  \\
& \quad \quad \quad \quad \quad  1 \leq w_{i1k}\leq \exp(\gamma z_{i(1)}). \ i=1,...,I; \ k=1,...,n_i 
\end{aligned}
\end{equation}
The problem (\ref{eqn: signomial optimization problem}) is a signomial program, which is computationally intractable even for small sample sizes (\citealp{boyd2007tutorial}), suggesting computational infeasibility for conducting an asymptotically exact sensitivity analysis in matched observational studies when both the exposure and the outcome are continuous. Therefore, a meaningful future research direction is to find a statistically conservative but computationally feasible sensitivity analysis method in this case, possibly through deriving a computationally efficient algorithm for finding an informative lower bound for the optimization problem (\ref{eqn: signomial optimization problem}).

\section{Application: A matched cohort study on the effect of early-life lead exposure on the incidence of committing offensive behaviors}

We apply our new methods to a matched cohort study on the effect of early-life lead exposure on later-life juvenile delinquency (offensive behaviors). We look at the binary outcome of whether the study unit (teen) committed a later-life offense worthy of complaint. We match study units with similar measured pre-exposure confounders using the \texttt{nbpfull} package from (\citealp{bo_continuous}), which resulted in 2007 matched sets with 4134 study units (see Appendix E for details). We conduct two-step confounder balance diagnostics to assess the matching quality for the measured confounders. In the first step, we conduct a preliminary balance diagnosis for each unmeasured confounder by looking at the corresponding standardized mean difference between high-dose (above the median dose of the corresponding matched set) and low-dose (below the median dose of the corresponding matched set) groups \citep{bo_continuous} and conducting the corresponding Kolmogorov-Smirnov test for equality of confounder distribution between high-dose and low-dose groups. The results are summarized in Table~\ref{tab: balance} in Appendix E in the online supplementary materials, which shows that for the measured confounders, the absolute standardized mean differences (SMDs) are all less than or equal to $0.01$ (with most of these absolute SMDs being nearly zero) and $p$-values reported by the Kolmogorov-Smirnov test are all greater than 0.8 (with most of these $p$-values being nearly one). In the second step, we further test the uniform exposure dose assignment assumption assuming no unmeasured confounding (i.e., assumption (\ref{eqn: random assignment assumption}) in Section~\ref{subsec: notations for perfect matching}) through generalizing the existing randomization-based balance tests in the binary exposure case \citep{gagnon_cpt, branson2021randomization, biased_randomization} to the continuous exposure case. The detailed procedure is described in Appendix G of the online supplementary material. After adopting this randomization-based balance test to our matched dataset, we failed to reject the assumption (\ref{eqn: random assignment assumption}) based on measured confounder information under the significance level $0.1$. Putting the above results together, the measured confounders are sufficiently balanced in the matched dataset, and a remaining major concern is how unmeasured confounding bias would affect our inference results, which can be addressed using the newly developed sensitivity analysis methods in this work. Specifically, in Table~\ref{tab:outcome}, for various sensitivity parameters $\Gamma = \exp(\gamma)$, we reported: 1) the worst-case $p$-value under Fisher's sharp null $H_{0}$ and 2) the worst-case 90\% confidence interval (CI) for the threshold attributable effect (TAE) with the dose threshold being set at the CDC marked exposure threshold of 3.5 $\mu$g/dl. Since there were 25 cases (i.e., outcomes being one) in the matched dataset with exposure above this CDC threshold, the upper bound for the TAE is 25. The test statistic used for testing the sharp null of no effect $H_{0}$ and conducting inference for the TAE is $\sum_{i=1}^I \sum_{j=1}^{n_i} \mathbbm{1}\{Z_{ij} > 3.5\}R_{ij}$, which is the number of cases of committing a juvenile offense among the study units who were exposed to a lead level higher than the CDC marked level of 3.5 $\mu$g/dl. We find that under the no unmeasured confounding assumption ($\Gamma=1$), the $p$-value under the sharp null $H_{0}$ is $0.021$, and the $90\%$ CI of the threshold attributable effect (TAE) has a lower bound of $5$. As $\Gamma$ increases, we find that: 1) the worst-case $p$-value under $H_{0}$ increases, crossing the $0.05$ significance level at around $\Gamma = 1.25$ and the $0.1$ significance level at around $\Gamma=1.5$; and 2) the lower bound of the $90\%$ CI of the TAE decreases, equaling zero at $\Gamma > 1.5$. Since $\Gamma = 1.5$ corresponds to a non-trivial but not very large strength of unmeasured confounding (\citealp{rosenbaum2009amplification}), the causal conclusion (assuming no unmeasured confounding) that early-life lead exposure increases the risk of committing serious offensive behaviors may not be very insensitive to potential unmeasured confounding and should be further carefully investigated by domain scientists.

\begin{table}[ht]
\centering
\caption{Sensitivity analysis for the sharp null $H_{0}$ and threshold attributable effect (TAE).}
\begin{tabular}{cccccc}
\hline
 & $\Gamma=1.000$ & $\Gamma=1.125$ & $\Gamma=1.250$  & $\Gamma=1.375$ & $\Gamma=1.500$ \\ 
  \hline 
 Worst-case $p$-value for $H_{0}$  & 0.021 & 0.034 & 0.052 & 0.074 & 0.100
 \\
 Worst-case 90\% CI for TAE   & $[5,25]$ &  $[4,25]$ & $[3,25]$  & $[2,25]$  &  $[1,25]$ \\
   \hline
\end{tabular}
\label{tab:outcome} 
\end{table}

\section{Discussion}
In this paper, we developed sensitivity analysis methods for matched observational studies with continuous exposures and binary outcomes. Our results suggest two future research directions. First, since our design sensitivity formula and power simulations suggest that different test statistics may perform very distinctly in a sensitivity analysis, a future study on developing a better adaptive approach to select the appropriate test statistic in a sensitivity analysis would be useful. A possible strategy is to extend the previous adaptive approaches (\citealp{rosenbaum2012testing}; \citealp{heng2021increasing}) to the continuous exposure case. Second, since our results in Section~\ref{sec: hardness of SA for continuous outcomes} suggest that conducting an exact or asymptotically exact sensitivity analysis when both the exposure and outcome are continuous could be computationally infeasible in general settings, it would be helpful to instead find a computationally feasible way to conduct a conservative sensitivity analysis in such settings (i.e., a trade-off between statistical and computational efficiencies).

\section*{Acknowledgement}
The authors would like to thank the participants of ACIC 2023 for their helpful comments. Jeffrey Zhang was supported in part by the National Institutes of Health (NIH) grant R01HD101415. The work of Dylan Small was supported in part by a grant from the NIH (\#5R01AG065276-02). A grant from the New York University Research Catalyst Prize and a New York University School of Global Public Health Research Support Grant partially supported Siyu Heng's work. 

\section*{Supplementary Material}
Online supplementary materials include technical proofs, additional details, examples, simulations, and discussions.

\bibliographystyle{apalike}
\bibliography{references}

\begin{thebibliography}{}

\bibitem[Baiocchi et~al., 2010]{baiocchi2010building}
Baiocchi, M., Small, D.~S., Lorch, S., and Rosenbaum, P.~R. (2010).
\newblock Building a stronger instrument in an observational study of perinatal care for premature infants.
\newblock {\em Journal of the American Statistical Association}, 105(492):1285--1296.

\bibitem[Bonvini et~al., 2022]{bonvini2022sensitivity}
Bonvini, M., Kennedy, E., Ventura, V., and Wasserman, L. (2022).
\newblock Sensitivity analysis for marginal structural models.

\bibitem[Boyd et~al., 2007]{boyd2007tutorial}
Boyd, S., Kim, S.-J., Vandenberghe, L., and Hassibi, A. (2007).
\newblock A tutorial on geometric programming.
\newblock {\em Optimization and Engineering}, 8:67--127.

\bibitem[Branson, 2021]{branson2021randomization}
Branson, Z. (2021).
\newblock Randomization tests to assess covariate balance when designing and analyzing matched datasets.
\newblock {\em Observational Studies}, 7(2):1--36.

\bibitem[Branson and Keele, 2020]{branson2020evaluating}
Branson, Z. and Keele, L. (2020).
\newblock Evaluating a key instrumental variable assumption using randomization tests.
\newblock {\em American Journal of Epidemiology}, 189(11):1412--1420.

\bibitem[Chen et~al., 2023]{biased_randomization}
Chen, K., Heng, S., Long, Q., and Zhang, B. (2023).
\newblock Testing biased randomization assumptions and quantifying imperfect matching and residual confounding in matched observational studies.
\newblock {\em Journal of Computational and Graphical Statistics}, 32(2):528--538.

\bibitem[Fogarty et~al., 2021]{fogarty2019biased}
Fogarty, C.~B., Lee, K., Kelz, R.~R., and Keele, L.~J. (2021).
\newblock Biased encouragements and heterogeneous effects in an instrumental variable study of emergency general surgical outcomes.
\newblock {\em Journal of the American Statistical Association}, pages 1--12.

\bibitem[Fogarty et~al., 2017]{fogarty_composite}
Fogarty, C.~B., Shi, P., Mikkelsen, M.~E., and Small, D.~S. (2017).
\newblock Randomization inference and sensitivity analysis for composite null hypotheses with binary outcomes in matched observational studies.
\newblock {\em Journal of the American Statistical Association}, 112(517):321--331.

\bibitem[Gagnon-Bartsch and Shem-Tov, 2019]{gagnon_cpt}
Gagnon-Bartsch, J. and Shem-Tov, Y. (2019).
\newblock The classification permutation test: A flexible approach to testing for covariate imbalance in observational studies.
\newblock {\em Annals of Applied Statistics}, 13:1464--1483.

\bibitem[Gastwirth et~al., 1998]{gastwirth1998dual}
Gastwirth, J.~L., Krieger, A.~M., and Rosenbaum, P.~R. (1998).
\newblock Dual and simultaneous sensitivity analysis for matched pairs.
\newblock {\em Biometrika}, 85(4):907--920.

\bibitem[Gastwirth et~al., 2000]{asymp_sep}
Gastwirth, J.~L., Krieger, A.~M., and Rosenbaum, P.~R. (2000).
\newblock Asymptotic separability in sensitivity analysis.
\newblock {\em Journal of the Royal Statistical Society: Series B (Statistical Methodology)}, 62(3):545--555.

\bibitem[Gibson et~al., 2020]{Gibson2020}
Gibson, J.~M., Fisher, M., Clonch, A., MacDonald, J.~M., and Cook, P.~J. (2020).
\newblock {Children drinking private well water have higher blood lead than those with city water}.
\newblock {\em Proceedings of the National Academy of Sciences}, 117(29):16898--16907.

\bibitem[Gibson et~al., 2022]{lead_crime}
Gibson, J.~M., Macdonald, J.~M., Fisher, M., Chen, X., Pawlick, A., and Cook, P.~J. (2022).
\newblock Early life lead exposure from private well water increases juvenile delinquency risk among us teens.
\newblock {\em Proceedings of the National Academy of Sciences}.

\bibitem[Gr{\"a}tzer, 2002]{gratzer2002general}
Gr{\"a}tzer, G. (2002).
\newblock {\em General Lattice Theory}.
\newblock Springer Science \& Business Media.

\bibitem[Hansen and Bowers, 2009]{hansen2009attributing}
Hansen, B.~B. and Bowers, J. (2009).
\newblock Attributing effects to a cluster-randomized get-out-the-vote campaign.
\newblock {\em Journal of the American Statistical Association}, 104(487):873--885.

\bibitem[Hasegawa and Small, 2017]{hasegawa2017sensitivity}
Hasegawa, R. and Small, D. (2017).
\newblock Sensitivity analysis for matched pair analysis of binary data: From worst case to average case analysis.
\newblock {\em Biometrics}, 73(4):1424--1432.

\bibitem[Heng et~al., 2021]{heng2021increasing}
Heng, S., Kang, H., Small, D.~S., and Fogarty, C.~B. (2021).
\newblock Increasing power for observational studies of aberrant response: An adaptive approach.
\newblock {\em Journal of the Royal Statistical Society Series B: (Statistical Methodology)}, 83(3):482--504.

\bibitem[Heng et~al., 2023]{heng2023instrumental}
Heng, S., Zhang, B., Han, X., Lorch, S.~A., and Small, D.~S. (2023).
\newblock Instrumental variables: to strengthen or not to strengthen?
\newblock {\em Journal of the Royal Statistical Society: Series A (Statistics in Society)}, page qnad075.

\bibitem[Hollander et~al., 1977]{hollander1977functions}
Hollander, M., Proschan, F., and Sethuraman, J. (1977).
\newblock Functions decreasing in transposition and their applications in ranking problems.
\newblock {\em The Annals of Statistics}, 5(4):722--733.

\bibitem[Holley, 1974]{holley1974remarks}
Holley, R. (1974).
\newblock Remarks on the fkg inequalities.
\newblock {\em Communications in Mathematical Physics}, 36:227--231.

\bibitem[Imbens and Rubin, 2015]{imbens2015causal}
Imbens, G.~W. and Rubin, D.~B. (2015).
\newblock {\em Causal Inference in Statistics, Social, and Biomedical Sciences}.
\newblock Cambridge University Press.

\bibitem[Li and Ding, 2017]{li_ding_finite}
Li, X. and Ding, P. (2017).
\newblock General forms of finite population central limit theorems with applications to causal inference.
\newblock {\em Journal of the American Statistical Association}, 112(520):1759--1769.

\bibitem[Lu et~al., 2011]{Lu2011OptimalNM}
Lu, B., Greevy, R.~A., Xu, X., and Beck, C. (2011).
\newblock Optimal nonbipartite matching and its statistical applications.
\newblock {\em The American Statistician}, 65:21 -- 30.

\bibitem[Lu et~al., 2001]{Lu_NBP}
Lu, B., Zanutto, E., Hornik, R., and Rosenbaum, P.~R. (2001).
\newblock Matching with doses in an observational study of a media campaign against drug abuse.
\newblock {\em Journal of the American Statistical Association}, 96(456):1245--1253.
\newblock PMID: 25525284.

\bibitem[Pimentel, 2022]{pimentel2022covariateadaptive}
Pimentel, S.~D. (2022).
\newblock Covariate-adaptive randomization inference in matched designs.

\bibitem[Rigdon and Hudgens, 2015]{rand_bin}
Rigdon, J. and Hudgens, M.~G. (2015).
\newblock Randomization inference for treatment effects on a binary outcome.
\newblock {\em Statistics in Medicine}, 34(6):924--935.

\bibitem[Rosenbaum, 1987]{rosenbaum1987sensitivity}
Rosenbaum, P.~R. (1987).
\newblock Sensitivity analysis for certain permutation inferences in matched observational studies.
\newblock {\em Biometrika}, 74(1):13--26.

\bibitem[Rosenbaum, 1989]{rosenbaum1989sensitivity}
Rosenbaum, P.~R. (1989).
\newblock Sensitivity analysis for matched observational studies with many ordered treatments.
\newblock {\em Scandinavian Journal of Statistics}, 16(3):227--236.

\bibitem[Rosenbaum, 2001]{rosenbaum_attributable}
Rosenbaum, P.~R. (2001).
\newblock Effects attributable to treatment: Inference in experiments and observational studies with a discrete pivot.
\newblock {\em Biometrika}, 88:219--231.

\bibitem[Rosenbaum, 2002a]{rosenbaum_attributable_obs}
Rosenbaum, P.~R. (2002a).
\newblock Attributing effects to treatment in matched observational studies.
\newblock {\em Journal of the American Statistical Association}, 97:183--192.

\bibitem[Rosenbaum, 2002b]{rosenbaum_obs}
Rosenbaum, P.~R. (2002b).
\newblock {\em Observational Studies}.
\newblock Springer New York.

\bibitem[Rosenbaum, 2004]{rosenbaum2004design}
Rosenbaum, P.~R. (2004).
\newblock Design sensitivity in observational studies.
\newblock {\em Biometrika}, 91(1):153--164.

\bibitem[Rosenbaum, 2012]{rosenbaum2012testing}
Rosenbaum, P.~R. (2012).
\newblock Testing one hypothesis twice in observational studies.
\newblock {\em Biometrika}, 99(4):763--774.

\bibitem[Rosenbaum, 2020a]{rosenbaum2020design}
Rosenbaum, P.~R. (2020a).
\newblock {\em Design of Observational Studies (Second Edition)}.
\newblock Springer.

\bibitem[Rosenbaum, 2020b]{rosenbaum2020modern}
Rosenbaum, P.~R. (2020b).
\newblock Modern algorithms for matching in observational studies.
\newblock {\em Annual Review of Statistics and Its Application}, 7:143--176.

\bibitem[Rosenbaum and Krieger, 1990]{rosen_krieger}
Rosenbaum, P.~R. and Krieger, A.~M. (1990).
\newblock Sensitivity of two-sample permutation inferences in observational studies.
\newblock {\em Journal of the American Statistical Association}, 85(410):493--498.

\bibitem[Rosenbaum and Silber, 2009]{rosenbaum2009amplification}
Rosenbaum, P.~R. and Silber, J.~H. (2009).
\newblock Amplification of sensitivity analysis in matched observational studies.
\newblock {\em Journal of the American Statistical Association}, 104(488):1398--1405.

\bibitem[Su and Li, 2023]{su2022treatment}
Su, Y. and Li, X. (2023).
\newblock Treatment effect quantiles in stratified randomized experiments and matched observational studies.
\newblock {\em Biometrika}, page in press.

\bibitem[Tan, 2006]{tan_msm}
Tan, Z. (2006).
\newblock A distributional approach for causal inference using propensity scores.
\newblock {\em Journal of the American Statistical Association}, 101(476):1619--1637.

\bibitem[Yu et~al., 2023]{yu2023risk}
Yu, R., Kelz, R., Lorch, S., and Keele, L.~J. (2023).
\newblock The risk of maternal complications after cesarean delivery: Near-far matching for instrumental variables study designs with large observational datasets.
\newblock {\em The Annals of Applied Statistics}, 17(2):1701--1721.

\bibitem[Zhang et~al., 2023a]{zhang2023social}
Zhang, B., Heng, S., Ye, T., and Small, D.~S. (2023a).
\newblock Social distancing and covid-19: Randomization inference for a structured dose-response relationship.
\newblock {\em The Annals of Applied Statistics}, 17(1):23--46.

\bibitem[Zhang et~al., 2023b]{bo_continuous}
Zhang, B., Mackay, E.~J., and Baiocchi, M. (2023b).
\newblock Statistical matching and subclassification with a continuous dose: Characterization, algorithm, and application to a health outcomes study.
\newblock {\em The Annals of Applied Statistics}, 17(1):454--475.

\end{thebibliography}

\clearpage

%% Here are the title, author names and addresses

\begin{center}
    \large \bf Online Supplementary Materials for ``Sensitivity Analysis for Matched Observational Studies with Continuous Exposures and Binary Outcomes"
\end{center}

\section*{Appendix A: Technical Proofs}

\subsection{Proof of Theorem 1}

Before proving Theorem 1, we state and prove the following two useful lemmas. For any $ {s}\in \Omega$, we define $ {s}^{1}=( {s}^{1}_{1}, \dots,  {s}^{1}_{I})$ and let $\text{sum}( {s}^{1})$ denote the summation of all the elements of $ {s}^{1}$.
\begin{lemma}
\label{lemma: identity}
For any $ {s},  {s}^{*}\in \Omega$, we have
\begin{equation*}
   \textnormal{sum}(( {s} \lor  {s}^{*})^1) + \textnormal{sum}(( {s} \land  {s}^{*})^1) = \textnormal{sum}( {s}^1) + \textnormal{sum}( {s}^{*1}). 
\end{equation*}
\end{lemma}
\begin{proof}
Recall that the $i$-th entry of $( {s} \lor  {s}^{*})^1$ is $\max\{ {s}_i^1,  {s}_i^{*1}\}$ and the $i$-th entry of $( {s} \land  {s}^{*})^1$ is $\min\{ {s}_i^1,  {s}_i^{*1}\}$. Then for each $i$, by the fact that $\max\{a,b\} + \min\{a,b\} = a+b$ for any $(a, b)$, we have $( {s} \lor  {s}^{*})_i^1 + ( {s} \land  {s}^{*})_i^1 =  {s}_i^1 +  {s}_i^{*1}$ for all $i$. Summing these identities over $i$ yields the desired result.
\end{proof}

\begin{lemma}
\label{useful2}
For any $ {s},  {s}^{*}\in \Omega$, we have
\begin{equation*}
    ( {s} \land  {s}^{*})^0 \geq  {s}^0,  {s}^{*0} \textit{ element-wise and } ( {s} \lor  {s}^{*})^0 \leq  {s}^0,  {s}^{*0} \textit{ element-wise.}
\end{equation*}
\end{lemma}
\begin{proof}
Note that we only need to prove that for any $ {s},  {s}^{*}\in \Omega$,
if $ {s}^{1} \leq  {s}^{*1}$ element-wise, then $ {s}^{0} \geq  {s}^{*0}$ element-wise (i.e., $ {s}_{i}^{0} \geq  {s}_{i}^{*0}$ element-wise for all $i$). Then the desired result will follow immediately. 

Let $ {d}_{i}=(z_{i(1)}, \dots, z_{i(n_{i})})$ be the ordered exposure doses in the matched set $i$, with $n_{i}$ doses in total. Without loss of generality, we assume no ties with doses and take the doses to be $[n_{i}]=\{1,\ldots,n_{i}\}$. It is clear that $ {s}_{i}$ and $ {s}^{*}_{i}$ are in one-to-one correspondence with sets of indices in $ {d}$, i.e., one group of indices corresponds to doses assigned to subjects with outcome 1 and another group of indices correspond to doses assigned to subjects with outcome 0. Fix some $m_{i} < n_{i}$ to be the number of subjects who experience outcome 1. Let $(s_{i1}^0,\ldots, s_{i(n_{i}-m_{i})}^0)$ and $(s_{i1}^{*0},\ldots, s_{i(n_{i}-m_{i})}^{*0})$ be the indices of subjects who experience outcome 0 in matched set $i$, and $(s_{i(n_{i}-m_{i}+1)}^1,\ldots, s_{in_{i}}^1)$ and $(s_{i(n_{i}-m_{i}+1)}^{*1},\ldots, s_{in_{i}}^{*1})$ the indices of subjects who experience outcome 1 in matched set $i$, ordered from smallest to largest. Suppose that $s_{ij}^1 \leq s_{ij}^{*1}$ for all $j=n_{i}-m_{i}+1,\dots, n_{i}$. We now show that $s_{ij}^0 \geq s_{ij}^{*0}$ for all $j=1,\dots, n_{i}-m_{i}$ by induction.

\textbf{Base cases: } $m_{i} = 0,1$. If $m_{i}=0$, $s_{ij}^0$ and $s_{ij}^{*0}$ are clearly identical for all $j$, so the (in)equality holds. If $m_{i}=1$ and $s_{i1}^1=s_{i1}^{*1}$, then again $s_{ij}^0$ and $s_{ij}^{*0}$ are identical for all $j$, so the (in)equality holds. If $s_{i1}^1<s_{i1}^{*1}$, then $ {s}_{i}^0$ and $ {s}_{i}^{*0}$ differ only in that $s_{i1}^{*1}$ is in but $s_{i1}^1$ is not in $ {s}^0$, and $s_1^{1}$ is in but $s_1^{*1}$ is not in $ {s}^{*0}$. So to change from $ {s}^{*0}$ to $ {s}^{0}$ requires changing one element in $ {s}^{*0}$ with a strictly larger element, so after ordering from smallest to largest, it is clear that $s_{ij}^0 \geq s_{ij}^{*0}$ for all $j$.

\textbf{Inductive Hypothesis:} The statement is true for all fixed $m_{i}=1,\dots, k-1$ (here $m_{i}\leq n_{i}$).

\textbf{Inductive Step:} Suppose that there are $k$ subjects with outcome 1. There are two cases: $s_{i1}^{1} = s_{i1}^{*1} $ (Case 1) and $s_{i1}^{1} < s_{i1}^{*1} $ (Case 2). In Case 1, it is obvious that $s_{ij}^0 = s_{ij}^{*0}$ for $j=1,\dots, s_{i1}^{1}-1$. Next, discarding the elements that are less than or equal to $s_{i1}^{1} = s_{i1}^{*1} $, we have a new stratum with $n-s_{i1}^{1}$ subjects and $k-1$ of them have outcome 1. Invoking the inductive hypothesis, we have $s_{ij}^0 \geq s_{ij}^{*0}$ for the remaining $j$. In Case 2, consider the $ {\wh{s}}_{i}=( {\wh{s}}_{i}^{0},  {\wh{s}}_{i}^{1})$ with $ {\wh{s}}_{i}^1 = (s_{i(n_{i}-k+1)}^1,s_{i(n_{i}-k+2) }^{*1},\ldots,s_{in_{i}}^{*1}) $. Following similar arguments for the base case $m=1$, we have $\wh{s}_{ij}^0 \geq s_{ij}^{*0}$ for all $j$. Since comparing $ {\wh{s}}_i^0$ and $ {s}_i^0$ belongs to Case 1, we have $\wh{s}_{ij}^0 \leq s_{ij}^{0}$ for all $j$. Therefore, we have ${s}_{ij}^{0} \geq s_{ij}^{*0}$ for all $j$.
\end{proof}
Then we are ready to present a proof of Theorem 1. 
\begin{proof}(Theorem 1)
First, since exposure dose assignments are independent across matched sets. We have $\text{pr}( {S}_{ {u}^{+}}= {s}\lor  {s}^{*}) = \prod_{i=1}^I \text{pr}( {S}_{ {u}^{+}i}=( {s}\lor  {s}^{*})_{i})$, $\text{pr}( {S}_{ {u}}=( {s}\land  {s})^{*}) = \prod_{i=1}^I \text{pr}( {S}_{ {u}i}=( {s}\land  {s}^{*})_{i})$, $\text{pr}( {S}_{ {u}^{+}}= {s}) = \prod_{i=1}^I \text{pr}( {S}_{ {u}^{+}i}= {s}_i)$, and $\text{pr}( {S}_{ {u}}= {s}^{*}) = \prod_{i=1}^I \text{pr}( {S}_{ {u}i}= {s}_{i}^{*})$. Thus, it suffices to show the desired inequality for each matched set $i$, i.e., $\text{pr}( {S}_{ {u}^{+}i}=( {s}\lor  {s}^{*})_{i})\times \text{pr}( {S}_{ {u}i}=( {s}\land  {s}^{*})_{i}) \geq \text{pr}( {S}_{ {u}^{+}i}= {s}_i) \times \text{pr}( {S}_{ {u}i}= {s}_i^{*})$. For an unmeasured confounder vector $ {u}_{i}$ corresponding to matched set $i$, let $ {u}_{i,1}$ represent the entries of $ {u}_{i}$ for subjects with outcome 1, and $ {u}_{i, 0}$ represent the entries of $ {u}_{i}$ for subjects with outcome 0. Also, let $ {s}^a_{i,\pi_{i}^a}$ denote the appropriate vector $ {s}^a_{i}$ permuted according to the permutation $\pi^a_{i}$ ($a=0,1$). We can decompose each term as follows:
\begin{equation*}
\begin{aligned}
\text{pr}( {S}_{ {u}^{+}i}=( {s}\lor  {s}^{*})_{i})&= h_{ {u}^{+}i}^{-1}\sum_{\pi_{i}^{0},\pi_{i}^{1}}\exp\{\gamma( {s}\lor  {s}^{*})_{i,\pi_{i}^{0}}^0 ( {u}^{+}_{i,0})^{T}\}\exp\{\gamma( {s}\lor  {s}^{*})_{i,\pi_{i}^{1}}^1 ( {u}^{+}_{i,1})^{T}\}\\
&= h_{ {u}^{+}i}^{-1}\big(m_{i}\big)! \big(n_{i}-m_{i}\big)! \exp\{\gamma\text{sum}(( {s}\lor  {s}^{*})_{i}^1)\}, \\ 
\text{pr}( {S}_{ {u}i}=( {s}\land  {s}^{*})_{i}) &= h_{ {u}i}^{-1}\sum_{\pi_{i}^{0},\pi_{i}^{1}}\exp\{\gamma( {s}\land  {s}^{*})_{i,\pi_{i}^{0}}^0  {u}_{i,0}^{T}\}\exp\{\gamma( {s}\land  {s}^{*})_{i, \pi_{i}^{1}}^1  {u}_{i,1}^{T}\},\\ 
\text{pr}( {S}_{ {u}^{+}i}= {s}_i) &= h_{ {u}^{+}i}^{-1}\sum_{\pi_{i}^{0},\pi_{i}^{1}}\exp\{\gamma {s}_{i, \pi_{i}^{0}}^0 ( {u}^{+}_{i,0})^{T} \} \exp\{\gamma {s}_{i,\pi_{i}^{1}}^1 ( {u}^{+}_{i,1})^{T}\} \\
&=h_{ {u}^{+}i}^{-1}\big(m_{i}\big)! \big(n_{i}-m_{i}\big)!\exp\{\gamma\text{sum}( {s}_{i}^1)\}, \\ 
\text{pr}( {S}_{ {u}i}= {s}_i^{*})&= h_{ {u}i}^{-1}\sum_{\pi_{i}^{0},\pi_{i}^{1}}\exp\{\gamma {s}^{*0}_{i, \pi_{i}^{0}}  {u}_{i,0}^{T}\}\exp\{\gamma {s}^{*1}_{i, \pi_{i}^{1}}  {u}_{i,1}^{T}\},
\end{aligned}
\end{equation*}
where 
\begin{equation*}
h_{ {u}^{+}i}=\sum_{\widetilde{\pi}_{i}\in S_{n_{i}}}\exp\{\gamma( {z}_{i\widetilde{\pi}_{i}} {u}_{i}^{+T})\}, \quad h_{ {u}i}=\sum_{\widetilde{\pi}_{i}\in S_{n_{i}}}\exp\{\gamma( {z}_{i\widetilde{\pi}_{i}} {u}_{i}^{T})\}.
\end{equation*}
Therefore, it suffices to show that for each fixed $\pi_{i}^{0}$, $\pi_{i}^{1}$, we have 
\begin{equation}
\label{eqn:holley_condition}
\begin{aligned}
\exp\{\gamma\text{sum}(( {s}\lor  {s}^{*})_{i}^1)\} &\times \exp\{\gamma( {s}\land  {s}^{*})_{i,\pi_{i}^{0}}^0 {u}_{i,0}^{T}\}\exp\{\gamma( {s}\land  {s}^{*})_{i,\pi_{i}^{1}}^1 {u}_{i,1}^{T}\} \\ & \geq \exp\{\gamma\text{sum}( {s}_{i}^1)\} \times \exp\{\gamma(  {s}^{*})_{i,\pi_{i}^{0}}^0 {u}_{i,0}^{T}\}\exp\{\gamma(  {s}^{*})_{i,\pi_{i}^{1}}^1 {u}_{i,1}^{T}\}.
\end{aligned}
\end{equation}
Since all entries of $ {u}_{i,1}$ are in $[0,1]$, we have $$( {s}^{*})_{i,\pi_{i}^{1}}^1  {u}_{i,1}^{T}-( {s}\land  {s}^{*})_{i,\pi_{i}^{1}}^1 {u}_{i,1}^{T} \leq \text{sum}(( {s}^{*})_{i,\pi_{i}^{1}}^1)-\text{sum}( {s}\land  {s}^{*1}_{i,\pi_{i}^{1}}).$$
Lemma~\ref{lemma: identity} and its proof imply that $$\text{sum}((  {s}^{*})_{i,\pi_{i}^{1}}^1)-\text{sum}(( {s}\land  {s}^{*})_{i,\pi_{i}^{1}}^1)=\text{sum}(( {s}\lor  {s}^{*})_{i,\pi_{i}^{1}}^1)- \text{sum}(( {s})_{i,\pi_{i}^{1}}^1),$$
which in turn implies that 
$$\text{sum}(( {s}\lor  {s}^{*})_{i,\pi_{i}^{1}}^1) + ( {s}\land  {s}^{*})_{i,\pi_{i}^{1}}^1  {u}_{i,1}^{T} \geq \text{sum}( {s}_{i,\pi_{i}^{1}}^1)+ {s}^{*1}_{i,\pi_{i}^{1}}  {u}_{i, 1}^{T}.$$ By Lemma~\ref{useful2}, the element-wise inequality $( {s}\land  {s}^{*})_{i,\pi_{i}^{0}}^0 \geq ( {s}^{*})_{i,\pi_{i}^{0}}^0$ holds. This implies that $( {s}\land  {s}^{*})_{i,\pi_{i}^{0}}^0  {u}_{i,0}^{T} \geq ( {s}^{*})_{i,\pi_{i}^{0}}^0 {u}_{i,0}^{T}$ since the elements of $ {u}_{i, 0}$ are nonnegative. Therefore, we have shown the desired inequality (\ref{eqn:holley_condition}). Applying this argument to each $\pi_{i}^{0}, \pi_{i}^{1}$ and each matched set $i$, we get the desired result.
\end{proof}

\subsection{Proof of Theorem 2}

Our proof can be viewed as a generalization of the corresponding proof in \citet{su2022treatment} to the continuous dose case. We first state a version of the Lindeberg-Feller central limit theorem (see also Lemma A5 in \citet{su2022treatment}):
\begin{lemma}
For each $n$, let $\xi_{n, m}, 1 \leq m \leq n$, be independent random variables with $E \xi_{n, m}=0$. Suppose that
\begin{enumerate}[(i)]
    \item $\sum_{m=1}^n E \xi_{n, m}^2 \rightarrow \sigma^2>0$
    \item  For all $c>0, \lim _{n \rightarrow \infty} \sum_{m=1}^n E\left[\xi_{n, m}^2 \mathbbm{1}\left\{\left|\xi_{n, m}\right|>c\right\}\right]=0$.
\end{enumerate}
Then $\Xi_n=\xi_{n, 1}+\ldots+\xi_{n, n} \stackrel{d}{\longrightarrow} N 
\left(0, \sigma^2\right)$ as $n \rightarrow \infty$.
\end{lemma}

Using notations close to those defined in Section 3 of the main manuscript, we first denote $f_i( {S}_{i})$ as the contribution to the test statistic from matched set $i$ (so the full test statistic is $T = f( {S})=\sum_{i=1}^{I}f_{i}( {S}_{i})$), where $ {S}_{i}$ represents the (random) allocation of the doses to subjects with outcome 1 and that with outcome 0, in matched set $i$. We let $ {s}_{i}$ denote some realization of the random variable $ {S}_{i}$. Let $ {s}_{i}^+$ and $ {s}_{i}^-$ denote the allocations that maximize and minimize $f_i$, respectively, i.e., we have $f_i( {s}_{i}^-) \leq f_i( {s}_{i}) \leq f_i( {s}_{i}^+)$ for any allocation $ {s}_{i}$. Let $W_i = f_i( {s}_{i}^+) - f_i( {s}_{i}^-)$, $\mu_i = E_{\Gamma}[f_i( {S}_{i})]$, and $\nu_i^2 = \text{Var}_{\Gamma}[f_i( {S}_{i})]$ (where $\Gamma=\exp(\gamma)$). Assuming independence across matched sets and $f = \sum_{i=1}^I f_i$, we have $\mu = E_{\Gamma}[f] = \sum_{i=1}^I \mu_i$ and $\sigma^2 = \text{Var}_{\Gamma}[f] = \sum_{i=1}^I \nu_i^2$. Let $m_i$ be the number of subjects with outcome $1$ out of $n_i$ total subjects in matched set $i$. Define $Y_i = (f_i( {S}_{i})-\mu_i)/\sigma$, which has mean zero and $\sum_{i=1}^I E_{\Gamma}(Y_i^2) = 1$. Consider the following condition:
\begin{condition}
\label{normal_reg}
Let $$l_{i}=\frac{1}{1+({n_i \choose m_i}-1)\exp\{\gamma(\sum_{j = \lceil n_i/2\rceil+1}^{n_i}z_{i(j)}-\sum_{j =1}^{\lfloor n_i/2\rfloor}z_{i(j)})\}} .$$ As $I \to \infty$, we have $$\max_{1 \leq i \leq I}\frac{W_i^2}{\sum_{i=1}^I l_{i}^3W_i^2} \to 0.$$
\end{condition}
Note that Condition~\ref{normal_reg} will hold when the size $n_i$ of a matched set is bounded, the test statistic contributed by each matched set is bounded from above and below by some fixed constants, and the number of discordant matched sets (i.e., not all doses are the same and not all the binary outcomes are the same) goes to infinity.
\begin{lemma}\label{lemma: asymptotic normality}
Let $ {r}\in \{0,1\}^{N}$ be a fixed binary potential outcome vector and $f( {S}) = T( {Z}, {r})$. Assume that Condition \ref{normal_reg} holds, then as $I\rightarrow \infty$, we have
$$\frac{f( {S})-E_{\Gamma}[f( {S})]}{\textnormal{Var}_{\Gamma}[f( {S})]} \stackrel{d}{\to} N(0,1).$$
\end{lemma}
\begin{proof}
We have $f_i( {s}_{i}^-) \leq \mu_i \leq f_i( {s}_{i}^+)$, which implies $|f_i( {s}_{i}) -\mu_i| \leq f_i( {s}_{i}^+) - f_i( {s}_{i}^-) = W_i$ for any $ {s}\in \Omega$. Let $P_{ {u}, \Gamma}(\cdot)$ be the probability distribution for the dose assignment under the Rosenbaum sensitivity analysis model with given $\Gamma$. It follows that 
\begin{align*}
     \nu_i^2 &= E_{ {u}, \Gamma}[(f_i( {S}_{i})-\mu_i)^2]\\
     &=\sum_{ {s}_{i}} P_{ {u}, \Gamma}( {S}_{i} =  {s}_{i})\{f_i( {s}_{i})-\mu_i\}^2 \\
     &\geq p_i^-\{f_i( {s}_{i}^-)-\mu_i\}^2+p_i^+\{f_i( {s}_{i}^+)-\mu_i\}^2,
\end{align*}
where $p_i^- = P_{ {u}, \Gamma}( {S}_{i} =  {s}_{i}^-)$ and $p_i^+ = P_{ {u}, \Gamma}( {S}_{i} =  {s}_{i}^+)$. Moreover, we have
\begin{equation*}
    p_i^+ f_i( {s}_{i}^+) + (1-p_i^+)f_i( {s}_{i}^-) \leq \mu_i \leq  p_i^- f_i( {s}_{i}^-) + (1-p_i^-)f_i( {s}_{i}^+),
\end{equation*}
which implies that
\begin{equation*}
    \mu_i-f_i( {s}_{i}^-) \geq p_i^+\left\{f_i( {s}_{i}^+)-f_i( {s}_{i}^-)\right\}=p_i^+ W_i,
\end{equation*}
and
\begin{equation*}
     f_i( {s}_{i}^+)-\mu_i \geq p_i^-\left\{f_i( {s}_{i}^+)-f_i( {s}_{i}^-)\right\}=p_i^- W_i .
\end{equation*}
Therefore, we have
\begin{align*}
    \nu_i^2 &\geq p_i^-\{f_i( {s}_{i}^-)-\mu_i\}^2+p_i^+\{f_i( {s}_{i}^+)-\mu_i\}^2 \\
    &\geq p_i^- p_i^{+2} W_i^2+p_{i}^+ p_{i}^{-2} W_i^2\\
    &=p_i^- p_i^{+}\left(p_i^- +p_i^{+}\right) W_i^2.
\end{align*}
Since $u_{ij}\in [0,1]$, we have 
 \begin{align*}
     \max_{ {s}_{i},  {s}_{i}^{\prime},  {u}}\frac{\text{pr}_{ {u}, \Gamma}( {S}_{i} =  {s}_{i})}{\text{pr}_{ {u}, \Gamma}( {S}_{i}=  {s}_{i}^{\prime})}&= \frac{\sum_{\pi_{i}^{0},\pi_{i}^{1}}\exp\{\gamma {s}^{0}_{i, \pi_{i}^{0}}  {u}_{i,0}^{T}\}\exp\{\gamma {s}^{1}_{i, \pi_{i}^{1}}  {u}_{i,1}^{T}\}}{\sum_{\pi_{i}^{0},\pi_{i}^{1}}\exp\{\gamma {s}^{\prime 0}_{i, \pi_{i}^{0}}  {u}_{i,0}^{T}\}\exp\{\gamma {s}^{\prime 1}_{i, \pi_{i}^{1}}  {u}_{i,1}^{T}\}} \\
     &\leq \max_{\pi_{i} \pi'_{i},  {u} } \exp\{\gamma \sum_{j} (z_{i\pi_{i}(j)}-z_{i\pi'_{i}(j)})u_{ij}\} \\
     &\leq  \exp\Big\{\gamma\Big(\sum_{j = \lceil n_i/2\rceil+1}^{n_i}z_{i(j)}-\sum_{j =1}^{\lfloor n_i/2\rfloor}z_{i(j)}\Big)\Big\}.
 \end{align*}Therefore, we have  
\begin{equation*}
    p_i^+, \ p_i^- \geq \frac{1}{1+({n_i \choose m_i}-1)\exp\{\gamma(\sum_{j = \lceil n_i/2\rceil+1}^{n_i}z_{i(j)}-\sum_{j =1}^{\lfloor n_i/2\rfloor}z_{i(j)})\}} = l_{i},
\end{equation*}
where $z_{i(j)}$ are the order statistics of the doses in matched set $i$ in increasing order. This immediately implies that 
$\nu_i^2 \geq 2l_{i}^3W_i^2.$ Therefore, we have $$Y_i^2 =\frac{(f_i( {S}_{i})-\mu_i)^2}{\sigma^2}= \frac{(f_i( {S}_{i})-\mu_i)^2}{\sum_{i=1}^I\nu_i^2} \leq \frac{W_i^2}{\sum_{i=1}^I2l_{i}^3W_i^2}.$$
Then under Condition \ref{normal_reg}, as $I \to \infty$, we have
\begin{equation*}
    \max_{1 \leq i \leq I} Y_i^2 \to 0.
\end{equation*}
This implies that, for any $\epsilon>0$, we have $\sum_{i=1}^IE_{\Gamma}\left[Y_i^2 \mathbbm{1}\left\{\left|Y_i\right|>\epsilon\right\}\right] \to 0$. Invoking the Lindeberg-Feller central limit theorem, we proved the desired result.
\end{proof}

Then Theorem 2 can be easily proved as follows:
\begin{proof}
(Proof of Theorem 2) By Theorem 1, we have 
\begin{equation*}
    \max_{ {u}\in [0,1]^{N} }\text{pr}(T\geq t \mid \mathcal{F}, \mathcal{Z})=\text{pr}_{ {u}^{+}}(T\geq t \mid \mathcal{F}, \mathcal{Z}),
\end{equation*}
in which $ {u}^{+}= {R}$ (note that the observed outcome vector $ {R}$ is fixed under Fisher's sharp null $H_{0}$). Assuming that Condition~\ref{normal_reg} holds, by Lemma~\ref{lemma: asymptotic normality} (which holds for any $ {u}$ under Condition~\ref{normal_reg}), we have 
\begin{equation*}
    \text{pr}_{ {u}^{+}}(T\geq t \mid \mathcal{F}, \mathcal{Z})\doteq 1-\Phi\big[\{t-E_{ {u}^+}(T)\} / \sqrt{\text{var}_{ {u}^+}(T)}\big] \ \text{as $I\rightarrow \infty$}.
\end{equation*}

\end{proof}

\subsection{More details on the validity of asymptotic $p$-values }
In this section, we further verify the validity of the $p$-value obtained by the normal approximation in Theorem 2, which can also be viewed as an extension of the arguments in \citet{su2022treatment} to the continuous exposure case. Let $\Tilde{\mu}_i$ denote the expectation of $f_i( {S}_{i})$ under the Rosenbaum exposure dose model under some fixed $\gamma$ and $ {u}$ equal to the potential outcome vector $ {r}$ (which equals the observed outcome vector $ {R}$ under Fisher's sharp null $H_{0}$), and $\Tilde{\nu}_{i}$ be the corresponding variance of $f_i( {S}_{i})$ under this setting. Accordingly, let $\Tilde{\mu} = \sum_{i=1}^I \Tilde{\mu}_i$ and $\Tilde{\sigma}^2 = \sum_{i=1}^I \Tilde{\nu}_{i}^2$. Let $\mathcal{A} = \{i: \nu_i^2 > \Tilde{\nu}_{i}^2, 1 \leq i \leq I\}$. Next, we let $\Delta_{\Gamma}(\mu) = |\mathcal{A}|^{-1} \sum_{i \in \mathcal{A}}(\Tilde{\mu}_i-\mu_i)$ and  $\Delta_{\Gamma}(\sigma^{2}) = |\mathcal{A}|^{-1} \sum_{i \in \mathcal{A}}(\nu_i^2-\Tilde{\nu}_{i}^2)$. 
\begin{condition}
\label{norm_reg2}
As $I \to \infty$, we have $\frac{\Delta_{\Gamma}(\mu)}{\Delta_{\Gamma}(\sigma^{2})}\sqrt{2\sum_{i=1}^I l_{i}^3 W_i^2} \to \infty$.
\end{condition}
\begin{lemma}\label{lemma: asymptotic seperability}
Under Condition \ref{norm_reg2}, for any $k \geq 0$, there exists some $N(k)$ such that for any $I > N(k)$, we have $\Tilde{\mu}+k\Tilde{\sigma} \geq \mu+k\sigma$. 
\end{lemma}
\begin{proof}
First, note that we have $\Tilde{\mu}_i \geq \mu_i$ for all $i$ by the stochastic dominating property (shown in Section 3.1 of the main text) of taking $ {u}$ to be the potential outcome vector $ {r}$ (i.e., the observed outcome vector $ {R}$ under Fisher's sharp null $H_{0}$). We have

\begin{equation*}
\begin{aligned}
\Tilde{\mu}-\mu &= \sum_{i=1}^I (\Tilde{\mu}_i - \mu_i) \geq \sum_{i \in \mathcal{A}} (\Tilde{\mu}_i - \mu_i), \\ 
\sigma - \Tilde{\sigma} &= \sqrt{\sum_{i=1}^I\nu_i^2} - \sqrt{\sum_{i=1}^I\Tilde{\nu}_{i}^2} \\
&= \sqrt{\sum_{i \in \mathcal{A}}\nu_i^2 + \sum_{i \notin \mathcal{A}}\nu_i^2} - \sqrt{\sum_{i\in \mathcal{A}}\Tilde{\nu}_{i}^2+ \sum_{i \notin \mathcal{A}}\Tilde{\nu}_{i}^2} \\ &\leq \sqrt{\sum_{i \in \mathcal{A}}\nu_i^2 + \sum_{i \notin \mathcal{A}}\Tilde{\nu}_{i}^2} - \sqrt{\sum_{i\in \mathcal{A}}\Tilde{\nu}_{i}^2+ \sum_{i \notin \mathcal{A}}\Tilde{\nu}_{i}^2} \\ &\leq \frac{1}{2\Tilde{\sigma}}\sum_{i \in \mathcal{A}}(\nu_i^2 - \Tilde{\nu}_{i}^2)\\
&=|\mathcal{A}|\frac{\Delta_{\Gamma}(\sigma^{2})}{2\Tilde{\sigma}},
\end{aligned}
\end{equation*}
where the last inequality can be obtained by taking $a = 1$, $x' = \sum_{i\in \mathcal{A}}\nu_{i}^2$, $x = \sum_{i\in \mathcal{A}}\Tilde{\nu}_{i}^2$, and $b = \sum_{i \notin \mathcal{A}}\Tilde{\nu}_{i}^2$ in Lemma 1 in \citet{asymp_sep}. So for any $k \geq 0$, we have 
\begin{equation*}
\begin{aligned}
(\Tilde{\mu}+k\Tilde{\sigma}) - (\mu+k\sigma) &= (\Tilde{\mu}-\mu) -k(\sigma-\Tilde{\sigma}) \\ & \geq \sum_{i \in \mathcal{A}} (\Tilde{\mu}-\mu) - \frac{k}{2\Tilde{\sigma}} \sum_{i \in \mathcal{A}}(\nu_i^2 - \Tilde{\nu}_{i}^2) \\ &= |\mathcal{A}|\frac{\Delta_{\Gamma}(\sigma^{2})}{\Tilde{\sigma}}\left(\frac{\Delta_{\Gamma}(\mu)}{\Delta_{\Gamma}(\sigma^{2})}\Tilde{\sigma}- \frac{k}{2} \right).
\end{aligned}
\end{equation*}
From the proof of Lemma \ref{normal_reg}, we know that $\tilde{\sigma}^2 \geq 2\sum_{i=1}^I l_{i}^3 W_i^2 $, so we have 
\begin{equation}\label{eqn: lemma 5}
    (\Tilde{\mu}+k\Tilde{\sigma}) - (\mu+k\sigma) \geq |\mathcal{A}|\frac{\Delta_{\Gamma}(\sigma^{2})}{\Tilde{\sigma}}\left(\frac{\Delta_{\Gamma}(\mu)}{\Delta_{\Gamma}(\sigma^{2})}\sqrt{2\sum_{i=1}^I l_{i}^3 W_i^2}- \frac{k}{2} \right).
\end{equation}
Condition \ref{norm_reg2} ensures that the right-hand side of inequality (\ref{eqn: lemma 5}) is nonnegative when $I > N(k)$ for some sufficiently large $N(k)$, so the desired conclusion follows.
\end{proof}

We now verify that the normal approximation in Theorem 2 leads to asymptotically valid $p$-values under the proposed regularity conditions.
\begin{proposition}
Let $T=T( {Z}, {R})$ be a conditionally isotonic test statistic, where $ {Z}\in \mathbbm{R}^{N}$ and $ {R}\in \{0,1\}^{N}$. Let $\tilde{p} = 1-\Phi((T-\tilde{\mu})/\tilde{\sigma})$. Under Conditions \ref{normal_reg} and \ref{norm_reg2} and the sharp null $H_{0}$, we have
\begin{equation*}
    \limsup_{I \to \infty} \text{pr}(\tilde{p} \leq \alpha) \leq \alpha.
\end{equation*}
\end{proposition}
\begin{proof}
 Let $T=f( {S})$ be a conditionally isotonic test statistic, and let $\mu$ and $\sigma$ be the expectation and standard deviation of $T$ under $H_{0}$, respectively. By Lemma~\ref{lemma: asymptotic seperability}, we know that any $k \geq 0$, there is some $N(k)$ such that for $I > N(k)$, we have $\Tilde{\mu}+k\Tilde{\sigma} \geq \mu+k\sigma$. Then $\text{pr}(T \geq \Tilde{\mu}+k\Tilde{\sigma}) \leq \text{pr}(T \geq \mu+k\sigma)$ for $I > N(k)$, which implies $\limsup_{I \to \infty}\text{pr}(T \geq \Tilde{\mu}+k\Tilde{\sigma}) \leq \limsup_{I \to \infty}\text{pr}(T \geq \mu+k\sigma)$. By Lemma~\ref{lemma: asymptotic normality}, we have $(T-\mu)/\sigma \stackrel{d}{\to} N(0,1)$. Therefore, we have
    \begin{equation*}
    \begin{aligned}
        \limsup_{I \to \infty} \text{pr}(\tilde{p} \leq \alpha) &= \limsup_{I \to \infty} \text{pr}(1-\Phi((T-\Tilde{\mu})/\Tilde{\sigma}) \leq \alpha) \\
        &= \limsup_{I \to \infty} \text{pr}((T-\widetilde{\mu})/\widetilde{\sigma} \geq \Phi^{-1}(1-\alpha) ) \\ 
        &\leq  \limsup_{I \to \infty} \text{pr}((T-\mu)/\sigma \geq \Phi^{-1}(1-\alpha) ) \\ 
       &=1-\Phi(\Phi^{-1}(1-\alpha)) = \alpha.
    \end{aligned}
    \end{equation*}
\end{proof}

\subsection{Proof of Theorem 3}
\begin{lemma}\label{lemma: double inequalities}
Consider a general class of stratum-wise isotonic test statistics with the form $T=\sum_{i=1}^{I}q_{i}( {Z}_i, {R}_i) = \sum_{i=1}^{I}\sum_{j=1}^{n_i} m(Z_{ij}) \times R_{ij}$, where $m$ is some arbitrary bounded and monotonically non-decreasing function. Suppose that $(Z_{i1}, \ldots,  Z_{in_i}, R_{i1}, \ldots,  R_{in_i}, n_i)$ are independent and identically distributed realizations from some multivariate distribution $F$, where $2 \leq n_i \leq M$ is an integer random variable, $Z_{ij}\in \mathbbm{R}$, and $R_{ij} \in \{0,1\}$. Suppose that, conditional on $n_i$, the $(Z_{i1},R_{i1}), \ldots, (Z_{in_i},R_{i_{n_i}})$ are pairwise independent and identically distributed, and $0 < \text{cor}(m(Z_{ij}),R_{ij}) < 1$. Then we have 
\begin{equation}\label{inequality: double less}
    E\left[\max_{\pi_{i} }  q_{i}( {Z}_{i\pi_{i} }, {R}_{i})\right]>E\left[ q_{i}( {Z}_{i}, {R}_{i})\right] > E\left[\sum_{\pi_{i} } \frac{1}{n_i!}\times q_{i}( {Z}_{i\pi_{i} }, {R}_{i})\right].
\end{equation}
\end{lemma}
\begin{proof}
We first show the second inequality in (\ref{inequality: double less}). We will show that for any permutation $\pi_{i}$ that is not the identity permutation, we have $$E\left[ q_{i}( {Z}_{i}, {R}_{i})\right] > E\left[q_{i}( {Z}_{i\pi_{i} }, {R}_{i})\right].$$ Then by linearity of expectation, $E\left[\sum_{\pi_{i} } \frac{1}{n_i!}\times q_{i}( {Z}_{i\pi_{i}}, {R}_{i})\right] = \sum_{\pi_{i}} \frac{1}{n_i!}\times E\left[q_{i}( {Z}_{i\pi_{i}}, {R}_{i})\right]$, which is a weighted sum of $E\left[ q_{i}( {Z}_{i}, {R}_{i})\right]$ and elements smaller than it, immediately proving the second inequality. The justification is as follows: for a permutation that is not identity, there exist multiple indices $j$ such that $\pi_{i}(j) \neq j$. We next observe that for all such $j$, $E[m(Z_{i\pi_{i}(j)})R_{ij}]=E[m(Z_{i\pi_{i}(j)})]E[R_{ij}]$ by the pairwise independence assumption. Meanwhile, by the strictly positive correlation condition, we know that $E[m(Z_{ij})R_{ij}]>E[m(Z_{ij})]E[R_{ij}]$ and $E[m(Z_{ij})]E[R_{ij}]=E[m(Z_{i\pi_{i}(j)})]E[R_{ij}]$ by the identical distribution assumption. Thus, we have $E[m(Z_{ij})R_{ij}]>E[m(Z_{i\pi_{i}(j)})R_{ij}]$ for $\pi_{i}(j)\neq j$. By the linearity of expectation and the $q_{i}=\sum_{j=1}^{n_{i}}m(Z_{ij})R_{ij}$, it immediately follows that $E\left[ q_{i}( {Z}_{i}, {R}_{i})\right] > E\left[q_{i}( {Z}_{i\pi_{i}}, {R}_{i})\right].$ For the first inequality in (\ref{inequality: double less}), the permutation that maximizes $\max_{\pi_{i}}  q_{i}( {Z}_{i\pi_{i}}, {R}_{i})$ is the permutation that aligns the largest $Z_{ij}$ with the $R_{ij}$. A lack of perfect correlation $\text{cor}(m(Z_{ij}),R_{ij}) < 1$ implies a positive probability that the ranks of $Z_{ij}$ do not align with the ranks of $R_{ij}$. In other words, for a realization of $(Z_{i1}, \ldots,  Z_{in_i}, R_{i1}, \ldots,  R_{in_i}, n_i)$, the probability of the event $\{\exists j \neq k \mid m(Z_{ij}) < m(Z_{ik}) \text{ and } R_{ij} = 1, R_{ik} = 0\}$ is strictly greater than 0. Under such an event, it is clear that the statistic will be strictly increased by permuting $Z_{ij}$ and $Z_{ik}$, i.e., we have $\max_{\pi_{i} }  q_{i}( {Z}_{i\pi_{i} }, {R}_{i}) > q_{i}( {Z}_{i }, {R}_{i})$. Since the inequality $\max_{\pi_{i} }  q_{i}( {Z}_{i\pi_{i} }, {R}_{i}) \geq q_{i}( {Z}_{i }, {R}_{i})$ always holds and $P(\max_{\pi_{i} }  q_{i}( {Z}_{i\pi_{i} }, {R}_{i}) > q_{i}( {Z}_{i }, {R}_{i}))>0$, we have $E\left[\max_{\pi_{i} }  q_{i}( {Z}_{i\pi_{i} }, {R}_{i})\right]>E\left[ q_{i}( {Z}_{i}, {R}_{i})\right]$.
\end{proof}
\begin{lemma}
\label{lem:unique}
 Define 
 \begin{equation*}
     \phi(\gamma) = E\left[\sum_{\pi_{i}\in S_{n_{i}} }\Big \{ \frac{\exp(\gamma ( {Z}_{i\pi_{i}} {R}_i^{T}))}{\sum_{\widetilde{\pi}_{i}\in S_{n_{i}} }\exp(\gamma ( {Z}_{i\widetilde{\pi}_{i}} {R}_i^{T}))}\times q_{i}( {Z}_{i\pi_{i}},  {R}_{i}) \Big\} \right]. 
 \end{equation*}
 Under the assumptions in Lemma~\ref{lemma: double inequalities}, we have
 \begin{enumerate}[(i)]
     \item $\phi(\gamma)$ is monotonically increasing for $\gamma \geq 0$. 
     \item $\lim_{\gamma \to 0^+} \phi(\gamma) < E\left[ q_{i}( {Z}_{i}, {R}_{i})\right]$.
     \item  $\lim_{\gamma \to \infty} \phi(\gamma) > E\left[ q_{i}( {Z}_{i}, {R}_{i})\right]$.
 \end{enumerate}
\end{lemma}
\begin{proof}
(i) It suffices to show that $E\left[\sum_{\pi_{i} } \frac{\exp(\gamma ( {Z}_{i\pi_{i}} {R}_i^{T}))}{\sum_{\Tilde{\pi}_{i} }\exp(\gamma ( {Z}_{i\Tilde{\pi}_{i} } {R}_i^{T}))}\times T( {Z}_{i\pi_{i}}, {R}_{i}) \mid  {R}_{i} =  {r}_{i},\mathcal{Z}_{i} \right]=\sum_{\pi_{i}} \frac{\exp(\gamma ( {z}_{i\pi_{i}} {r}_i^{T}))}{\sum_{\Tilde{\pi}_{i} }\exp(\gamma ( {z}_{i\Tilde{\pi}_{i}} {r}_i^{T}))}\times T( {z}_{i\pi_{i}}, {r}_{i})$ is monotonically increasing in $\gamma\geq 0$, for any $ {r}_{i} \in \{0,1\}^{n_i}$. This quantity can be viewed as the expectation of a random variable $f( {S})$ where $ {S}$ takes values in $\Omega$ as described previously in Section 3. Next, fix $\gamma > \gamma'$. Using the same notation as before, let $ {S}$ be distributed according to randomization distribution given by $p_{i\pi_{i}} =  \frac{\exp(\gamma ( {z}_{i\pi_{i}} {r}_i^{T}))}{\sum_{\Tilde{\pi}_{i} }\exp(\gamma ( {z}_{i\Tilde{\pi}_{i} } {r}_i^{T}))}$ and let $\widetilde{ {S}}$ be distributed according to randomization distribution given by $\tilde{p}_{i\pi_{i}} =  \frac{\exp(\gamma' ( {z}_{i\pi_{i}} {r}_i^{T}))}{\sum_{\Tilde{\pi}_{i}}\exp(\gamma' ( {z}_{i\Tilde{\pi}_{i}} {r}_i^{T}))}$, $i=1,\dots, I$. We will again apply the Holley inequality to show that for any isotonic function $f$, $E\{f( {S})\} \geq E\{f(\widetilde{ {S}})\}$. As before, we can decompose each term as follows:
\begin{align*}
    \text{pr}( {S}= {s}\lor  {s}^*)&=\prod_{i=1}^{I}\text{pr}( {S}_{i}=( {s}\lor  {s}^*)_{i}) \\
    &=\prod_{i=1}^{I}\Big[h_{ {r}_{i}}^{-1} \sum_{\pi_{i}^0,\pi_{i}^1}\exp(\gamma( {s}\lor  {s}^*)_{i, \pi_{i}^0}^0  {0}^{T})\exp(\gamma( {s}\lor  {s}^*)_{i,\pi_{i}^1}^1  {1}^{T})\Big]\\
&= \prod_{i=1}^{I}\Big[h_{ {r}_{i}}^{-1}\big(\sum_{j=1}^{n_{i}}r_{ij}\big)!\big(n_{i}-\sum_{j=1}^{n_{i}}r_{ij}\big)! \exp\{\gamma\text{sum}(( {s}\lor  {s}^*)_{i}^1)\}\Big],
\end{align*}
\begin{align*}
    \text{pr}(\Tilde{ {S}}= {s}\land  {s}^*)&=\prod_{i=1}^{I}\text{pr}(\Tilde{ {S}}_{i}=( {s}\land  {s}^*)_{i}) \\
    &=\prod_{i=1}^{I} \Big[h_{ {r}_{i}}^{\prime -1} \sum_{\pi_{i}^0,\pi_{i}^1}\exp(\gamma^{\prime}( {s}\land  {s}^*)_{i, \pi_{i}^0}^0  {0}^{T})\exp(\gamma^{\prime}( {s}\land  {s}^*)_{i,\pi_{i}^1}^1  {1}^{T})\Big]\\
&= \prod_{i=1}^{I}\Big[h_{ {r}_{i}}^{\prime-1}\big(\sum_{j=1}^{n_{i}}r_{ij}\big)!\big(n_{i}-\sum_{j=1}^{n_{i}}r_{ij}\big)! \exp\{\gamma^{\prime}\text{sum}(( {s}\land  {s}^*)_{i}^1)\}\Big],
\end{align*}
\begin{align*}
    \text{pr}( {S}= {s})&=\prod_{i=1}^{I}\text{pr}( {S}_{i}= {s}_{i}) \\
    &= \prod_{i=1}^{I}\Big[h_{ {r}_{i}}^{-1}\sum_{\pi_{i}^0,\pi_{i}^1}\exp(\gamma {s}_{i,\pi_{i}^0}^0  {0}^{T})\exp(\gamma {s}_{i,\pi_{i}^1}^1  {1}^{T})\Big]\\
    &= \prod_{i=1}^{I}\Big[h_{ {r}_{i}}^{-1}\big(\sum_{j=1}^{n_{i}}r_{ij}\big)!\big(n_{i}-\sum_{j=1}^{n_{i}}r_{ij}\big)! \exp\{\gamma\text{sum}( {s}_{i}^1)\}\Big],
\end{align*}
\begin{align*}
    \text{pr}(\tilde{ {S}}= {s}^*)&=\prod_{i=1}^{I}\text{pr}(\tilde{ {S}}_{i}= {s}^*_{i})\\
    &= \prod_{i=1}^{I} \Big[ h_{ {r}_{i}}^{\prime -1}\sum_{\pi_{i}^0,\pi_{i}^1}\exp(\gamma' {s}^{*0}_{i,\pi_{i}^0} {0}^{T})\exp(\gamma' {s}^{*1}_{i,\pi_{i}^1}  {1}^{T})\Big] \\
    &= \prod_{i=1}^{I} \Big[h_{ {r}_{i}}^{\prime -1}\big(\sum_{j=1}^{n_{i}}r_{ij}\big)!\big(n_{i}-\sum_{j=1}^{n_{i}}r_{ij}\big)!  \exp\{\gamma'\text{sum}( {s}_{i}^{*1})\}\Big],
\end{align*}
where we have 
\begin{equation*}
   h_{ {r}_{i}}= \sum_{\Tilde{\pi}_{i}}\exp(\gamma ( {z}_{i\Tilde{\pi}_{i}} {r}_i^{T})), \quad h_{ {r}_{i}}^{\prime}= \sum_{\Tilde{\pi}_{i}}\exp(\gamma' ( {z}_{i\Tilde{\pi}_{i}} {r}_i^{T})).
\end{equation*}
Therefore, to show that $\text{pr}( {S}= {s}\lor  {s}^*)\text{pr}(\Tilde{ {S}}= {s}\land  {s}^*)\geq \text{pr}( {S}= {s})\text{pr}(\tilde{ {S}}= {s}^*)$, it suffices to show that 
\begin{equation*}
    \exp\{\gamma \text{sum}(( {s}\lor  {s}^*)_{i}^1)\} \times  \exp\{\gamma'\text{sum}(( {s}\land  {s}^*)_{i}^1)\} \geq \exp\{\gamma\text{sum}( {s}^1_{i})\} \times \exp\{\gamma'\text{sum}( {s}^{*1}_{i})\}.
\end{equation*}
Using Lemma~\ref{lemma: identity} and the facts that $\gamma > \gamma'$ and $\text{sum}(( {s}\lor  {s}^*)_{i}^1) \geq \text{sum}( {s}_{i}^1)$, we have
\begin{equation*}
\begin{aligned}
&\quad \ \exp\{\gamma \text{sum}(( {s}\lor  {s}^*)_{i}^1)\} \times  \exp\{\gamma'\text{sum}(( {s}\land  {s}^*)_{i}^1)\} \\
&= \exp\{(\gamma-\gamma')\text{sum}(( {s}\lor  {s}^*)_{i}^1)\}\exp\{\gamma'\text{sum}( {s}_{i}^1)+\gamma'\text{sum}( {s}^{*1}_{i})\} \\ &\geq \exp\{(\gamma-\gamma')(\text{sum}( {s}^1_{i}))\}\exp\{\gamma'\text{sum}( {s}^1_{i})+\gamma'\text{sum}( {s}^{*1}_{i})\} \\ &= \exp\{\gamma\text{sum}( {s}^1_{i})\} \times \exp\{\gamma'\text{sum}( {s}^{*1}_{i})\}.
\end{aligned} 
\end{equation*}
Therefore, we have shown the desired condition in the Holley inequality.

(ii) Since $\sum_{\pi_{i}}\Big\{ \frac{\exp(\gamma ( {Z}_{i\pi_{i}} {R}_i^{T}))}{\sum_{\Tilde{\pi}_{i}}\exp(\gamma ( {Z}_{i\Tilde{\pi}_{i}} {R}_i^{T}))}\times q_{i}( {Z}_{i\pi_{i}}, {R}_{i})\Big\}$ is bounded, by the bounded convergence theorem and Lemma~\ref{lemma: double inequalities}, we have
\begin{equation*}
\begin{aligned}
    \lim_{\gamma \to 0^+} \phi(\gamma) &= \lim_{\gamma \to 0^+} E\left[ \sum_{\pi_{i}}\Big\{ \frac{\exp(\gamma ( {Z}_{i\pi_{i}} {R}_i^{T}))}{\sum_{\Tilde{\pi}_{i}}\exp(\gamma ( {Z}_{i\Tilde{\pi}_{i}} {R}_i^{T}))}\times q_{i}( {Z}_{i\pi_{i}}, {R}_{i})\Big\} \right] \\ &=  E\left[\lim_{\gamma \to 0^+} \sum_{\pi_{i}}\Big\{ \frac{\exp(\gamma ( {Z}_{i\pi_{i}} {R}_i^{T}))}{\sum_{\Tilde{\pi}_{i}}\exp(\gamma ( {Z}_{i\Tilde{\pi}_{i}} {R}_i^{T}))}\times q_{i}( {Z}_{i\pi_{i}}, {R}_{i})\Big\}\right] \\  &=  E\left[\sum_{\pi_{i}} \frac{1}{n_i!}\times q_{i}( {Z}_{i\pi_{i}}, {R}_{i})\right] \\ &<E\left[ q_{i}( {Z}_{i}, {R}_{i})\right],
\end{aligned}
\end{equation*}
where the last line follows by Lemma~\ref{lemma: double inequalities}.

(iii) Again, by the bounded convergence theorem, we have
\begin{equation*}
\begin{aligned}
    \lim_{\gamma \to \infty} \phi(\gamma) &= \lim_{\gamma \to \infty} E\left[ \sum_{\pi_{i}}\Big\{ \frac{\exp(\gamma ( {Z}_{i\pi_{i}} {R}_i^{T}))}{\sum_{\Tilde{\pi}_{i}}\exp(\gamma ( {Z}_{i\Tilde{\pi}_{i}} {R}_i^{T}))}\times q_{i}( {Z}_{i\pi_{i}}, {R}_{i})\Big\}\right] \\ &=  E\left[\lim_{\gamma \to \infty} \sum_{\pi_{i}}\Big\{ \frac{\exp(\gamma ( {Z}_{i\pi_{i}} {R}_i^{T}))}{\sum_{\Tilde{\pi}_{i}}\exp(\gamma ( {Z}_{i\Tilde{\pi}_{i}} {R}_i^{T}))}\times q_{i}( {Z}_{i\pi_{i}}, {R}_{i})\Big\}\right] \\ &=  E\left[ E\left[\lim_{\gamma \to \infty} \sum_{\pi_{i}}\Big\{ \frac{\exp(\gamma ( {Z}_{i\pi_{i}} {R}_i^{T}))}{\sum_{\Tilde{\pi}_{i}}\exp(\gamma ( {Z}_{i\Tilde{\pi}_{i}} {R}_i^{T}))}\times q_{i}( {Z}_{i\pi_{i}}, {R}_{i})\Big\}\mid  {R}_i= {r}_{i}, \mathcal{Z}_i \right]\right] \\  &=  E\left[\max_{\pi_{i}} q_{i}( {Z}_{i\pi_{i}}, {R}_{i})\right] \\ &> E\left[ q_{i}( {Z}_{i}, {R}_{i})\right],
\end{aligned}
\end{equation*}
where the last line follows by Lemma~\ref{lemma: double inequalities}. The fourth line follows from the following argument: we can again view $\sum_{\pi_{i}}\Big\{ \frac{\exp(\gamma ( {Z}_{i\pi_{i}} {R}_i^{T}))}{\sum_{\Tilde{\pi}_{i}}\exp(\gamma ( {Z}_{i\Tilde{\pi}_{i}} {R}_i^{T}))}\times q_{i}( {Z}_{i\pi_{i}}, {R}_{i})\Big\}$ as the expectation of a random variable $f( {S})$ where $ {S}$ takes values in $\Omega$. Let $ {S}$ be distributed according to $p_{i\pi_{i}} = \frac{\exp(\gamma ( {z}_{i\pi_{i}} {r}_i^{T}))}{\sum_{\Tilde{\pi}_{i}}\exp(\gamma ( {z}_{i\Tilde{\pi}_{i}} {r}_i^{T}))}$. For the maximal element $ {s}_{\max}$ in the lattice, we have
\begin{align*}
    \text{pr}( {S}= {s}_{\max})&=\prod_{i=1}^{I}\text{pr}( {S}_{i}= {s}_{\max,i})\\
    &=\prod_{i=1}^{I}\sum_{\pi_{i}^0,\pi_{i}^1}\exp(\gamma {s}_{\max i, \pi_{i}^0}^0  {0}^{T})\exp(\gamma {s}_{\max i, \pi_{i}^1}^1 {1}^{T})\\
    &= \prod_{i=1}^{I} \Big[\big(\sum_{j=1}^{n_{i}}r_{ij}\big)!\big(n_{i}-\sum_{j=1}^{n_{i}}r_{ij}\big)! \exp\{\gamma\text{sum}( {s}_{\max i}^{1})\}\Big]. 
\end{align*}
The distributive probabilistic lattice $\Omega$ is unique when the exposures are continuous and distinct (for handling ties of doses, see Section 2.10.3 in \citet{rosenbaum_obs}), and for any other $ {s}$ that is not $ {s}_{\max}$, we have $\text{pr}( {S}= {s}_{\max})/\text{pr}( {S}= {s}) = \prod_{i=1}^{I}\big(\sum_{j=1}^{n_{i}}r_{ij}\big)!\big(n_{i}-\sum_{j=1}^{n_{i}}r_{ij}\big)!\exp\{\gamma(\text{sum}( {s}_{\max, i}^{1})-\text{sum}( {s}^{1}_{i}))\} > 1$ and goes to infinity as $\gamma \to \infty$. Thus, as $\gamma \to \infty$, all of the probability mass concentrates on the maximal element $ {s}_{\max}$ in the lattice, which also corresponds to the largest value of the test statistic $f( {S})$ by the isotonic property of $f$.
\end{proof}

%\begin{theorem}
%Suppose that $(Z_{i1}, \ldots,  Z_{in_i}, R_{i1}, \ldots,  R_{in_i}, n_i)$ are i.i.d. realizations from some multivariate distribution $F$, where $2 \leq n_i \leq M$ is an integer random variable and $R_{ij} \in \{0,1\}$. Suppressing conditioning on $n_i$, suppose that larger dose $Z$ induces a larger outcome $R$ for each subject in the sense that $E\left[ T( {Z}_{i}, {R}_{i})\right] > E\left[\sum_{\pi_{i}} \frac{1}{n_i!}\times T( {Z}_{i\pi_{i}}, {R}_{i})\right]$. If $t$ is isotonic as defined in \ref{def:isotonic} as well as bounded, then the equation
%\begin{equation*}
%E\left[ T( {Z}_{i}, {R}_{i})\right] = E\left[\sum_{\pi_{i}} \frac{\exp(\gamma ( {Z}_{i\pi_{i}}' {R}_i))}{\sum_{\pi_{i}}\exp(\gamma ( {Z}_{i\pi_{i}}' {R}_i))}\times T( {Z}_{i\pi_{i}}, {R}_{i})\right]
%\end{equation*}
%has a unique solution $\gamma^* > 0$, and $\widetilde{\Gamma} = \exp(\gamma^*)$ is the design sensitivity.
%\end{theorem}
We now prove the design sensitivity formula stated in Theorem 3. 
\begin{proof}(Proof of Theorem 3) We define
\begin{align*}
    &E_{ {u}^{+}}(T)= \sum_{i=1}^I \sum_{\pi_{i}} \big\{\frac{\exp(\gamma ( {Z}_{i\pi_{i}} {R}_i^{T}))}{\sum_{\Tilde{\pi}_{i}}\exp(\gamma ( {Z}_{i\Tilde{\pi}_{i}} {R}_i^{T}))}\times q_{i}( {Z}_{i\pi_{i}}, {R}_{i})\big\},\\
    &\text{var}_{ {u}^{+}}(T)=\sum_{i=1}^I \sum_{\pi_{i}} \big\{\frac{\exp(\gamma ( {Z}_{i\pi_{i}} {R}_i^{T}))}{\sum_{\Tilde{\pi}_{i}}\exp(\gamma ( {Z}_{i\Tilde{\pi}_{i}} {R}_i^{T}))}\times q_{i}^{2}( {Z}_{i\pi_{i}}, {R}_{i})\big\}\\
    &\quad \quad \quad \quad \quad \quad \quad \quad \quad \quad -\sum_{i=1}^I \Big[\sum_{\pi_{i}} \big\{\frac{\exp(\gamma ( {Z}_{i\pi_{i}} {R}_i^{T}))}{\sum_{\Tilde{\pi}_{i}}\exp(\gamma ( {Z}_{i\Tilde{\pi}_{i}} {R}_i^{T}))}\times q_{i}( {Z}_{i\pi_{i}}, {R}_{i})\big\}\Big]^{2}.
\end{align*}
By the law of large numbers, we have as $I\rightarrow \infty$,
\begin{equation*}
\begin{aligned}
    &I^{-1}\sum_{i=1}^{I}q_{i}( {Z}_i, {R}_i) \xrightarrow{a.s.} E[q_{i}( {Z}_i, {R}_i)], \\ 
    & I^{-1}E_{ {u}^{+}}(T) \xrightarrow{a.s.}  E\left[\sum_{\pi_{i}} \big\{\frac{\exp(\gamma ( {Z}_{i\pi_{i}} {R}_i^{T}))}{\sum_{\Tilde{\pi}_{i}}\exp(\gamma ( {Z}_{i\Tilde{\pi}_{i}} {R}_i^{T}))}\times q_{i}( {Z}_{i\pi_{i}}, {R}_{i})\big\}\right]=\phi(\gamma),\\
    & I^{-1}\text{var}_{ {u}^{+}}(T) \xrightarrow{a.s.} E\Big[\sum_{\pi_{i}} \big\{\frac{\exp(\gamma ( {Z}_{i\pi_{i}} {R}_i^{T}))}{\sum_{\Tilde{\pi}_{i}}\exp(\gamma ( {Z}_{i\Tilde{\pi}_{i}} {R}_i^{T}))}\times q_{i}^{2}( {Z}_{i\pi_{i}}, {R}_{i})\big\}\Big]\\
    &\quad \quad \quad \quad \quad \quad \quad \quad \quad \quad \quad-E\Big(\Big[\sum_{\pi_{i}} \big\{\frac{\exp(\gamma ( {Z}_{i\pi_{i}} {R}_i^{T}))}{\sum_{\Tilde{\pi}_{i}}\exp(\gamma ( {Z}_{i\Tilde{\pi}_{i}} {R}_i^{T}))}\times q_{i}( {Z}_{i\pi_{i}}, {R}_{i})\big\}\Big]^{2}\Big).
\end{aligned}
\end{equation*}
By Lemma~\ref{lem:unique}, the following equation 
\begin{equation*}
    \phi(\gamma)=E[q_{i}( {Z}_i, {R}_i)]
\end{equation*}
has a unique solution, denoted as $\widetilde{\gamma}$, such that $\phi(\gamma)<E[q_{i}( {Z}_i, {R}_i)]$ for $\gamma<\widetilde{\gamma}$ and $\phi(\gamma)>E[q_{i}( {Z}_i, {R}_i)]$ for $\gamma>\widetilde{\gamma}$.

The assumptions made in the statement of Theorem 3 ensure that Condition~\ref{normal_reg} holds (see the remark under Condition~\ref{normal_reg}). Therefore, invoking Theorem 2 in the main text, power of sensitivity analysis using some test statistic $T=\sum_{i=1}^{I}q_{i}( {Z}_i, {R}_i) = \sum_{i=1}^{I}\sum_{j=1}^{n_i} m(Z_{ij}) \times R_{ij}$ is (or asymptotically equals)
\begin{equation*}
\begin{aligned}
&\quad \ \text{pr}\left( \frac{\sum_{i=1}^{I}q_{i}( {Z}_i, {R}_i)-E_{ {u}^{+}}(T) }{\sqrt{\text{var}_{ {u}^{+}}(T)}} \geq \Phi^{-1}(1-\alpha)\right) \\
&=\text{pr}\left( \frac{I^{-1}\sum_{i=1}^{I}q_{i}( {Z}_i, {R}_i)-I^{-1}E_{ {u}^{+}}(T)}{I^{-1/2}\sqrt{I^{-1}\text{var}_{ {u}^{+}}(T)}}\geq \Phi^{-1}(1-\alpha)\right),
\end{aligned}
\end{equation*}
in which we have 
\begin{align*}
    &I^{-1}\sum_{i=1}^{I}q_{i}( {Z}_i, {R}_i)-I^{-1}E_{ {u}^{+}}(T)\xrightarrow{a.s.} E[q_{i}( {Z}_i, {R}_i)]-\phi(\gamma)>0 \text{ \ for \ $\gamma<\widetilde{\gamma}$,}\\
    &I^{-1}\sum_{i=1}^{I}q_{i}( {Z}_i, {R}_i)-I^{-1}E_{ {u}^{+}}(T)\xrightarrow{a.s.} E[q_{i}( {Z}_i, {R}_i)]-\phi(\gamma)<0 \text{\ for \  $\gamma>\widetilde{\gamma}$.}
\end{align*}
Since $I^{-1/2}\sqrt{I^{-1}\text{var}_{ {u}^{+}}(T)}\xrightarrow{a.s.}0$ as $I\rightarrow \infty$, we have  
\begin{align*}
    \frac{\sum_{i=1}^{I}q_{i}( {Z}_i, {R}_i)-E_{ {u}^{+}}(T) }{\sqrt{\text{var}_{ {u}^{+}}(T)}}\xrightarrow{a.s.} +\infty \text{\ for \ $\gamma<\widetilde{\gamma}$,}\\
    \frac{\sum_{i=1}^{I}q_{i}( {Z}_i, {R}_i)-E_{ {u}^{+}}(T) }{\sqrt{\text{var}_{ {u}^{+}}(T)}}\xrightarrow{a.s.} -\infty \text{\ for \ $\gamma>\widetilde{\gamma}$.}
\end{align*}
Recall that $\Gamma=\exp(\gamma)$ and $\widetilde{\Gamma}=\exp(\widetilde{\gamma})$. Therefore, we have: as $I\rightarrow \infty$, power of sensitivity analysis $\Psi_{\Gamma, I}\rightarrow 1$ if $\Gamma < \widetilde{\Gamma}$ and $\Psi_{\Gamma, I}\rightarrow 0$ if $\Gamma > \widetilde{\Gamma}$.
\end{proof}

\section*{Appendix B: Implications of the Sensitivity Parameter $\Gamma$ After Matching}

Note that the sensitivity parameter $\Gamma$ (or equivalently, $\gamma$) is a parameter imposed in the population exposure dose assignment model (i.e., the Rosenbaum dose assignment model)  before matching (\citealp{rosenbaum1989sensitivity, gastwirth1998dual}). After matching, it is helpful to examine to what extent various values of $\Gamma$ would depart exposure dose assignments from random (i.e., uniform) assignments among matched datasets. A meaningful quantity to report is the average (over all the matched sets) total variation distance (calculated within each matched set) between the biased dose assignment distribution in the worst-case scenario and the uniform dose assignment distribution. This shows that, in the worst-case scenario, how far a departure from random dose assignments arises under the Rosenbaum sensitivity model under some fixed $\Gamma$. In our real data application, we plot various histograms (under various $\Gamma$) of such total variation distances in Figure~\ref{fig:TV}.
\begin{figure}[h]    
\caption{Histogram of total variation distances between the biased exposure dose assignments distribution (in the worst-case scenario) and the uniform distribution among the 75 discordant matched sets in the real data example. The red dotted line represents the mean value.}
    \begin{center}
            \includegraphics[scale = 0.7]{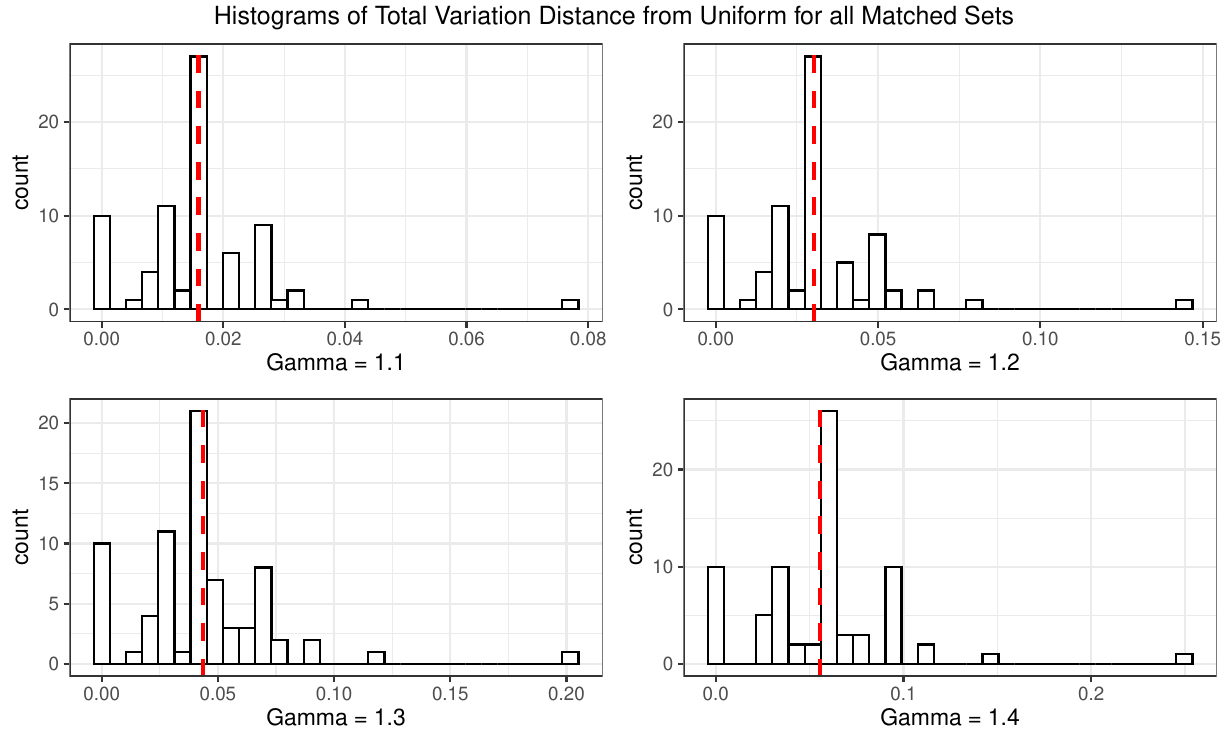}
    \end{center}
    \label{fig:TV}
\end{figure}

%\begin{equation*}
%p_{i\pi}=\frac{\exp(\gamma( {z}_{i\pi}' {u}_i))\prod_{j=1}^{n_i}f_{z_{ij}}(0)\exp(\kappa_{Z_{ij}}(x_{ij}))}{\sum_{ {z}_{i\pi} \in \Omega_i}\exp(\gamma( {z}_{i\pi}' {u}_i))\prod_{j=1}^{n_i}f_{z_{ij}}(0)\exp(\kappa_{Z_{ij}}(x_{ij}))}=\frac{\exp(\gamma( {z}_{i\pi}' {u}_i))}{\sum_{ {z}_{i\pi} \in \Omega_i}\exp(\gamma( {z}_{\pi}' {u}_i))}.
%\end{equation*}

\section*{Appendix C: Additional Details About the Hardness of Sensitivity Analysis with Continuous Exposures}

We here give more details on the hardness of sensitivity analysis in matched observational studies with continuous exposures. Assuming independence across matched sets, it suffices to illustrate the principle for each matched set $i$. Under the Rosenbaum sensitivity analysis model, the sharp null $H_{0}$, and the test statistic $T=\sum_{i=1}^{I}q_{i}( {Z}_{i},  {R}_{i})$ where $q_{i}( {Z}_{i},  {R}_{i})=\sum_{j=1}^{n_{i}}q_{ij}( {Z}_{i},  {R}_{i})$, for any fixed $k\in \mathbbm{R}$, we define 
\begin{equation*}
    \alpha_{i}( {r}_{i}, {u}_{i})=\text{pr}(q_{i}( {Z}_{i},  {r}_{i})\geq k) = \sum_{\pi_{i}} \Big[ \mathbbm{1}\{q_{i}( {Z}= {z}_{i\pi_{i} },  {r}_{i})\geq k\}\times \frac{\exp\{\gamma( {z}_{i\pi_{i}} {u}_{i}^{T})\}}{\sum_{\widetilde{\pi}_{i}}\exp\{\gamma( {z}_{i\widetilde{\pi}_{i}} {u}_{i}^{T})\}}\Big],
\end{equation*}
in which we have ${r}_{i}=(r_{i1},\dots, r_{in_{i}})=(R_{i1},\dots, R_{in_{i}})$ under $H_{0}$ and $ {u}_{i}=(u_{i1},\dots, u_{in_{i}})$. Without loss of generality, we assume $r_{i1}\leq \dots \leq r_{in_{i}}$. It is easy to verify that $p_{i\pi_{i}} ( {z}_{i}, {u}_i)=\exp\{\gamma( {z}_{i\pi_{i}} {u}_{i}^{T})\}/\sum_{\widetilde{\pi}_{i}}\exp\{\gamma( {z}_{i\widetilde{\pi}_{i}} {u}_{i}^{T})\}$ is arrangement increasing (\citealp{rosen_krieger, rosenbaum_obs}). Applying Theorem 3.3 from \citet{hollander1977functions}, we have that if $q_{i}( {Z}_{i}, {r}_i)$ is arrangement increasing, so is $\alpha( {r}_{i}, {u}_{i})$. This means that the maximizing $ {u}_i$ must be ordered according to $ {r}_{i}$. Therefore, we must have $u_{i1} \leq \cdots \leq u_{in_i}$. For clarity of exposition, we consider the case of one matched set and suppress the dependence on $i$ (assuming independence across different matched sets, conclusions can be easily generalized to the multiple matched sets case). Let $A_c^j=\sum_{\pi:z_{\pi(j)}=z_{(c)},q( {z}_{\pi}, {r}) \geq k}\exp(\gamma {z}_{\pi} {u}^{T})$ and $D_c^j=\sum_{\pi:z_{\pi(j)}=z_{(c)}}\exp(\gamma {z}_{\pi} {u}^{T})$, in which $A_c^j$ represents the sum of the probability weights for permutations under which subject $j$ is assigned the $c$-th smallest dose $z_{(c)}$ (i.e., we have $z_{(1)}< \dots z_{(c-1)}< z_{(c)}< z_{(c+1)}< \dots < z_{(n)}$) and the test statistic exceeds $k$. Meanwhile, $D_c^j$ represents the sum of the probability weights for permutations under which subject $j$ is assigned the $c$-th smallest dose $z_{(c)}$. We can then rewrite 
\begin{align*}
\alpha( {r}, {u}) &= \frac{\sum_{c=1}^{n}A_c^j}{\sum_{c=1}^{n} D_c^j}.
\end{align*}
Then it follows that 
\begin{align*}
    \frac{\partial \alpha( {r}, {u})}{\partial u_j}&=\frac{\left(\sum_{c=1}^{n}D_c^j\right)\left(\sum_{c=1}^{n}\gamma z_{(c)}A_c^j\right)-\left(\sum_{c=1}^{n_i}A_c^j\right)\left(\sum_{c=1}^{n}\gamma z_{(c)}D_c^j\right)}{(\sum_{c=1}^{n} D_c^j) ^2} \\ &= \frac{\left(\sum_{c=1}^{n}\gamma z_{(c)} A_c^j\right)}{(\sum_{c=1}^{n} D_c^j)}-\frac{\alpha( {r}, {u})\left(\sum_{c=1}^{n}\gamma z_{(c)}D_c^j\right)}{(\sum_{c=1}^{n} D_c^j)},
\end{align*}
which can be positive or negative. Unlike the binary exposure case demonstrated in \citet{rosen_krieger}, the $p$-value is not necessarily monotone in $u_{ij}$ here. In the binary exposure case, this monotonicity property guarantees that the worst-case (maximal) $p$-value is achieved at one of the corners of the unit cube, no matter the number of subjects in each matched set. The lack of monotonicity when both the exposure and outcome are continuous precludes a straightforward way of solving the worst-case $p$-value. In the main text, a counterexample is presented where the worst-case $p$-value is not achieved at a corner of the unit cube when both the exposure and the outcome are continuous.

\section*{Appendix D: Sensitivity Analysis for the Threshold Attributable Effect via A Generalized Asymptotic Separability Algorithm}

Under the setting introduced in Section 4 of the main text (including the monotonicity assumption: $r_{ij}(z) \leq r_{ij}(z')$ for all $z \leq z'$), there are study designs in which the problem is simpler in that the contribution to the threshold attributable effect (TAE) from each matched set is either 0 or 1. This allows for an analytical solution to inference and sensitivity analysis for the TAE in a similar approach as given in \citet{rosenbaum_attributable_obs}. A prominent example is if the exposure is dichotomized at some prespecified dose level $c$ for the purpose of matching so that in each matched set, there is one subject with exposure greater than $c$ and the rest having exposure less than $c$ (although the statistical inference and sensitivity analysis will still use the exact dose information). We hope to find a computationally tractable way to construct a prediction set for the TAE in the above cases. Note that for some value $a$ of the TAE $=\sum_{i=1}^{I}\sum_{j=1}^{n_{i}}[\mathbbm{1}\{Z_{ij} > c\}R_{ij} - \mathbbm{1}\{Z_{ij} > c\}r_{0ij}]$ to occur, in exactly $a$ of the matched sets with $\sum_{j=1}^{n_{i} }\mathbbm{1}\{Z_{ij} > c\}R_{ij} = 1$, the event (i.e., outcome equal to one) that occurred would need to be attributable to exposure greater than $c$. The goal is to find the $a$ such matched sets that make the pattern most difficult to reject. If we can reject this allocation of reference potential outcomes, then we can reject all allocation of reference potential outcomes that lead to a TAE value of $a$. First, note that we can write $\sum_{i=1}^{I} \sum_{j=1}^{n_{i}}\mathbbm{1}\{Z_{ij} > c\}r_{0ij} = \sum_{i=1}^{I} B_i$ where $B_i \sim \text{Bern}(\pi_i)$ for some unknown $\pi_i$. Let $T = \sum_{i=1}^{I}\sum_{j=1}^{n_{i}}\mathbbm{1}\{Z_{ij} > c\}R_{ij} $ be the number of overexposed (i.e., $Z_{ij}>c$) individuals that experienced an event (i.e., $R_{ij}=1$). Let $\beta(k,  {\pi})$ be the probability that there are at least $k$ successes in $I$ independent binary trials, where the trial $i$ has a probability of success $\pi_i$ and $ {\pi}=(\pi_{1}, \dots, \pi_{I})$. Under the Rosenbaum sensitivity analysis model, based on the arguments in Sections 3 and 4 of the main text, we have 
\begin{equation}
\label{bern_worst}
    \pi_i \leq \frac{\sum_{\pi_{i} }  {z}_{i\pi_{i} > c} {r}_{0i}^{T}\exp(\gamma( {z}_{i\pi_{i}} {r}_{0i}^{T}))}{\sum_{\pi_{i} } \exp(\gamma( {z}_{i\pi_{i}} {r}_{0i}^{T}))} = \overline{\pi}_i,
\end{equation}
since the $ {u}$ values that maximize the tail probability $\pi_i$ is at $ {u}_{i}= {r}_{0i}$. Note that it follows immediately that $\beta(k,\overline{ {\pi}}) \geq \beta(k, {\pi})$. As $I \to \infty$, under mild regularity conditions, we have 
\begin{equation}
\label{large_samp}
    \beta(k,\overline{ {\pi}}) \approx 1 - \Phi\left(\frac{k - \sum_{i=1}^{I} \overline{\pi}_{i}}{\sqrt{\sum_{i=1}^{I} \overline{\pi}_{i}(1-\overline{\pi}_{i})}}\right).
\end{equation}
For each matched set $i$, we want to compare the difference in expectations of $B_i$ if we attribute the event/effect to exposure or not in matched set $i$ (i.e., if we set $\sum_{j=1}^{n_{i}}[\mathbbm{1}\{Z_{ij} > c\}R_{ij} - \mathbbm{1}\{Z_{ij} > c\}r_{0ij}]=1$ or $0$). Let $\overline{\overline{\lambda}}_i$ denote the value of $\overline{\pi}_{i}$ in equation (\ref{bern_worst}) if we do not attribute the effect to matched set $i$ and $\overline{\lambda}_i$ if we attribute. In a case-control study, we have $\overline{\lambda}_i = 0$ since attributing the effect will result in $ {r}_{0i}$ being a vector of zeros. Meanwhile, $\overline{\overline{\lambda}}_i$ can be easily computed and depends on the observed doses in matched set $i$. Next, let $\overline{\overline{\omega}}_i = \overline{\overline{\lambda}}_i(1-\overline{\overline{\lambda}}_i)$ and $\overline{\omega}_i = \overline{\lambda}_i(1-\overline{\lambda}_i)$ be the associated contributed variances from matched set $i$. Following a similar argument to that in \citet{rosenbaum_attributable_obs}, the $ {r}_0$ formed using the $a$ smallest declines in expectations (i.e., $\overline{\overline{\lambda}}_i-\overline{\lambda}_i$) will be the hardest to reject based on large sample approximation, according to the notion of asymptotic separability originally developed in \citet{asymp_sep}. Specifically, in Proposition 1 from \citet{asymp_sep}, it is shown that as $I \to \infty$, the approximate worst-case (largest) $p$-value $\beta(k, \pi)$ is obtained from the normal reference distribution with the highest expectation, and among those with the same highest expectation, the one with the highest variance. This motivates the following generalized asymptotic separability algorithm for testing $\text{TAE}=a$ in a sensitivity analysis, which is computationally efficient and asymptotically exact.
\begin{algorithm}\label{algo: asymptotic separability}
\caption{A Generalized Asymptotic Separability Algorithm for Testing $\text{TAE}=a$}
\textbf{Input}: The prespecified $a$ for testing, the significance level $\alpha$, the sensitivity parameter $\Gamma$, the observed exposure doses $(Z_{11}, \dots, Z_{In_{I}})$ and the observed outcomes $(R_{11}, \dots, R_{In_{I}})$.

\textbf{Step 1: } Stop if $a > \sum_{i=1}^{I}\sum_{j=1}^{n_{i}} \mathbbm{1}\{Z_{ij} > c\} R_{ij} $ and reject the null hypothesis.

\textbf{Step 2: } For each matched set $i$ with $\sum_{j=1}^{n_i}\mathbbm{1}\{Z_{ij} > c\} R_{ij} = 1$, compute $\overline{\overline{\lambda}}_i$, $\overline{\lambda}_i$, $\overline{\overline{\omega}}_i$,$\overline{\omega}_i$.

\textbf{Step 3: } Among those matched sets with $\sum_{j=1}^{n_i}\mathbbm{1}\{Z_{ij} > c\} R_{ij} = 1$, select the $a$ sets with smallest $\overline{\overline{\lambda}}_i-\overline{\lambda}_i$, and if there are ties, select ones with smallest $\overline{\overline{\omega}}_i-\overline{\omega}_i$. For those $a$ sets, write $\overline{\pi}_i = \overline{\lambda}_i$. For the other $I-a$ sets, write $\overline{\pi}_i = \overline{\overline{\lambda}}_i$. Next, if $\sum_{i=1}^{I} \overline{\pi}_i \geq \sum_{i=1}^{I}\sum_{j=1}^{n_{i}} \mathbbm{1}\{Z_{ij} > c\} R_{ij} - a$, i.e., the expectation of the test statistic exceeds the observed test statistic, we accept the null hypothesis and stop.

\textbf{Step 4: } Compute the large-sample approximation of the upper
 bound on the significance level $\beta(k, \pi)$ (i.e., the $\beta(k, \overline{\pi})$) given in (\ref{large_samp}), where $k = \sum_{i=1}^{I}\sum_{j=1}^{n_{i}} \mathbbm{1}\{Z_{ij} > c\} R_{ij} - a$. That is, we approximate $\beta(k, \overline{\pi})$ via 
 \begin{equation}\label{eqn: approximate p-value}
     1 - \Phi\left(\frac{\sum_{i=1}^{I}\sum_{j=1}^{n_{i}} \mathbbm{1}\{Z_{ij} > c\} R_{ij} - a - \sum_{i=1}^{I} \overline{\pi}_{i}}{\sqrt{\sum_{i=1}^{I} \overline{\pi}_{i}(1-\overline{\pi}_{i})}}\right).
 \end{equation}
 If equation (\ref{eqn: approximate p-value}) is less than $\alpha$, we reject the null hypothesis; otherwise, we accept the null hypothesis.
\end{algorithm}

To construct an asymptotically exact $100(1-\alpha)\%$ confidence set for TAE in a sensitivity analysis with sensitivity parameter $\Gamma$, we just need to collect all the values $a$ that cannot be rejected by Algorithm 1 with level $\alpha$ and sensitivity parameter $\Gamma$.

\section*{Appendix E: More Details on the Matching In the Real Data Application}
We first describe the preprocessing and trimming that we performed to arrive at a sufficiently balanced matched dataset. First, any individuals missing in the outcome, exposure, or measured confounders to be matched were discarded, leaving 10,582 individuals. Next, we applied the \texttt{nbpfull} package from \citet{bo_continuous} using a rank-based Mahalanobis distance with a penalty parameter to avoid overly large matched sets and another penalty parameter to avoid the case that all the matched individuals have similar exposure doses (\citealp{bo_continuous}). Near-fine balance was encouraged on the confounder ``Birth year" and the confounder ``Year of the blood test." Moreover, to further improve the matching quality, we followed the advice in \citet{rosenbaum2020modern} to conduct data trimming to ensure adequate confounder balance. Specifically, matched sets with ``star-distance" (see \citet{bo_continuous}) above the 95th percentile were removed, as well as matched sets that were not exactly matched on the indicator confounder ``Black." Next, matched sets with a range of confounders ``Year of construction," ``Percent black (block)," and ``Log home value" exceeding 20, 0.25, and 0.25, respectively, were removed. This resulted in the final 4,134 individuals in the matched dataset.

In Table~\ref{tab: balance}, we present a confounder balance table for the matched cohort study considered in our real data application. For each confounder, before matching, we consider the sample means among subjects below and above the median dose, respectively, and the standardized mean difference (SMD) is computed by taking the difference in these two mean values and dividing it by the standard deviation of the confounder among the whole sample. After matching, we collect the subjects whose doses are above the median dose within their corresponding matched sets, calling them the ``high" group. We group the remaining subjects into the ``low" group. Once again, we compute the sample means and the SMD, where the SMD is calculated by taking the difference in these two newly calculated mean values and dividing it by the standard deviation of the confounder among the whole sample (i.e., the denominator is the same as that used in the before-matching SMD). In addition to comparisons of mean values of confounders, we also compute $p$-values reported by Kolmogorov-Smirnov (KS) tests for testing the equality of the distributions of the below (median dose) versus above (median dose) groups before matching and the low dose group versus high dose group after matching. In summary, matching can sufficiently balance the observed confounders in our real data application, with all $p$-values reported by KS tests are above 0.81 and all absolute SMDs below 0.1.

\begin{table}[ht]
\begin{center}
\caption{Assessment of confounder balance before and after matching. We report the mean values of confounders before and after matching among different dose groups, the standardized mean differences (SMDs) before and after matching, and $p$-values for testing equity of confounder distributions between different dose groups before and after matching (using KS tests).}
\label{tab: balance}
\begin{tabular}{lrrrrrrrr}
 \hline
  & \multicolumn{4}{c}{Before Matching} & \multicolumn{4}{c}{After Matching} \\
Confounder & Below & Above & SMD & KS $p$ & Low & High & SMD & KS $p$\\
  \hline
Age & 23.20 & 22.79 & -0.02 & 0.00 & 19.83 & 19.63 & -0.01 & 1.00 \\ 
Year of test & 2005.01 & 2004.33 & -0.37 & 0.00 & 2004.53 & 2004.52 & 0.00 & 1.00 \\ 
Gender & 0.49 & 0.50 & 0.03 & 0.82 & 0.48 & 0.48 & -0.00 & 1.00 \\ 
Year of Construction & 1987.05 & 1982.37 & -0.22 & 0.00 & 1989.61 & 1989.53 & -0.01 & 0.97 \\ 
Log home value & 12.27 & 12.07 & -0.33 & 0.00 & 12.13 & 12.13 & -0.00 & 0.81 \\ 
Blood draw (venous) & 0.87 & 0.91 & 0.13 & 0.00 & 0.92 & 0.92 & 0.00 & 1.00 \\ 
Birth year & 2003.09 & 2002.44 & -0.34 & 0.00 & 2002.89 & 2002.89 & 0.00 & 1.00 \\ 
Percent black (block)& 0.22 & 0.28 & 0.23 & 0.00 & 0.23 & 0.23 & 0.00 & 0.95 \\ 
Log(med income) (block) & 11.19 & 11.07 & -0.26 & 0.00 & 11.18 & 11.19 & 0.01 & 0.97 \\ 
Season 2 & 0.27 & 0.28 & 0.01 & 1.00 & 0.29 & 0.29 & 0.00 & 1.00 \\ 
Season 3 & 0.24 & 0.24 & 0.00 & 1.00 & 0.25 & 0.25 & 0.01 & 1.00 \\ 
Season 4 & 0.23 & 0.24 & 0.04 & 0.57 & 0.22 & 0.22 & 0.00 & 1.00 \\ 
Black & 0.17 & 0.30 & 0.31 & 0.00 & 0.21 & 0.21 & 0.00 & 1.00 \\ 
   \hline
\end{tabular}
\end{center}
\end{table}

\section*{Appendix F: Additional Simulation Results}

In Table \ref{tab:ds+pow_supp}, we present additional simulation results on investigating the design sensitivity and finite-sample power of sensitivity analysis of various test statistics. The data-generating mechanisms are identical to that considered in Section 3 of the main text, with the only difference being that the random effects for each matched set are drawn from $N(-2,1)$ rather than $N(0,1)$. The test statistics compared are the same, and the general patterns of the simulation results are consistent with those from the main text.
\begin{table}[H]
\centering
\caption{Design sensitivity and finite-sample power (under the significance level 0.05) for the five test statistics under $f(z) = z^a$ for $a = 1/4, 1/2, 2, 4$ and $\beta=1.5$ for the outcome model. The random effect for each matched set is drawn from $N(-2,1)$, and the exposure dose $Z$ is drawn from $\text{Unif}[0,1]$ or $\text{Beta}(2,2)$.}
\small
\begin{tabular}{lrrrrr|rrrrr}
\multicolumn{11}{c}{$f(z)=z^{0.25}$} \\
 & \multicolumn{5}{c|}{$Z\sim \text{Unif}[0,1]$} & \multicolumn{5}{c}{$Z\sim$ Beta(2,2)} \\
\cline{1-1} \cline{2-6} \cline{7-11}
Test & $T_{t}$ & $T_{0.1}$ & $T_{0.25}$ & $T_{0.5}$ & $T_{\text{adap}}$ & $T_{t}$ & $T_{0.1}$ & $T_{0.25}$ & $T_{0.5}$ & $T_{\text{adap}}$ \\

\hline
$\widetilde{\Gamma}$ & 2.14 & 3.04 & 2.48 & 2.06 & 3.04 & 2.06 & 3.19 & 2.24 & 2.10 & 3.19\\

$\Gamma = 1.75$ & 0.38 & 0.57 & 0.50 & 0.24 & 0.52 & 0.15 & 0.18 & 0.19 & 0.12 & 0.17\\

$\Gamma = 2.00$ & 0.16 & 0.42 & 0.30 & 0.09 & 0.32 & 0.06 & 0.12 & 0.13 & 0.04 & 0.09\\

$\Gamma = 2.25$ & 0.04 & 0.29 & 0.16 & 0.03 & 0.20 & 0.02 & 0.09 & 0.07 & 0.01 & 0.07\\

$\Gamma = 2.50$ & 0.01 & 0.20 & 0.08 & 0.00 & 0.13 & 0.01 & 0.07 & 0.04 & 0.00 & 0.05\\
%\hline
%Power5 & 0.00 & 0.14 & 0.03 & 0.00 & 0.08 & 0.00 & 0.06 & 0.01 & 0.00 & 0.04\\
%\hline
%Power6 & 0.00 & 0.09 & 0.02 & 0.00 & 0.08 & 0.00 & 0.05 & 0.01 & 0.00 & 0.04\\
\hline
\end{tabular}

\begin{tabular}{lrrrrr|rrrrr}
\multicolumn{11}{c}{$f(z)=z^{0.5}$} \\
 & \multicolumn{5}{c|}{$Z\sim \text{Unif}[0,1]$} & \multicolumn{5}{c}{$Z\sim$ Beta(2,2)} \\
\cline{1-1} \cline{2-6} \cline{7-11}
Test & $T_{t}$ & $T_{0.1}$ & $T_{0.25}$ & $T_{0.5}$ & $T_{\text{adap}}$ & $T_{t}$ & $T_{0.1}$ & $T_{0.25}$ & $T_{0.5}$ & $T_{\text{adap}}$ \\
\hline
$\widetilde{\Gamma}$ & 3.21 & 4.85 & 3.93 & 3.12 & 4.85 & 3.18 & 5.60 & 4.09 & 3.11 & 5.60\\

$\Gamma = 3.00$ & 0.12 & 0.31 & 0.22 & 0.08 & 0.22 & 0.06 & 0.14 & 0.12 & 0.04 & 0.11\\

$\Gamma = 3.25$ & 0.06 & 0.24 & 0.15 & 0.04 & 0.16 & 0.02 & 0.12 & 0.08 & 0.02 & 0.06\\

$\Gamma = 3.50$ & 0.02 & 0.17 & 0.10 & 0.02 & 0.11 & 0.01 & 0.10 & 0.07 & 0.01 & 0.05\\

$\Gamma = 3.75$ & 0.00 & 0.14 & 0.07 & 0.00 & 0.09 & 0.00 & 0.09 & 0.04 & 0.01 & 0.04\\
%\hline
%Power5 & 0.00 & 0.11 & 0.04 & 0.00 & 0.06 & 0.00 & 0.08 & 0.02 & 0.00 & 0.04\\
%\hline
%Power6 & 0.00 & 0.09 & 0.02 & 0.00 & 0.06 & 0.00 & 0.06 & 0.02 & 0.00 & 0.04\\
\hline
\end{tabular}

\begin{tabular}{lrrrrr|rrrrr}
\multicolumn{11}{c}{$f(z)=z^{2}$} \\
 & \multicolumn{5}{c|}{$Z\sim \text{Unif}[0,1]$} & \multicolumn{5}{c}{$Z\sim$ Beta(2,2)} \\
\cline{1-1} \cline{2-6} \cline{7-11}
Test & $T_{t}$ & $T_{0.1}$ & $T_{0.25}$ & $T_{0.5}$ & $T_{\text{adap}}$ & $T_{t}$ & $T_{0.1}$ & $T_{0.25}$ & $T_{0.5}$ & $T_{\text{adap}}$ \\
\hline
$\widetilde{\Gamma}$ & 4.55 & 2.94 & 3.48 & 4.58 & 4.55 & 4.48 & 2.84 & 3.38 & 4.56 & 4.48\\

$\Gamma = 3.00$ & 0.64 & 0.04 & 0.14 & 0.52 & 0.54 & 0.47 & 0.01 & 0.06 & 0.36 & 0.33\\

$\Gamma = 3.25$ & 0.52 & 0.04 & 0.08 & 0.38 & 0.39 & 0.33 & 0.01 & 0.04 & 0.27 & 0.23\\

$\Gamma = 3.50$ & 0.37 & 0.02 & 0.04 & 0.26 & 0.25 & 0.23 & 0.01 & 0.03 & 0.20 & 0.17\\

$\Gamma = 3.75$ & 0.25 & 0.02 & 0.02 & 0.19 & 0.17 & 0.17 & 0.01 & 0.03 & 0.15 & 0.11\\
%\hline
%Power5 & 0.17 & 0.01 & 0.01 & 0.13 & 0.09 & 0.13 & 0.01 & 0.02 & 0.11 & 0.07\\
%\hline
%Power6 & 0.10 & 0.01 & 0.00 & 0.09 & 0.05 & 0.09 & 0.01 & 0.02 & 0.08 & 0.05\\
\hline
\end{tabular}

\begin{tabular}{lrrrrr|rrrrr}
\multicolumn{11}{c}{$f(z)=z^{4}$} \\
 & \multicolumn{5}{c|}{$Z\sim \text{Unif}[0,1]$} & \multicolumn{5}{c}{$Z\sim$ Beta(2,2)} \\
\cline{1-1} \cline{2-6} \cline{7-11}
Test & $T_{t}$ & $T_{0.1}$ & $T_{0.25}$ & $T_{0.5}$ & $T_{\text{adap}}$ & $T_{t}$ & $T_{0.1}$ & $T_{0.25}$ & $T_{0.5}$ & $T_{\text{adap}}$ \\
\hline
$\widetilde{\Gamma}$ & 3.50 & 2.19 & 2.30 & 3.18 & 3.50 & 3.29 & 2.14 & 1.89 & 2.89 & 3.29\\

$\Gamma = 1.75$ & 0.96 & 0.09 & 0.24 & 0.80 & 0.92 & 0.72 & 0.10 & 0.07 & 0.44 & 0.58\\

$\Gamma = 2.00$ & 0.83 & 0.05 & 0.12 & 0.57 & 0.76 & 0.50 & 0.07 & 0.04 & 0.27 & 0.37\\

$\Gamma = 2.25$ & 0.65 & 0.02 & 0.06 & 0.37 & 0.52 & 0.32 & 0.06 & 0.02 & 0.17 & 0.23\\

$\Gamma = 2.50$ & 0.42 & 0.01 & 0.03 & 0.19 & 0.30 & 0.20 & 0.04 & 0.01 & 0.11 & 0.14\\
%\hline
%Power5 & 0.25 & 0.01 & 0.02 & 0.09 & 0.14 & 0.12 & 0.03 & 0.00 & 0.07 & 0.08\\
%\hline
%Power6 & 0.11 & 0.00 & 0.01 & 0.04 & 0.07 & 0.06 & 0.03 & 0.00 & 0.04 & 0.05\\
\hline
\end{tabular}
\label{tab:ds+pow_supp}
\end{table}

We also conducted additional simulations under a setting where the dose distribution has a point mass at the zero dose (i.e., many individuals are completely unexposed) and is right-skewed, which is realistic in many datasets (e.g., the real data example considered in the main text). Specifically, when generating the exposure dose $Z$ under the general data-generating framework described in Section 3.2 of the main text, we randomly generate each exposure dose $Z_{ij}$ from a mixture distribution: the first component distribution of this mixture distribution is a point mass at 0 with probability $0.2$, and the second component distribution is $\text{Unif}[0,1]$ or $\text{Beta}(2,8)$ with probability $0.8$. In Table~\ref{tab:ds+pow point_mass}, we report the corresponding design sensitivity and simulated finite-sample power (under the significance level 0.05) of the five considered test statistics under $f(z)=\mathbbm{1}\{z>0\}$ and $f(z) = z^a$ for $a = 1/4, 1/2, 2, 4$ and $\beta=0.8$ for the outcome model considered in Section 3.2, with 2000 matched sets for calculating the finite-sample power. Several insights are observed in Table~\ref{tab:ds+pow point_mass}. First, when there is a point mass at the zero dose, the design sensitivity can still provide helpful guidance for comparing different test statistics as it agrees well with the finite-sample power: a test statistic with larger design sensitivity, in general, tends to have higher finite-sample power (although exceptions may exist in some finite-sample settings). Second, by comparing the design sensitivity and finite-sample power of the five considered test statistics when there is a point mass at the zero exposure dose, we still find that different test statistics may perform very differently in a sensitivity analysis, and which test statistics is most powerful (either asymptotically or in finite sample) depends on the underlying data-generating process. In particular, for the considered data-generating processes, we notice that when $f(z)=\mathbbm{1}\{z>0\}$, $f(z)=z^{0.25}$, or $f(z)=z^{0.5}$ (i.e., $f(z)$ is concave in $z$), the test statistic $T_{0.1}$ typically performs best. When $f(z) = z^2$ or $z^4$, either $T_t$ or $T_{0.5}$ performs best. These patterns generally agree with those observed in Table 2 in the main text. Therefore, the relevant discussions in Section 3.2 of the main text (on pages 8 and 9 of the main text) about the insights on the different performances of various test statistics are still helpful when there is a point mass at the zero exposure dose.

\begin{table}[H]
\centering
\caption{Design sensitivity and finite-sample power (under the significance level 0.05) for the five test statistics under $f(z)=\mathbbm{1}\{z>0\}$ and $f(z) = z^a$ for $a = 1/4, 1/2, 2, 4$ and $\beta=0.8$ for the outcome model. The random effect for each matched set is drawn from $N(0,1)$, and the exposure dose $Z$ is drawn from a mixture distribution. The first component of the mixture is a point mass at 0 with probability $0.2$, and the second component is $\text{Unif}[0,1]$ or $\text{Beta}(2,8)$.}
\small
  \begin{tabular}{lrrrrr|lrrrrr}
    \multicolumn{12}{c}{$f(z) = \mathbbm{1}\{z > 0\}$} \\
    \multicolumn{1}{c}{} & \multicolumn{5}{c}{$Z\sim \text{Unif}[0,1]$} & \multicolumn{1}{c}{} & \multicolumn{5}{c}{$Z\sim$ Beta(2,8)} \\
    \cline{1-6}\cline{7-12}
    Test & $T_{t}$ & $T_{0.1}$ & $T_{0.25}$ & $T_{0.5}$ & $T_{\text{adap}}$ & Test & $T_{t}$ & $T_{0.1}$ & $T_{0.25}$ & $T_{0.5}$ & $T_{\text{adap}}$\\
    \cline{1-6}\cline{7-12}
    $\Tilde{\Gamma}$ & 1.87 & 3.07 & 2.15 & 1.61 & 3.07 & $\Tilde{\Gamma}$ & 3.91 & 6.83 & 2.25 & 1.39 & 6.83\\
    $\Gamma = 1.33$ & 0.71 & 0.98 & 0.81 & 0.25 & 0.97 & $\Gamma = 2.00$ & 0.60 & 0.84 & 0.13 & 0.03 & 0.78\\
    $\Gamma = 1.67$ & 0.18 & 0.88 & 0.37 & 0.02 & 0.80 & $\Gamma = 2.33$ & 0.43 & 0.75 & 0.06 & 0.02 & 0.67\\
    $\Gamma = 2.00$ & 0.02 & 0.61 & 0.09 & 0.00 & 0.48 & $\Gamma = 2.67$ & 0.31 & 0.66 & 0.04 & 0.01 & 0.54\\
    $\Gamma = 2.33$ & 0.00 & 0.34 & 0.01 & 0.00 & 0.24 & $\Gamma = 3.00$ & 0.20 & 0.56 & 0.02 & 0.01 & 0.46\\
    $\Gamma = 2.67$ & 0.00 & 0.17 & 0.00 & 0.00 & 0.09 & $\Gamma = 4.00$ & 0.04 & 0.36 & 0.00 & 0.01 & 0.25\\
    \cline{1-6}\cline{7-12}
  \end{tabular}
  \begin{tabular}{lrrrrr|lrrrrr}
    \multicolumn{12}{c}{$f(z) = z^{0.25}$}\\
    \multicolumn{1}{c}{} & \multicolumn{5}{c}{$Z\sim \text{Unif}[0,1]$} & \multicolumn{1}{c}{} & \multicolumn{5}{c}{$Z\sim$ Beta(2,8)} \\
    \cline{1-6}\cline{7-12}
    Test & $T_{t}$ & $T_{0.1}$ & $T_{0.25}$ & $T_{0.5}$ & $T_{\text{adap}}$ & Test & $T_{t}$ & $T_{0.1}$ & $T_{0.25}$ & $T_{0.5}$ & $T_{\text{adap}}$\\
    \cline{1-6}\cline{7-12}
    $\Tilde{\Gamma}$ & 2.11 & 2.95 & 2.37 & 1.93 & 2.95 & $\Tilde{\Gamma}$ & 4.26 & 6.87 & 3.04 & 1.98 & 6.87\\
    $\Gamma = 1.33$ & 0.97 & 0.99 & 0.95 & 0.76 & 0.99 & $\Gamma = 2.00$ & 0.56 & 0.69 & 0.17 & 0.03 & 0.63\\
    $\Gamma = 1.67$ & 0.53 & 0.91 & 0.63 & 0.25 & 0.83 & $\Gamma = 2.33$ & 0.39 & 0.57 & 0.11 & 0.02 & 0.46\\
    $\Gamma = 2.00$ & 0.13 & 0.60 & 0.24 & 0.04 & 0.48 & $\Gamma = 2.67$ & 0.23 & 0.44 & 0.07 & 0.01 & 0.34\\
    $\Gamma = 2.33$ & 0.02 & 0.31 & 0.06 & 0.01 & 0.22 & $\Gamma = 3.00$ & 0.15 & 0.34 & 0.03 & 0.01 & 0.26\\
    $\Gamma = 2.67$ & 0.00 & 0.14 & 0.01 & 0.00 & 0.08 & $\Gamma = 4.00$ & 0.03 & 0.16 & 0.01 & 0.00 & 0.11\\
    \cline{1-6}\cline{7-12}
  \end{tabular}
  \begin{tabular}{lrrrrr|lrrrrr}
    \multicolumn{12}{c}{$f(z) = z^{0.5}$} \\
    \multicolumn{1}{c}{} & \multicolumn{5}{c}{$Z\sim \text{Unif}[0,1]$} & \multicolumn{1}{c}{} & \multicolumn{5}{c}{$Z\sim$ Beta(2,8)} \\
    \cline{1-6}\cline{7-12}
    Test & $T_{t}$ & $T_{0.1}$ & $T_{0.25}$ & $T_{0.5}$ & $T_{\text{adap}}$ & Test & $T_{t}$ & $T_{0.1}$ & $T_{0.25}$ & $T_{0.5}$ & $T_{\text{adap}}$\\
    \cline{1-6}\cline{7-12}
    $\Tilde{\Gamma}$ & 2.26 & 2.65 & 2.40 & 2.17 & 2.65 & $\Tilde{\Gamma}$ & 3.58 & 4.83 & 2.83 & 3.29 & 4.83\\
    $\Gamma = 1.33$ & 0.97 & 0.97 & 0.96 & 0.89 & 0.98 & $\Gamma = 2.00$ & 0.34 & 0.39 & 0.15 & 0.06 & 0.34\\
    $\Gamma = 1.67$ & 0.63 & 0.73 & 0.62 & 0.42 & 0.69 & $\Gamma = 2.33$ & 0.19 & 0.28 & 0.08 & 0.05 & 0.22\\
    $\Gamma = 2.00$ & 0.19 & 0.39 & 0.24 & 0.12 & 0.33 & $\Gamma = 2.67$ & 0.10 & 0.19 & 0.04 & 0.03 & 0.14\\
    $\Gamma = 2.33$ & 0.03 & 0.18 & 0.07 & 0.02 & 0.11 & $\Gamma = 3.00$ & 0.06 & 0.13 & 0.03 & 0.02 & 0.09\\
    $\Gamma = 2.67$ & 0.00 & 0.06 & 0.01 & 0.00 & 0.03 & $\Gamma = 4.00$ & 0.01 & 0.05 & 0.00 & 0.01 & 0.02\\
    \cline{1-6}\cline{7-12}
  \end{tabular}
   \begin{tabular}{lrrrrr|lrrrrr}
    \multicolumn{12}{c}{$f(z) = z^{2}$}  \\
    \multicolumn{1}{c}{} & \multicolumn{5}{c}{$Z\sim \text{Unif}[0,1]$} & \multicolumn{1}{c}{} & \multicolumn{5}{c}{$Z\sim$ Beta(2,8)} \\
    \cline{1-6}\cline{7-12}
    Test & $T_{t}$ & $T_{0.1}$ & $T_{0.25}$ & $T_{0.5}$ & $T_{\text{adap}}$ & Test & $T_{t}$ & $T_{0.1}$ & $T_{0.25}$ & $T_{0.5}$ & $T_{\text{adap}}$\\
    \cline{1-6}\cline{7-12}
    $\Tilde{\Gamma}$ & 2.07 & 1.85 & 1.94 & 2.10 & 2.07 & $\Tilde{\Gamma}$ & 1.40 & 1.24 & 1.52 & 1.98 & 1.40\\
    $\Gamma = 1.10$ & 1.00 & 0.75 & 0.92 & 0.99 & 0.99 & $\Gamma = 1.10$ & 0.18 & 0.10 & 0.15 & 0.14 & 0.14\\
    $\Gamma = 1.25$ & 0.95 & 0.52 & 0.75 & 0.93 & 0.91 & $\Gamma = 1.25$ & 0.11 & 0.06 & 0.08 & 0.11 & 0.07\\
    $\Gamma = 1.40$ & 0.81 & 0.29 & 0.51 & 0.76 & 0.73 & $\Gamma = 1.40$ & 0.06 & 0.03 & 0.06 & 0.09 & 0.04\\
    $\Gamma = 1.55$ & 0.58 & 0.14 & 0.28 & 0.55 & 0.45 & $\Gamma = 1.55$ & 0.03 & 0.02 & 0.04 & 0.08 & 0.02\\
    $\Gamma = 1.70$ & 0.34 & 0.07 & 0.12 & 0.33 & 0.23 & $\Gamma = 1.70$ & 0.02 & 0.01 & 0.03 & 0.07 & 0.01\\
    \cline{1-6}\cline{7-12}
  \end{tabular}
\begin{tabular}{lrrrrr|lrrrrr}
    \multicolumn{12}{c}{$f(z) = z^{4}$}  \\
    \multicolumn{1}{c}{} & \multicolumn{5}{c}{$Z\sim \text{Unif}[0,1]$} & \multicolumn{1}{c}{} & \multicolumn{5}{c}{$Z\sim$ Beta(2,8)} \\
    \cline{1-6}\cline{7-12}
    Test & $T_{t}$ & $T_{0.1}$ & $T_{0.25}$ & $T_{0.5}$ & $T_{\text{adap}}$ & Test & $T_{t}$ & $T_{0.1}$ & $T_{0.25}$ & $T_{0.5}$ & $T_{\text{adap}}$\\
    \cline{1-6}\cline{7-12}
    $\Tilde{\Gamma}$ & 1.67 & 1.38 & 1.47 & 1.68 & 1.67 & $\Tilde{\Gamma}$ & 1.09 & 1.07 & 1.15 & 1.00 & 1.09\\
    $\Gamma = 1.10$ & 0.88 & 0.27 & 0.49 & 0.81 & 0.80 & $\Gamma = 1.05$ & 0.06 & 0.04 & 0.06 & 0.05 & 0.04\\
    $\Gamma = 1.25$ & 0.60 & 0.11 & 0.20 & 0.52 & 0.47 & $\Gamma = 1.10$ & 0.05 & 0.03 & 0.05 & 0.04 & 0.03\\
    $\Gamma = 1.40$ & 0.28 & 0.04 & 0.07 & 0.26 & 0.19 & $\Gamma = 1.15$ & 0.03 & 0.03 & 0.04 & 0.04 & 0.02\\
    $\Gamma = 1.55$ & 0.11 & 0.01 & 0.02 & 0.09 & 0.06 & $\Gamma = 1.20$ & 0.02 & 0.02 & 0.03 & 0.03 & 0.02\\
    $\Gamma = 1.70$ & 0.03 & 0.00 & 0.01 & 0.02 & 0.02 & $\Gamma = 1.25$ & 0.01 & 0.02 & 0.02 & 0.03 & 0.02\\
    \cline{1-6}\cline{7-12}
  \end{tabular}
  \label{tab:ds+pow point_mass}
\end{table}

\section*{Appendix G: Testing the Uniform Dose Assignment Assumption}

In Appendix G, under the no unmeasured confounding assumption, we describe a detailed procedure for testing the following uniform exposure dose assignment assumption (i.e., assumption (1) in Section 2.1 of the main text): 
\begin{align*}
    &\text{$p_{i\pi_{i}}=\text{pr}( {Z}_i= {z}_{i\pi_{i}}|\mathcal{F}_{0},\mathcal{Z}_{i})=\frac{1}{n_{i}!}$ for all $\pi_{i} \in S_{n_i}$ (i.e., for all $ {z}_{i\pi_{i}} \in \mathcal{Z}_{i}$).} \\
   & \quad \quad \quad \quad \text{(The Uniform Dose Assignment Assumption)} 
\end{align*}
One strategy to derive a valid randomization test for testing the uniform dose assignment assumption is to directly generalize the existing randomization-based balance tests in the binary exposure case (e.g., \citealp{gagnon_cpt, branson2020evaluating, branson2021randomization, biased_randomization}) to the continuous exposure case. For example, one specific way is to generalize the testing procedure in \citet{biased_randomization} to the continuous exposure case, of which a detailed procedure is described in Algorithm~\ref{alg: perfect randomization} as follows.

\begin{algorithm}
\caption{Testing the Uniform Dose Assignment Assumption.}\label{alg: perfect randomization}
\textbf{Input} The measured confounder information $\mathcal{X} = \{x_{ij}, i = 1,\ldots,I; j \ 1,\ldots,n_i\}$, the treatment indicators $Z =(Z_{11}, \dots, Z_{In_{I}})$, and the significance level $\alpha$.

\textbf{Step 1:}  Randomly split the matched set indexes into two non-overlapping parts (with an equal number of matched sets): $\mathcal{I}^{(1)}$ and $\mathcal{I}^{(2)}$ such that
$\mathcal{I}^{(1)} \cup \mathcal{I}^{(2)} = \{1,\ldots, I\}$ and $\mathcal{I}^{(1)} \cap \mathcal{I}^{(2)} = \emptyset$.

\textbf{Step 2:}  Train a prediction model (e.g., linear regression model) $g_1$ with the measured confounder information $X^{(1)} = \{x_{ij}, i \in \mathcal{I}^{(1)}, j = 1, \ldots, n_i\}$ from the first sample part $\mathcal{I}^{(1)}$ as predictors and the corresponding dose $Z^{(1)} = \{Z_{ij}, i \in \mathcal{I}^{(1)}, j = 1, \ldots, n_i\}$ as labels (i.e., dependent variables).

\textbf{Step 3:}  Train a prediction model (e.g., linear regression model) $g_2$ with the measured confounder information $X^{(2)} = \{x_{ij}, i \in \mathcal{I}^{(1)}, j = 1, \ldots, n_i\}$ from the second sample part $\mathcal{I}^{(2)}$ as predictors and the corresponding dose $Z^{(2)} = \{Z_{ij}, i \in \mathcal{I}^{(2)}, j = 1, \ldots, n_i\}$ as labels (i.e., dependent variables). 

\textbf{Step 4:}  Let $g_{1\to 2} = \{
g_1(x_{ij} ), i \in \mathcal{I}^{(2)}, j = 1,\ldots, n_i \}$ denote the predicted exposure doses for $\mathcal{I}^{(2)}$ based on
the fitted model $g_1$. Define the test statistic $T_{1\to2} =
\sum_{i\in \mathcal{I}^{(2)}} \sum_{j=1}^{n_i} Z_{ij}g_1(x_{ij})$, and use its null distribution (i.e., the randomization distribution of $T_{1\to2}$ under the uniform dose assignment assumption) to obtain a corresponding permutational $p$-value $p_{1 \to 2}$. 

\textbf{Step 5:} Let $g_{2\to 1} = \{
g_2(x_{ij} ), i \in \mathcal{I}^{(1)}, j = 1,\ldots, n_i \}$ denote the predicted exposure doses for $\mathcal{I}^{(1)}$ based on
the fitted model $g_2$. Define the test statistic $T_{2\to1} =
\sum_{i\in \mathcal{I}^{(1)}} \sum_{j=1}^{n_i} Z_{ij}g_2(x_{ij})$, and use its null distribution (i.e., the randomization distribution of $T_{2\to1}$ under the uniform dose assignment assumption) to obtain a corresponding permutational $p$-value $p_{2 \to 1}$. 

\textbf{Step 6:} Reject the uniform dose assignment assumption under the significance level $\alpha$ if and only if
$\min\{p_{1 \to 2},p_{2 \to 1}\} < \alpha/2$.
\end{algorithm}

In our data application, applying Algorithm~\ref{alg: perfect randomization} with setting the dose prediction models $g_{1}$ and $g_{2}$ as linear regression models (i.e., the default prediction model used in Algorithm~\ref{alg: perfect randomization}), we failed to reject the random exposure dose assignment assumption among the matched dataset under the significance level $\alpha=0.1$.

\begin{remark}
If the uniform dose assignment assumption was rejected, i.e., if the matched dataset significantly violated the uniform dose assignment assumption based on the measured confounder information, then the bias due to inexact matching on measured confounders is also a major concern and should be further quantified and carefully addressed. In matched observational studies with binary exposures, some previous studies proposed to quantify the biased exposure assignment probabilities and incorporate it into the downstream sensitivity analysis (\citealp{biased_randomization, pimentel2022covariateadaptive}). To generalize this type of strategy to the continuous exposure and binary outcome case, there are two main steps, and we identified both opportunities and challenges among them:
\begin{itemize}
    \item Step 1: Quantifying the biased exposure dose assignment due to inexact matching on measured confounders (e.g., some calibrated value $\Gamma_{0}$ such that $1<\Gamma_{0}\leq \Gamma$).
    \item Step 2: Incorporating the biased exposure dose assignments due to inexact matching (e.g., the calibrated value $\Gamma_{0}$) into our downstream sensitivity analysis (e.g., by setting the sensitivity parameter $\Gamma\geq \Gamma_{0})$).
\end{itemize}
Step 2 is relatively straightforward to be accomplished in the setting with continuous exposure and binary outcome. Specifically, by directly combining the arguments in Theorem 3 in \citet{pimentel2022covariateadaptive} with our proposed sensitivity analysis method, it is straightforward to show that we can directly adjust our sensitivity analysis method to conduct a finite-population-exact sensitivity analysis under Fisher's sharp null $H_{0}$ with some biased (due to inexact matching) baseline exposure dose assignment probability distribution (in the continuous exposure and binary outcome case). Therefore, the only remaining problem is how to appropriately determine/quantify the biased exposure dose assignments due to inexact matching (i.e., the goal of Step 1). After some preliminary investigations, we found that this problem (i.e., Step 1) is much more challenging than that in the binary exposure case. Specifically, in the binary exposure case, there are two main strategies to quantify the biased exposure assignments in matched observational studies:
\begin{itemize}
    \item Strategy 1 (\citealp{biased_randomization}): Researchers conduct randomization balance tests for biased exposure assignments to quantify some sensible calibrated value $\Gamma_{0}$ (called the residual sensitivity value) and use it as a baseline value for the sensitivity parameter $\Gamma$ in the downstream sensitivity analysis.

    \item Strategy 2 (\citealp{pimentel2022covariateadaptive}): Researchers adopt some working model (e.g., logistic model) to estimate the biased exposure assignment probabilities due to inexact matching for each matched set and use them as the baseline exposure distribution in the downstream sensitivity analysis.
\end{itemize}

Generalizing either one of the above two strategies to the continuous exposure case is challenging. For generalizing Strategy 1, researchers need to come up with an appropriate model for the biased exposure dose assignment due to inexact matching (assuming no unmeasured confounding):
\begin{equation}\label{eqn: biased dose assignment model}
    \text{$p_{i\pi_{i}}=\text{pr}({Z}_i={z}_{i\pi_{i}}|\mathcal{F}_{0},\mathcal{Z}_{i})=\frac{\prod_{j=1}^{n_i} f(z = z_{i\pi_{i}(j)} \mid x_{ij})}{\sum_{\widetilde{\pi}_{i}\in S_{n_{i}} }\prod_{j=1}^{n_i} f(z = z_{i\widetilde{\pi}_{i}(j)} \mid x_{ij})}$ for $i=1,\dots, I$,}
\end{equation}
where $Z_{i}=(Z_{i1},\dots, Z_{in_{i}})$ is the random exposure dose vector for matched set $i$, the ${z}_{i\pi_{i}}=({z}_{i\pi_{i}(1)},\dots, {z}_{i\pi_{i}(n_{i})})$ is a possible realization of the random exposure vector (corresponding to permutation $\pi_{i}$ of the observed exposure dose vector $(z_{i1}, \dots, z_{in_{i}})$), and $f(z\mid x_{ij})$ is the conditional density function of continuous exposure dose (conditional on the measured confounders $x_{ij}$ for subject $j$ in matched set $i$). In the binary exposure dose case, researchers directly consider a biased exposure dose model that is similar to the Rosenbaum sensitivity bounds in the binary exposure case (\citealp{biased_randomization}). However, in the continuous exposure case, an appropriate biased exposure dose assignment model for (\ref{eqn: biased dose assignment model}) should incorporate the exposure dosing information ${z}_{i\pi_{i}}$ in each matched set, which involves modeling the conditional density function $f(z\mid x_{ij})$ (assuming no unmeasured confounding). Therefore, without imposing additional strong modeling assumptions on $f(z\mid x_{ij})$, Strategy 1 is difficult to generalize to the continuous exposure case. A similar difficulty also holds for generalizing Strategy 2 because when directly estimating the biased exposure dose assignment probabilities (\ref{eqn: biased dose assignment model}) using measured confounding information, we need to accurately estimate the conditional density $f(z\mid x_{ij})$, which is a very challenging problem without imposing additional strong modeling assumptions on $f(z\mid x_{ij})$, unless there are only very few measured confounders.

\end{remark}

\end{document}